\documentclass[11pt]{amsart}
\usepackage{amsthm,amsmath,amstext,amssymb,extsizes,xcolor,scrextend}
\usepackage[all]{xy}
\usepackage[active]{srcltx}
\usepackage{tikz}
\usetikzlibrary{decorations.shapes}
\usetikzlibrary{decorations.text}
\usetikzlibrary{calc}
\usepackage{xcolor}

\usepackage[lmargin=80pt,rmargin=80pt,tmargin=80pt,bmargin=80pt]{geometry}

\usepackage{comment}

\usepackage{caption}
\usepackage{subcaption}

\usepackage{hyperref}
\usepackage[normalem]{ulem}
\usepackage{epsfig}

\usepackage{transparent}

\usepackage{amstext,amsmath,amssymb,amsfonts,bbm,slashed}

\usepackage{array}
\usepackage{palatino}

\usepackage{amsmath}
\usepackage{pst-all}
\usepackage{pstricks}
\usepackage{enumerate}
\usepackage{lipsum}
\usepackage{float}
\usepackage{multicol}
\usepackage{mathrsfs}
\usepackage{longtable}
\usepackage{pbox}

\theoremstyle{definition}
\newtheorem{definition}{Definition}

\theoremstyle{plain}
 \newtheorem{thm}{Theorem}
 \newtheorem{cor}{Corollary}
 
 \newtheorem{prop}{Proposition}
\theoremstyle{remark}
\newtheorem{remark}{Remark}

\usepackage{multicol}\usepackage{array}
\newcommand{\bea}{\begin{eqnarray}}
\newcommand{\eea}{\end{eqnarray}}

\newtheorem{claim}{Claim}
\newtheorem{lemma}{Lemma}

\renewcommand{\le}{\leqslant}
\renewcommand{\ge}{\geqslant}

\newcommand{\be}{\begin{equation}}
\newcommand{\ee}{\end{equation}}

\newcommand{\f}{\frac}

\newcommand{\tr}{{\rm tr}}

\def\uN{\mathrm{U}(N)}
\def\oD{\mathrm{O}(D)}
\def\Ne{\mathrm{N{}_e}}
\def\No{\mathrm{N{}_o}}

        \let\l=\lambda
\let\m=\mu    \let\n=\nu           \let\om=\omega

     \let\X=F

\newcommand{\cF}{\mathcal{F}}

\usepackage{hyperref}
\hypersetup{
	colorlinks=true,	
	linkcolor=blue,	
	citecolor=[rgb]{0.0, 0.5, 0.0},	
	filecolor=magenta,	
	urlcolor=blue	
}

\author[R\'emi Cocou Avohou]{R\'emi Cocou Avohou}
\author[Reiko Toriumi]{Reiko Toriumi}
\author[Matthias Vancraeynest]{Matthias Vancraeynest}

\address[1]{Okinawa Institute of Science and Technology Graduate University, 1919-1, Tancha, Onna, Kunigami District, Okinawa 904-0495, Japan, \& ICMPA-UNESCO Chair, 072BP50, Cotonou, \& Ecole Normale Superieure, B.P 72, Natitingou, Benin,\newline {\tt \href{mailto:remi.avohou@oist.jp}{remi.avohou@oist.jp}}}
\address[2]{Okinawa Institute of Science and Technology Graduate University, 1919-1, Tancha, Onna, Kunigami District, Okinawa 904-0495, Japan,\newline {\tt \href{mailto:reiko.toriumi@oist.jp}{reiko.toriumi@oist.jp}}}
\address[3]{School of Mathematics and Maxwell Institute for Mathematical Sciences,
University of Edinburgh, EH9 3FD, UK \newline {\tt \href{mailto:matthias.vancraeynest@ed.ac.uk}{matthias.vancraeynest@ed.ac.uk}}}

\empty

\begin{document}

\title{\bf Classification of higher grade $\ell$ graphs for $\uN^2 \times \oD$ multi-matrix models}

\begin{abstract}
The authors studied in [Ann. Inst. Henri Poincar\'e D 9, 367-433, (2022)], a complex multi-matrix model with $\mathrm{U}(N)^2 \times \mathrm{O}(D)$ symmetry, and whose double scaling limit where simultaneously the large-$N$ and large-$D$ limits were taken while keeping the ratio $N/\sqrt{D}=M$ finite and fixed. In this double scaling limit, the complete recursive characterization of the Feynman graphs of arbitrary genus for the leading order grade $\ell=0$ was achieved. In this current study, we classify the higher order graphs in $\ell$. More specifically, $\ell=1$ and $\ell=2$ with arbitrary genus, in addition to a specific class of two-particle-irreducible (2PI) graphs for higher $\ell \geqslant 3$ but with genus zero. 
Furthermore, we demonstrate that each 2PI graph with a single $\mathrm{O}(D)$-loop with an arbitrary $\ell$ corresponds to a reduced alternating knot diagram with $\ell$ crossings as listed in the Rolfsen knot table, or a resulting alternating knot diagram obtained after performing the Tait flyping moves. We generalize to 2PR by considering the connected sum and the Reidemeister move I.

\end{abstract}

\maketitle

\setcounter{footnote}{0}
\tableofcontents

\section*{Introduction} 
Tensor models are derived from the idea that the classical geometry of some manifold could emerge from the statistical sum of random geometries associated with triangulations of this background in attempts to generalize in higher dimensions matrix model results on 2D quantum gravity (QG) \cite{MR2231334, MR1465433}. They may also be related to a broader hypothesis that gravity is derived from more fundamental (quantum) objects and laws \cite{MR2443568, MR1182621}. 

One of the most compelling findings in this area comes from the special case of matrices. In fact, the Feynman integral of matrix models generates ribbon graphs arranged in a $1/N$ (or genus) expansion, allowing this statistical sum to be effectively controlled using only analytical methods \cite{MR0782334}. In the large $N$ limit, only planar graphs contribute to leading order. Invariant matrix models undergo a phase transition to a continuum theory of random infinitely refined surfaces when the coupling constant is tuned to some critical value because planar graphs form a summable family. Matrix models, which are related to conformal field theory on fixed geometries, provide a framework for the analytic study of two-dimensional random geometries coupled with conformal matter. Furthermore, because of the fact that their Feynman diagrams are dual to piecewise flat manifolds of dimension two (for matrices) and higher (for tensors), their free energy is the generating function of connected piecewise flat manifolds, which in the continuum limit is anticipated to define a functional integral for Euclidean quantum gravity \cite{MR2231334, MR1465433}.

The success of matrix models inspired their generalization to random tensor models in higher dimensions, as cited. In the case of tensor models, the crucial $1/N$ expansion tool that leads to understanding and control of the partition function was missing. A method for understanding the partition function of tensor models analytically was abandoned for a while, and as a result, computations in theories implementing a discrete version of QG in higher dimension rely heavily on numerics even today. In fact, tensor models have failed to provide an analytically controlled theory of random geometries. No progress could be made because no generalization of the $1/N$ expansion to tensors has been found for a long time. Tensor models' interest could have been significantly increased if they were given an appropriate notion of $1/N$-expansion. 
This occurred precisely following Gurau's identification of an authentic concept of large $N$-expansion for a particular category of random tensor models \cite{MR2802384, MR2909101}.
This class includes colored tensor models that are associated with triangulations of simplicial pseudo-manifolds in any dimension \cite{MR2942819, MR2738259, MR2793930}. The colored models' perturbative series supports a $1/N$ expansion indexed by Gurau degree, a positive integer that plays the role of the genus in higher dimensions but is not a topological invariant. Melonic, or leading order graphs, triangulate the $D$-dimensional sphere in any dimension and are a summable family \cite{MR2826235}.  Tensor models, like their two-dimensional counterparts, go through a phase transition to a theory of continuous infinitely refined random spaces when tuned to criticality. Colored random tensors provide the first analytically accessible theory of higher-dimensional random geometries. They are discovered to go through a phase transition into the so-called branched polymer phase \cite{MR2942819, MR2831765}. More findings address long-standing issues in statistical mechanics on random lattices and mathematical physics \cite{MR2826235, MR2968506}.

The leading order (Gurau degree zero) of a $1/N$ expansion of tensor models is indeed well understood; it is dominated by melonic graphs \cite{MR2831765} and topologically they represent a subclass of spheres. For the higher orders in $1/N$ expansion (i.e., Gurau degree bigger than zero), it is known that indeed, given a Gurau degree, there is a finite number of schemes
\cite{MR3549799,MR3336566},
and each scheme corresponds to a summable class of graphs. 
However, except for the lowest orders (small Gurau degree) 
\cite{Kaminski:2013maa,Bonzom:2017pqs, Bonzom:2019yik, Bonzom:2019moj}, 
we do not know much about the classification of graphs of arbitrary Gurau degree.

Given a summable class of Feynman graphs, we are interested in their continuum limit.
For the class of diagrams selected by the large-$N$, such a continuum limit is possible, and it leads to two well-defined and ubiquitous continuum probabilistic models: the continuous random tree (so-called branched polymers) 
\cite{aldous1991, aldous} 
and the Brownian sphere 
\cite{legall2013, miermont2013}.
Typical matrix models lead to the latter, while typical tensor models to the former 
\cite{melbp}, 
but special models can be built with exchanged continuum limits 
\cite{Das:1989fq,Bonzom:2015axa}.
Finding different limits from these two remains an open problem, of particular relevance to quantum gravity, as neither branched polymers nor the Brownian spheres approach a classical geometry at large distances in higher dimensions than two.
It would therefore be desirable to find other variations of large-$N$ limits selecting different classes of Feynman graphs that could lead to new continuum models, potentially yielding a different limits from the familiar two described above.

A variation of large-$N$ limit was introduced by Ferrari in 
\cite{MR4002672},
and further developed and generalized in 
\cite{MR3994575, Azeyanagi:2017mre}. 
The fundamental variables can be described either as $D$ different $N\times N$ matrices,
(therefore it is named a multi-matrix model), 
or as a tensor.
In such a multi-matrix model, there are two parameters to play with owing to the presence of $D$ and $N$ in comparison to the more conventional tensor model governed by only one parameter, Gurau degree coming from having only $N$. 
Now with two parameters, different limits can be taken.
In particular, the large-$N$ limit of such models leads to a genus expansion, but one can furthermore perform the large-$D$ limit, thus introducing a new expansion at each fixed genus.
In a typical model, sending both $N$ and $D$ to infinity leads back to the large-$N$ limit of order-$3$ tensors, with melonic dominance (vanishing Gurau degree graphs), and hence to a branched polymer model.
In \cite{MR4450018}, subclass of higher orders of such double expansion of a multi-matrix model has been studied, namely vanishing grade $\ell=0$ but all genera $g$.
Specifically, they considered a $\mathrm{U}(N)^2\times \oD$ multi-matrix model,
which carries a nice geometric interpretation of its Feynman graphs. 
The ribbon structure generated by the propagation of $\mathrm{U}(N)^2$ indices is dual to quadrangulation of orientable surfaces of arbitrary genus. 
Then, $\oD$ symmetry decorates these surfaces by 
cyclic patterns, called $\oD$-loops.  
One can articulate that the statistical properties of the genus are controlled by the large-$N$ limit, while the proliferation of $\oD$-loops -- captured by a second integer number \cite{MR4002672} that we call the grade -- is directly tied to the large-$D$ limit.  
Therefore, one can say that this model exhibits an interplay between topological (genus) and combinatorial ($\oD$-loops) aspects.
By assigning a combinatorial grading ($\ell$) and topological property ($g$), or vice versa, we can potentially identify classes of Feynman graphs that are summable within this model. Speaking of the vanishing genus case, and the fact that the graphs we are studying are $4$-regular, we show that the enumeration of such graphs is possible by making use of alternating knot diagrams.
In fact, given a connected planar graph $G$, its medial graph $M(G)$ contains a vertex for each edge of $G$, as well as an edge between two vertices for each face of $G$, with the corresponding edges occurring consecutively. The edges of $M(G)$ are drawn in the corresponding faces, and $M(G)$ is then treated as a planar graph. The medial graph or $4$-regular graph can be associated with an alternating link using checkerboard coloring, where each crossing in the link diagram corresponds to a vertex of the graph \cite{MR2079925, MR3086663}.

The paper is organized as follows.
In Section \ref{sec:model}, we briefly review the essential of the $\uN^2 \times \oD$ multi-matrix model on which we extend the analysis done in \cite{MR4450018} in this present work. In particular, our analysis extends to $\ell > 0$ classification whereas they focused on the classification of $\ell=0$ graphs in \cite{MR4450018}.
 In Section \ref{sec:recursion}, we characterize the Feynman graphs generated by the $\uN^2 \times \oD$ multi-matrix model. Using combinatorial moves (Subsection  \ref{sec:comstruc}), we analyze the properties of Feynman graphs for $\ell=1$, and $\ell=2$ (Subsection \ref{sec:propell12}).
Subsection \ref{sec:completeell12} gives a complete characterization of $\ell=1$ and $\ell=2$ graphs with arbitrary genus $g$ via recursion. We classify planar ($g=0$) $\ell=3$ graphs in Section \ref{sec:l3g0}. In Section \ref{sec:higherlg0}, we give some classification of subclass of graphs (the 2PI graphs) with $\ell >3$. We make concluding remarks with open questions in Section \ref{sec:concl}. 
Appendix \ref{sec:appiso} lists some isomorphic graphs which are used to identify distinct schemes.
The construction of the 2PR graphs for $\ell=2$ is given with explicit examples in Appendix \ref{sec:app2PR}. 
Finally, Appendix \ref {sec:appb} provides some alternating knots as well as the corresponding 4-regular planar graphs.

\section{The $\uN^2 \times \oD$ multi-matrix model}
\label{sec:model}

We consider an $\oD$-invariant complex matrix model in zero dimension with  $D$ complex matrices $X_\m$ of size $N\times N$ and we may write $(X_\mu)_{ab} = X_{\mu ab}$ where $1\leq \mu\leq D$ and $1\leq a,b\leq N$. 
Indeed, one can view such $X$ as order-3 tensor, therefore this model can be viewed as order-3 $\uN^2 \times \oD$-invariant tensor model.

\begin{definition}[The $\uN^2\times \oD$ multi-matrix model]
The $\uN^2\times \oD$ multi-matrix model is a model involving a vector of $D$ complex matrices of size $N\times N$, denoted by $(X_\mu)_{\mu=1,\dots,D} =(X_1, \dots,X_D)$ which required to be invariant under unitary actions on the left and on the right of each $X_\mu$,
\be \label{eq:transfLaws}
X_\mu \rightarrow X'_\mu = O_{\mu \mu'} U_{(L)} X_{\mu'} U_{(R)}^{\dagger}\, ,
\ee
where $O$ is an orthogonal matrix in $\oD$, while $U_{(L)}$ and $U_{(R)}$ are two independent unitary matrices in two distinct copies of the group $\uN$, which we call left and right, respectively.
\end{definition}

The free energy of this model is given by

\be \label{eq:free-en-def}
\cF(\l) = \log \int [dX] \, e^{-S[X,X^\dagger]} \,,
\ee
with a measure $[dX] = \prod_{\mu,a,b} d{\rm Re} (X_\mu)_{ab}\, d{\rm Im}(X_\mu)_{ab}$ and an action $S[X,X^\dagger]$ studied in \cite{MR4450018} is given by
\be \label{eq:actionND}
S[X,X^\dagger] = ND \Bigl(  \tr\bigl[ X^\dagger_\m X_\m \bigr] - \frac{\l}{2}\sqrt{D} \, \tr\bigl[ X^\dagger_\m X_\n X^\dagger_\m X_\n \bigr] \Bigr)
\,,
\ee
where the interaction term is known as the tetrahedral interaction and $\lambda$ is the corresponding coupling constant. The coupling constant $\lambda$ has been scaled in such a way that it is kept fixed as $N,D \rightarrow + \infty$. This scaling yields well-defined large $N$ and large $D$ expansions \cite{MR4002672}, as further detailed below. The perturbative expansion in  $\lambda$ of the free energy $\cF(\lambda)$ enables a graphical description in terms of connected Feynman graphs. These Feynman graphs can be represented in three equivalent ways:

\begin{itemize}

\item[i)] as $4$-regular directed orientable maps with the following properties: each vertex has two outgoing and two ingoing half-edges, and additionally, the two outgoing (resp. two ingoing) half-edges appear on opposite sides of the vertex 
(see Figure \ref{fig:vertex_propa}, left column). 
It should be noted that self-loops/tadpoles and multiple edges are allowed in these maps.

\item[ii)] as directed orientable stranded graphs that are $4$-regular, obtained from the above representation by replacing each edge with a triple of parallel strands, two external and one internal, as shown in the top right corner of Figure \ref{fig:vertex_propa}.  Since the two $\uN$ symmetry groups' indices are carried by the external strands, they can be distinguished. 
When an external strand is on an edge connecting two half-edges, it is referred to as being \emph{left} (resp.\ \emph{right}) depending on which side of the edge it is on.  
The internal strand corresponds to the $\oD$ symmetry group.
According to the tetrahedral interaction's structure (see Figure \ref{fig:vertex_propa}, right column), each vertex's contraction pattern for the three different types of strands is determined. It is designed in such a way that a strand of a specific type (left, right, or internal) is always connected to another strand of the same type.  The strands in a Feynman graph are closed into loops.  External strand loops correspond to the faces of the underlying $4$-regular map and form a ribbon graph; we call them \emph{L- or R-faces}, depending on whether they are made up of left or right strands. The loops formed from internal strands will be referred to as \emph{straight faces} or \emph{$\oD$-loops}.

\item[iii)] as four-colored graphs obtained in the same manner as tensor models \cite{MR3616422}.
\end{itemize}

Notice that the Feynman graphs correspond to graphs embedded on Riemann surfaces and are dual to discretized surfaces; thus, the $\oD$-loops can be thought of as loops drawn on these discretized surfaces.
From the standpoint of tensor models of order-3, the Feynman graphs can also be viewed as dual to discretizations of three-dimensional pseudo-manifolds, which are indeed non-orientable in general \cite{MR3616422, MR3511784}.

\begin{figure}[htb]
    \centering
    \includegraphics[scale=0.9]{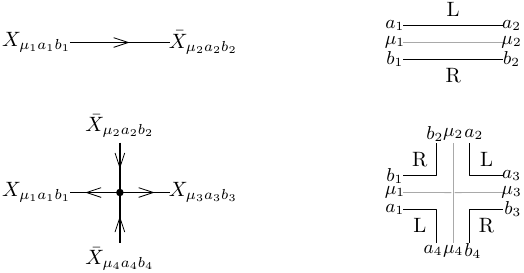}
    \caption{\small Propagator and vertex representations in the map (left) and stranded graph (right). The edge orientations are left implicit in the stranded representation for simplicity (note that fixing the orientation is equivalent to choosing R and L strands).}
    \label{fig:vertex_propa}
\end{figure}

The free energy in \cite{MR4002672} has double expansions in $1/N$ and $1/\sqrt{D}$,
\be \label{eq:free-en}
\cF(\l) = \sum_{g\in \mathbb{N}} N^{2-2g} \sum_{\ell\in \mathbb{N}} D^{1+g-\frac{\ell}{2}} \cF_{g,\ell}(\l)\,,
\ee
where $g\in\mathbb{N}$ is the genus of the corresponding $\uN^2$ ribbon graphs or, equivalently, the genus of the $4$-regular maps, which is also called the genus of the corresponding Feynman graphs. Through Euler's relation, it is defined by:
\be \label{eq:g}
2-2g = -e+v+f_L+f_R = -v +f \,,
\ee
where $e$ denotes the number of edges or propagators, $v$ the number of vertices ($e=2v$ because the maps are $4$-regular), $f_L$ (resp.\ $f_R$) denotes the number of L-faces (resp.\ R-faces) and $f=f_L+f_R$. Another parameter associated with Feynman graphs is the quantity $\ell\in\mathbb{N}$.
The parameter $\ell$ is given by:
\be \label{eq:ell}
\frac{\ell}{2} = 2 +  v - \frac{1}{2}f  - \varphi \,,
\ee
where $\varphi$ denotes the number of straight faces or $\oD$-loops.
The parameter $\ell$ introduces an extra grading in the standard genus expansion of matrix models, as shown in Eq.\ \eqref{eq:free-en}. As a result, we call it the \emph{grade}.\footnote{The parameter $\ell$ was called index in  
\cite{MR4002672}
whereas this name was used for a different quantity in 
\cite{MR3994575}.
To avoid confusion, in \cite{MR4450018}, a new name, grade, was introduced for it, and we follow here.}

Using Eq.\ \eqref{eq:g}, we obtain the following.
\begin{prop}
\label{prop:grade}
    The grade $\ell$ of Feynman graph can be expressed by the following relations:
    \begin{eqnarray}\label{eq:seven}
         \frac{\ell}{2} &=& 1+g+ \frac{1}{2}v-\varphi \,,\cr
         \f{\ell}{2} &=& g_L + g_R \,,
    \end{eqnarray}
    where $g_L$ (resp. $g_R$) is the genus of the ribbon graph obtained by deleting the $L$ (resp. $R$) strands from a Feynman graph in the stranded representation. Since $g_i \in\frac{\mathbb{N}}{2}$ for $i=L,R$ (the corresponding ribbon graphs are not necessarily orientable), $\ell\in\mathbb{N}$ follows. 
\end{prop}

We make an important remark here that indeed, our model defined by \eqref{eq:free-en} makes direct contact with the literature in tensor models in the following sense.
For $D=N$, we recover the large $N$ structure of  $\mathrm{U}(N)^2 \times \mathrm{O}(N)$ 
\cite{MR3192655}
and $\mathrm{O}(N)^3$ 
\cite{MR3555413}
tensor models
\be \label{eq:free-enND}
\cF(\l) = \sum_{\om \in \frac{\mathbb{N}}{2}} N^{3-\om}  \cF_{\om}(\l)\,,
\ee
where $\omega$ is an important combinatorial quantity called \emph{Gurau degree} \cite{MR3192655, MR3555413}.
\be \label{eq:omega}
\omega = g+\frac{\ell}{2} \,,
\ee
which is also called the \emph{Gurau degree}{\footnote{ It may be called \emph{index} in the more general context of \cite{MR3994575}.} \cite{MR2802384, MR2909101}
appearing in tensor model literature, which governs the large $N$ expansion of tensor models.

We  reorganize \eqref{eq:free-en} as
\be \label{eq:free-en-2}
\cF(\l) = \sum_{g\in \mathbb{N}} \left(\f{N}{\sqrt{D}}\right)^{2-2g} \sum_{\ell \in \mathbb{N}} D^{2-\frac{\ell}{2}} \cF_{g,\ell}(\l)\, .
\ee
We keep fixed $M \equiv \f{N}{\sqrt{D}}$
as we take $N\to\infty$ and $D\to\infty$. Then, 
\be \label{eq:free-en-0}
\lim_{ \substack{N,D\to\infty \\ M<\infty} } \f{1}{D^2} \cF(\l) = \sum_{g\geq 0} M^{2-2g}  \cF_{g,0}(\l) \equiv \cF^{(0)}(M,\l)\, .
\ee
In other words, by allowing $D\neq N$, but keeping the ratio constant, we obtain a \emph{double-scaling limit} that selects Feynman graphs with $\ell=0$, but of any genus. Since such graphs are much fewer than all the possible graphs, they lead to a summable series as worked out in \cite{MR4450018}.

In our present work, we will look at higher $\ell$ graphs corresponding to higher orders in this double scaling limit.

\medskip

\section{
Recursive characterization of Feynman graphs with $\ell=1,2$ with arbitrary $g$ 
}
\label{sec:recursion}

We perform a detailed combinatorial analysis of Feynman graphs with grades $\ell=1$ and $\ell=2$ for arbitrary $g$ in this section,
culminating in a complete characterization of 
Theorem \ref{propo:g1ell1},
Theorem \ref{propo:g1ell2},
and Theorem \ref{thm:induction}.

\subsection{Combinatorial structures of Feynman graphs}
\label{sec:comstruc}

We discuss some combinatorial structures on the $\uN^2 \times \oD$ multi-matrix model.
The majority of the structures listed below were first introduced in tensor model literature \cite{MR3549799, MR3336566}.

Remark that in this present work, we will work with unrooted Feynman graphs, whereas the preceding work \cite {MR4450018} exclusively worked with rooted Feynman graphs.

\medskip

Let us first introduce the \emph{cycle graph}, a specific Feynman graph that corresponds to an oriented edge; see Figure \ref{fig:cycle-graph} for illustration. It is characterized by $v=0$, $f=2$, and $\varphi=1$ by convention and has $g=\ell=\omega=0$ in particular.

\begin{figure}[htb]
    \centering
    \includegraphics[scale=.3]{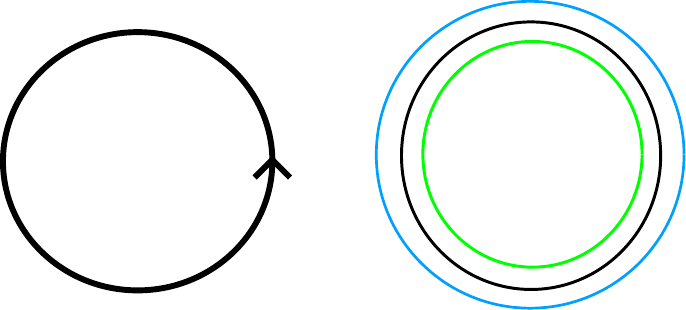}
    \caption{\small The cycle graph on the left and its stranded representation on the right. The black line on the right panel represents the $\mathrm{O}(D)$-loop or straight face, while the green and blue lines represent the L- and R-faces respectively.}
    \label{fig:cycle-graph}
\end{figure}

\medskip

\subsubsection{Melons and melon-free graphs}
\label{sec:MelonFree}

\begin{definition}[Melonic Feynman graph]
\label{def:melon}
\
\begin{itemize}
    \item[i)]  An  \emph{elementary melonic 2-point subgraph} is the subgraph represented on the right side of Figure \ref{fig:melon}, where the two external half-edges or legs are distinct.

    \item[ii)] A \emph{melonic insertion} on an edge $e$ of a Feynman graph is the replacement of $e$ by an elementary melonic $2$-point subgraph, respecting the orientation (see Figure \ref{fig:melon}). The reverse operation is called a \emph{melonic removal}. 

    \item[iii)] A  \emph{melonic} Feynman graph is a graph which can be reduced to the cycle-graph by successive removals of elementary melonic $2$-point subgraphs. 

\end{itemize}
\end{definition}

\begin{figure}[htb]
    \centering
    \includegraphics[scale=.7]{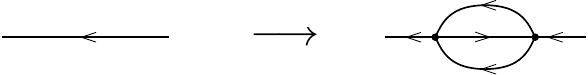}
    \caption{\small Insertion of a $2$-point elementary melonic subgraph.}
    \label{fig:melon}
\end{figure}

A Feynman graph is said to be melon-free if it lacks any elementary melonic $2$-point subgraph.  The 
cycle-graph is the only 
Feynman graph with $\omega=g=\ell=0$ that is not a melon.  Because of the following result 
\cite{MR3336566, MR3549799}, we can limit our investigation to melon-free 
Feynman graphs.

\begin{remark}
\label{rk:melon}
    From \cite{MR3336566, MR3549799} we can make the following remarks:

$\bullet$ The melonic Feynman graphs are exactly the Feynman graphs of Gurau degree $\omega=0$, that is of genus $g=0$ and of grade $\ell=0$.

$\bullet$ The genus, grade, and Gurau degree are all preserved when a melonic insertion or removal is performed.

$\bullet$ The closed equation 
\be \label{eq:melonEq}
T(\lambda)=1+\lambda^2 T(\lambda)^4 \,,
\ee
is obeyed by the generating function  $T(\lambda)$ of melonic Feynman graphs.
Remark that indeed $T(\lambda)$ is the generating function for Fuss-Catalan numbers.

$\bullet$ 
A sequence of successive melonic removal can produce any Feynman graph $G$ from a melon-free Feynman graph, known as the core of $G$. The cycle-graph, in particular, is the core of any melonic Feynman graph.
\end{remark}

\medskip

\subsubsection{Schemes}
\label{sec:schemes}

\begin{definition}[Dipole]

A \emph{dipole}  is defined as a two-edge subgraph  on two vertices with a face of length two. Dipoles are classified into three types (see  Figure \ref{fig:03dipoles}):
\begin{itemize}
\item 
An N-dipole, which contains a length-two straight face ($\oD$-loop);
\item 
an L-dipole, which contains a length-two L-face;
\item 
an R-dipole, which contains a length-two R-face.
\end{itemize}
\end{definition}

\begin{figure}
    \centering
\begin{minipage}[t]{0.95\textwidth}
\centering
\def\svgwidth{0.9\columnwidth}
\tiny{
\begingroup%
  \makeatletter%
  \providecommand\color[2][]{%
    \errmessage{(Inkscape) Color is used for the text in Inkscape, but the package 'color.sty' is not loaded}%
    \renewcommand\color[2][]{}%
  }%
  \providecommand\transparent[1]{%
    \errmessage{(Inkscape) Transparency is used (non-zero) for the text in Inkscape, but the package 'transparent.sty' is not loaded}%
    \renewcommand\transparent[1]{}%
  }%
  \providecommand\rotatebox[2]{#2}%
  \newcommand*\fsize{\dimexpr\f@size pt\relax}%
  \newcommand*\lineheight[1]{\fontsize{\fsize}{#1\fsize}\selectfont}%
  \ifx\svgwidth\undefined%
    \setlength{\unitlength}{531.24701072bp}%
    \ifx\svgscale\undefined%
      \relax%
    \else%
      \setlength{\unitlength}{\unitlength * \real{\svgscale}}%
    \fi%
  \else%
    \setlength{\unitlength}{\svgwidth}%
  \fi%
  \global\let\svgwidth\undefined%
  \global\let\svgscale\undefined%
  \makeatother%
  \begin{picture}(1,0.22966819)%
    \lineheight{1}%
    \setlength\tabcolsep{0pt}%
    \put(0,0){\includegraphics[width=\unitlength,page=1]{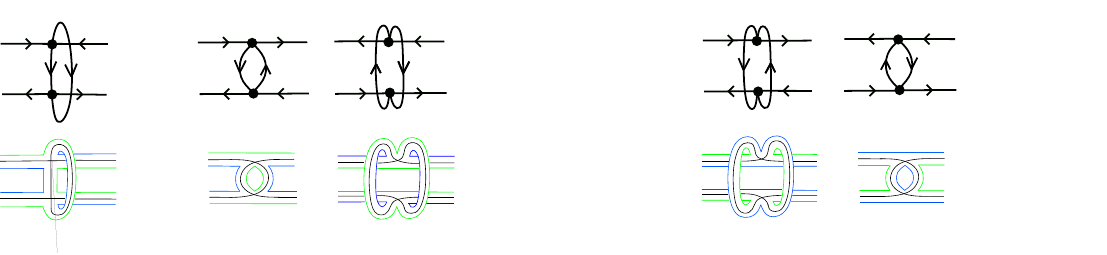}}%
    \put(0.01594667,0.00418275){\color[rgb]{0,0,0}\makebox(0,0)[lt]{\lineheight{1.25}\smash{\begin{tabular}[t]{l}\scalebox{1.5}{N-dipole}\end{tabular}}}}%
    \put(0.31499512,0.00418263){\color[rgb]{0,0,0}\makebox(0,0)[lt]{\lineheight{1.25}\smash{\begin{tabular}[t]{l}\scalebox{1.5}{L-dipole}\end{tabular}}}}%
    \put(0.76965607,0.00428436){\color[rgb]{0,0,0}\makebox(0,0)[lt]{\lineheight{1.25}\smash{\begin{tabular}[t]{l}\scalebox{1.5}{R-dipole}\end{tabular}}}}%
    \put(0,0){\includegraphics[width=\unitlength,page=2]{03dipolesnew.pdf}}%
    \put(0.40393167,0.16110666){\color[rgb]{0,0,0}\makebox(0,0)[lt]{\lineheight{1.25}\smash{\begin{tabular}[t]{l}\scalebox{2}{$\simeq$}\end{tabular}}}}%
    \put(0,0){\includegraphics[width=\unitlength,page=3]{03dipolesnew.pdf}}%
    \put(0.86824857,0.16466782){\color[rgb]{0,0,0}\makebox(0,0)[lt]{\lineheight{1.25}\smash{\begin{tabular}[t]{l}\scalebox{2}{$\simeq$}\end{tabular}}}}%
  \end{picture}%
\endgroup%

}
\end{minipage}
    \caption{\small 
    An N-dipole, two ways of representing an L-dipole, and two ways of representing an R-dipole
    are shown from left to right. 
    The upper row shows the so-called Feynman graph representation, while the lower row depicts the associated stranded structures. 
    An L-dipole (resp.\ R-dipole) is defined as having a face with a length-two L- (resp.\ R-) face. 
   Green (resp. blue) strands are L- (resp. R-) faces.
   }
    \label{fig:03dipoles}
\end{figure}

According to Figure \ref{fig:03dipoles}, which shows one pair on the left side of the dipole and the other pair on the right, the four half-edges/external legs incident to a dipole can be canonically divided into two pairs. The composition of dipoles into ladders, which we now define, will benefit from this distinction.

\begin{remark}
A priori, an $\oD$-loop of length $2$ can be of two forms given by the first two subgraphs in Figure \ref{dipoles}. The first one is the familiar N-dipole. However, one can observe that the second subgraph is in fact also a N-dipole by rotating a vertex by $180$ degrees as shown in the figure. 
This remark is simply due to the fact that our multi-orientable graphs are embedded graphs, i.e., combinatorial maps.
\begin{figure}[H]
\centering
\begin{minipage}[t]{0.8\textwidth}
\centering
\def\svgwidth{0.45\columnwidth}
\tiny{
\begingroup%
  \makeatletter%
  \providecommand\color[2][]{%
    \errmessage{(Inkscape) Color is used for the text in Inkscape, but the package 'color.sty' is not loaded}%
    \renewcommand\color[2][]{}%
  }%
  \providecommand\transparent[1]{%
    \errmessage{(Inkscape) Transparency is used (non-zero) for the text in Inkscape, but the package 'transparent.sty' is not loaded}%
    \renewcommand\transparent[1]{}%
  }%
  \providecommand\rotatebox[2]{#2}%
  \newcommand*\fsize{\dimexpr\f@size pt\relax}%
  \newcommand*\lineheight[1]{\fontsize{\fsize}{#1\fsize}\selectfont}%
  \ifx\svgwidth\undefined%
    \setlength{\unitlength}{233.37348517bp}%
    \ifx\svgscale\undefined%
      \relax%
    \else%
      \setlength{\unitlength}{\unitlength * \real{\svgscale}}%
    \fi%
  \else%
    \setlength{\unitlength}{\svgwidth}%
  \fi%
  \global\let\svgwidth\undefined%
  \global\let\svgscale\undefined%
  \makeatother%
  \begin{picture}(1,0.21788055)%
    \lineheight{1}%
    \setlength\tabcolsep{0pt}%
    \put(0,0){\includegraphics[width=\unitlength,page=1]{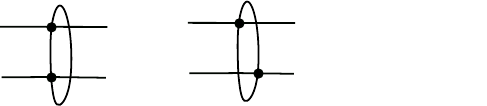}}%
    \put(0.61411189,0.10836178){\color[rgb]{0,0,0}\makebox(0,0)[lt]{\lineheight{1.25}\smash{\begin{tabular}[t]{l}\scalebox{2}{$=$}\end{tabular}}}}%
    \put(0,0){\includegraphics[width=\unitlength,page=2]{dipoles.pdf}}%
  \end{picture}%
\endgroup%

}
\end{minipage}
\caption{$\oD$-loops of length $2$ are contained in N-dipoles.}
\label{dipoles}
\end{figure}
\end{remark}

\begin{remark}
The arrows on the edges carry essential information regarding faces. 
However, as can be deduced from the expressions of $\ell$ in \eqref{eq:seven} which are symmetric under the swap of L- and R-faces (therefore under the overall flip of arrows of the edges), the underlying combinatorial map without arrows on edges determines the value of $\ell$ and also $g$.
We demonstrate in Figure \ref{fig:IO} this fact.
We will heavily use this fact in later proofs for constructing planar $g=0$ graphs.
\end{remark}
\begin{figure}[htb]
    \centering
\begin{minipage}[t]{0.95\textwidth}
\centering
\def\svgwidth{1\columnwidth}
\tiny{
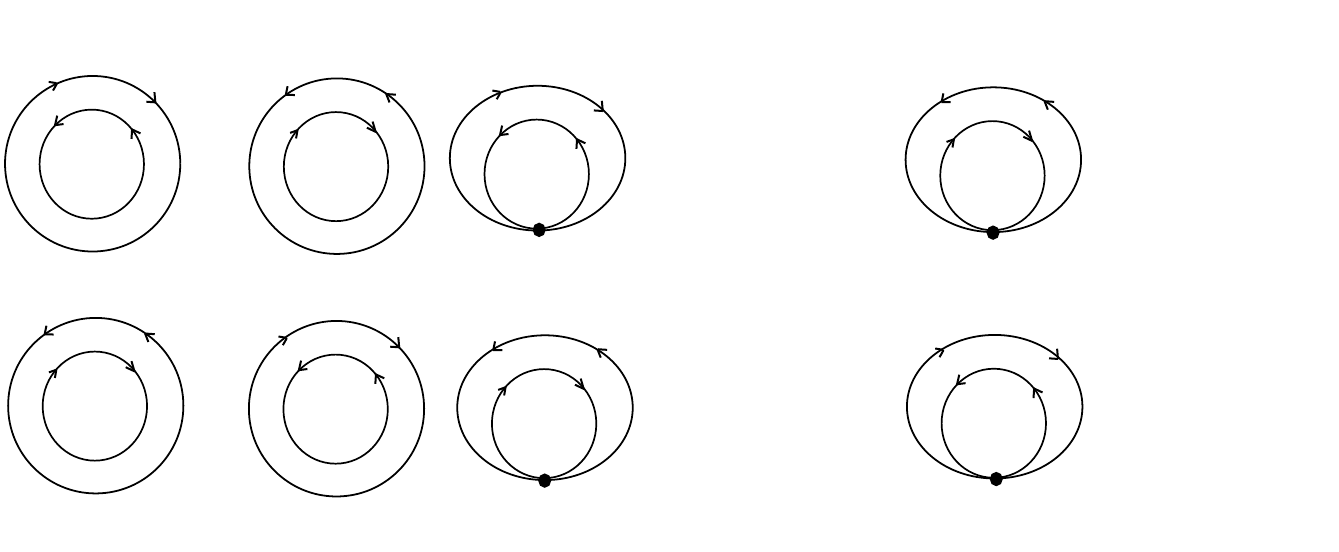
}
    \caption{\small 
Some examples of our multi-orientable tensor model graphs.
Notice that the arrows on edges are flipped on the graphs on the top row and the ones on the bottom row, but still give rise to the same values of $\ell$ and $g$.
}
    \label{fig:IO}
    \end{minipage}
\end{figure}

\begin{definition}[Ladder]
\

$\bullet$ A \emph{ladder} is a sequence of $n\geq2$ dipoles $(d_1, ..., d_n)$ in which two consecutive dipoles $d_i$ and $d_{i+1}$ are connected by two edges that involve two half-edges on the same side of $d_i$ and two half-edges on the same side of $d_{i+1}$. 

$\bullet$ 
If a ladder has rungs of different types, it is said to be \emph{broken}, or a B-ladder;
otherwise, it is \emph{unbroken}. As a result, there are three types of unbroken ladders: N-, L-, and R-ladders, which contain only rungs of type N, L, and R, respectively. 

$\bullet$ If a ladder is not included in a larger ladder, it is \emph{maximal}.
\end{definition}
The dipoles in a ladder are referred to as \emph{rungs} and the edges connecting them are divided into two \emph{rails}. 
Furthermore 
to keep track of the external face structure of ladders, we must divide N-ladders into two subfamilies: $\Ne$-ladders with an even number of rungs and $\No$-ladders with an odd number of rungs. This is illustrated in Figure \ref{fig:0chainvertexes}.

\begin{claim} 
Each ladder can be uniquely extended to a maximal ladder. Furthermore, any two distinct maximal ladders do not share any vertices.
\end{claim}

\medskip
The following replaces maximal ladders with new types of 4-point vertices that we refer to as \emph{ladder-vertices}\footnote{In the citation \cite{MR3336566}, they are referred to as chain-vertices.}. There are four different kinds of ladder-vertices; $\Ne$-, $\No$-, L-, R-, and B-vertices. Figure \ref{fig:0chainvertexes} serves as an illustration. 
We simply refer to them as N-vertices when we do not need to distinguish between $\Ne$- and $\No$-vertices.

\begin{figure}[h]
\centering
     \begin{minipage}[t]{0.8\textwidth}
      \centering
\def\svgwidth{0.7\columnwidth}
\tiny{
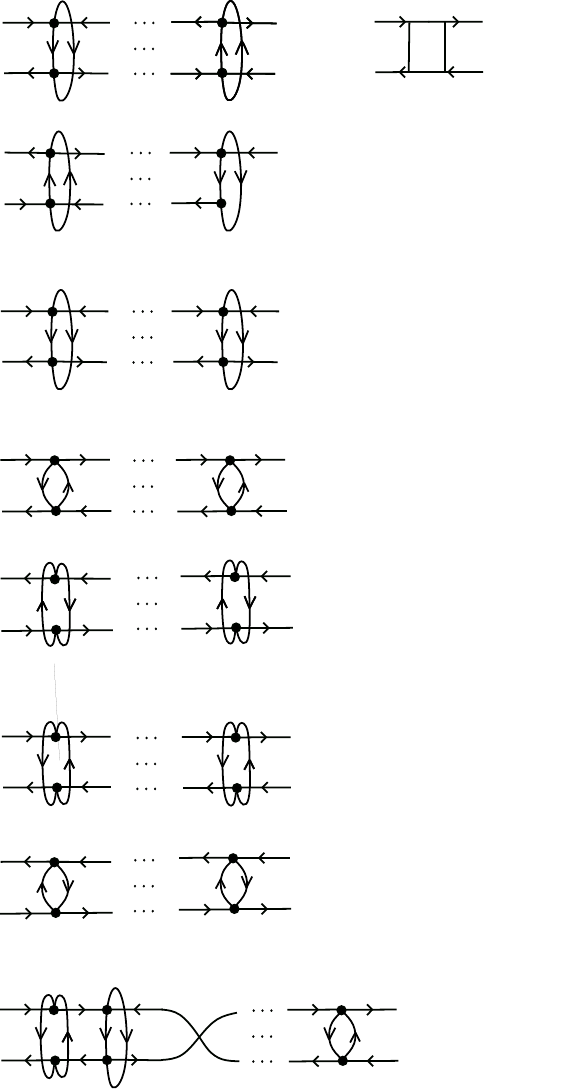
}
\caption{ {\small {Maximal ladders and their ladder-vertices.
We have represented the maximal ladders they represent (left panel) and the corresponding ladder-vertices (right panel).
By convention, the orientations of the external legs on the two sides of the B-ladder-vertex are fixed to be consistent as shown. In particular, this might entail `twisting' the two rails at one end of the corresponding B-ladder in the case of a B-vertex.
The B-ladders shown above are just examples.
}}}
\label{fig:0chainvertexes}
\end{minipage}
\end{figure}

\begin{definition}[Scheme]
Assume $G$ is a connected, melon-free Feynman graph. The \emph{scheme} $S_G$ of $G$ is the graph obtained by replacing any maximal ladder with the ladder-vertex of the corresponding type.
\end{definition}
We will heavily rely on \emph{schemes}, which describe equivalence classes of Feynman graphs defined up to melon and ladder insertions, in the remaining sections of the paper.

\begin{remark}
More broadly, we will consider the larger family of \emph{Feynman graphs with ladder-vertices}, which consists of all connected graphs constructed from edges, standard vertices, and ladder-vertices (see Figure \ref{fig:0vertexes}).  A scheme is, by definition, a Feynman graph with ladder-vertices. 
The converse, however, is not always true \cite{MR4450018}.
\end{remark}

\begin{figure}[htb]
    \centering
\begin{minipage}[t]{0.95\textwidth}
\centering
\def\svgwidth{0.8\columnwidth}
\tiny{
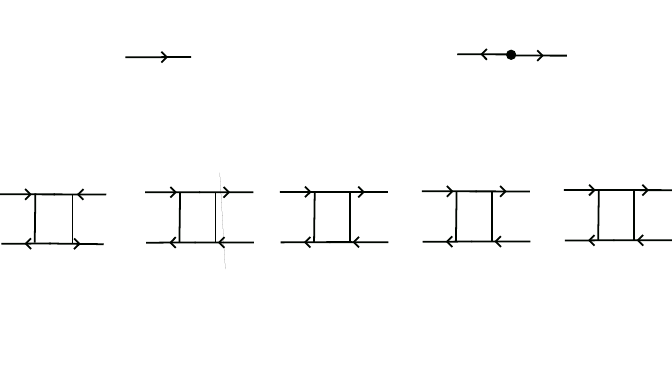
}
    \caption{\small 
    Schemes are made out of: edges, standard vertices and ladder-vertices.
We also show the stranded structure of the external faces of the ladder-vertices on the bottom row.
For the standard edge and the standard vertex, the stranded structure is shown in 
Figure \ref{fig:vertex_propa}.
    }
    \label{fig:0vertexes}
    \end{minipage}
\end{figure}

The following result is well known in \cite{MR3336566}.
\begin{prop}
\label{schemegraphcorresp}
   Let $G_1$ and $G_2$ be two Feynman graphs that are connected but do not contain any melons.
   If $S_{G_1} = S_{G_2}$, then:
\be
g(G_1) = g(G_2) \,, \qquad \omega(G_1) = \omega(G_2) \qquad \mathrm{and} \qquad \ell(G_1) = \ell(G_2)\,. 
\ee
\end{prop}

A Feynman graph with ladder vertices $G$ can itself be mapped to a unique scheme $S_G$ (obtained by a consistent replacement of melon two-point functions by propagators, and ladders by their corresponding ladder-vertices), which allows to define: $g(G):= g(S_G)$, $\omega(G):= \omega(S_G)$ and $\ell(G):= \ell(S_G)$.  In other words, the genus, Gurau degree, and grade of any equivalent class of graphs defined by a scheme are constant. 

Finally, distinguishing between two kinds of Feynman graphs with ladder-vertices will be useful. 
The family of connected Feynman graphs is included in the family of Feynman graphs with ladder-vertices whether they are \emph{two-particle reducible} (2PR) or \emph{two-particle irreducible} (2PI).
A Feynman graph with ladder-vertices $G$ is said to be 2PR if it contains an \emph{two-edge-cut}, which is a pair of edges in $G$ whose removal disconnects $G$. Otherwise, $G$ is 2PI. 

It is not hard to prove that a connected melon-free Feynman graph $G$,  $G$ is 2PI if and only if $S_G$ is 2PI.
In general, the 2PR/2PI property is transitive when a maximal ladder is replaced by a ladder-vertex in the class of melon-free Feynman graphs with ladder-vertices.

It is worth noting that we are working on Feynman graphs that allow half-edges to be attached to the vertices. As a result, we assume that the graphs examined throughout this work may have half-edges attached to them. Additionally, we may refer to Feynman graphs simply as graphs.

\medskip
\subsubsection{Combinatorial moves on  Feynman graphs}
\label{sec:CombMoves}

We now present a set of local operations on Feynman graphs with ladder-vertices and investigate how they affect the genus, grade, and Gurau degree of these graphs. 

A dipole or ladder-vertex \emph{contraction} is defined as the operation of 1) removing the dipole or ladder-vertex and 2) reconnecting the two half-edges on each side of the dipole or ladder-vertex. A dipole or ladder-vertex \emph{insertion} is the reverse operation. This is shown in Figure \ref{fig:0movesdipole}. 

\begin{figure}[htb]
\centering
\begin{minipage}[t]{0.8\textwidth}
\centering
\def\svgwidth{0.8\columnwidth}
\tiny{
\begingroup%
  \makeatletter%
  \providecommand\color[2][]{%
    \errmessage{(Inkscape) Color is used for the text in Inkscape, but the package 'color.sty' is not loaded}%
    \renewcommand\color[2][]{}%
  }%
  \providecommand\transparent[1]{%
    \errmessage{(Inkscape) Transparency is used (non-zero) for the text in Inkscape, but the package 'transparent.sty' is not loaded}%
    \renewcommand\transparent[1]{}%
  }%
  \providecommand\rotatebox[2]{#2}%
  \newcommand*\fsize{\dimexpr\f@size pt\relax}%
  \newcommand*\lineheight[1]{\fontsize{\fsize}{#1\fsize}\selectfont}%
  \ifx\svgwidth\undefined%
    \setlength{\unitlength}{354.13970226bp}%
    \ifx\svgscale\undefined%
      \relax%
    \else%
      \setlength{\unitlength}{\unitlength * \real{\svgscale}}%
    \fi%
  \else%
    \setlength{\unitlength}{\svgwidth}%
  \fi%
  \global\let\svgwidth\undefined%
  \global\let\svgscale\undefined%
  \makeatother%
  \begin{picture}(1,0.77895659)%
    \lineheight{1}%
    \setlength\tabcolsep{0pt}%
    \put(0,0){\includegraphics[width=\unitlength,page=1]{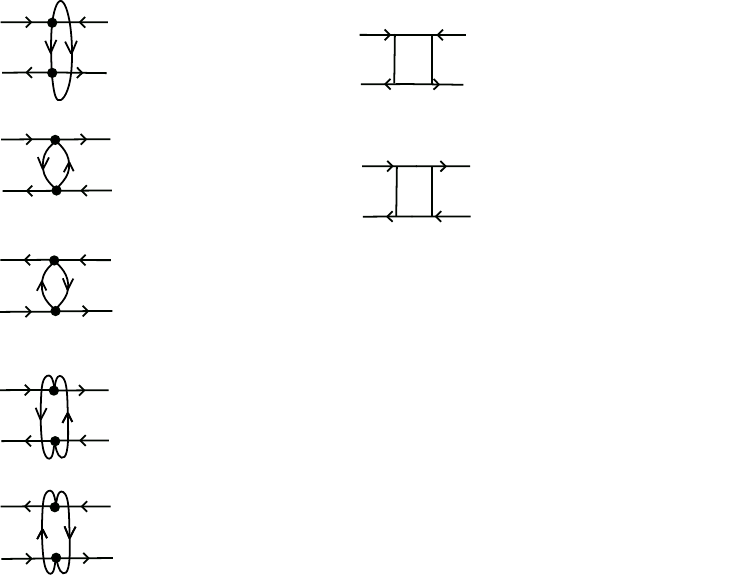}}%
    \put(0.54820885,0.68921724){\color[rgb]{0,0.01960784,0}\makebox(0,0)[lt]{\lineheight{1.25}\smash{\begin{tabular}[t]{l}\scalebox{1}{${\rm N}_{\rm o}$}\end{tabular}}}}%
    \put(0.55505201,0.5116766){\color[rgb]{0,0.01960784,0}\makebox(0,0)[lt]{\lineheight{1.25}\smash{\begin{tabular}[t]{l}\scalebox{1}{${\rm X}$}\end{tabular}}}}%
    \put(0,0){\includegraphics[width=\unitlength,page=2]{0movesdipolenew.pdf}}%
  \end{picture}%
\endgroup%

}
\caption{\small Dipole contraction/insertion (left panel) and ladder-vertices (right panel). $
{\rm X}
\in \{
{\rm L}, {\rm R},
{\rm N}_{\rm e}, {\rm B}\}$. 
}
\label{fig:0movesdipole}
\end{minipage}
\end{figure}

\begin{remark}
It should be noted that the graph may become disconnected if a dipole or a ladder-vertex contracts. Furthermore, the number of connected components can only increase by one. If the contraction of a dipole or a ladder-vertex increases the number of connected components by one, it is said to be \emph{separating}; otherwise, it is said to be \emph{non-separating}. 

In the topological sense, a separating (resp. non-separating) N-dipole is also separating (resp.\ non-separating) because the cycle it constitutes separates (resp.\ fails to separate) the discretized Riemann surface encoded in the Feynman graph into two disconnected regions.  
\end{remark} 
We shall work with connected graphs. 

Assume $G$ is a Feynman graph with ladder-vertices and contains a dipole. If this dipole is separating, we denote the two Feynman graphs with ladder-vertices obtained after contracting it as $G_1$ and $G_2$, respectively; if it is not separating, we denote the resulting Feynman graph with ladder-vertices as $G'$. In both cases, we clearly have $v(G_1)+v(G_2)=v(G)-2$ and $v(G')=v(G)-2$. To investigate the effect of the contraction on the genus, grade, and Gurau degree, one must examine how the total number of faces and $\oD$-loops is affected. 

\begin{prop}
\label{prop:dipolecases}
 Let $G$ be a Feynman graph with ladder-vertices which contains a dipole.   

\begin{enumerate}
    \item In case of a \textit{separating N-, L-, or R-dipole}, 
    we have the following equations:
    \begin{eqnarray}
  \label{eq:contSep}
    && f(G_1)+f(G_2)=f(G)\,, \quad  \varphi(G_1)+\varphi(G_2)=\varphi(G)\,, \cr
    && g(G_1)+g(G_2)=g(G)\,, \quad \ell(G_1)+\ell(G_2)=\ell(G) \quad \mathrm{and} \quad \omega(G_1)+\omega(G_2)=\omega(G).
    \end{eqnarray}   
    \item In case of a \textit{non-separating N-dipole} (see Figure \ref{connectingNdipole3types} for illustration), we have the following equations:
    \begin{eqnarray} 
    \label{eq:contNonSepN}
    &&\varphi(G')=\varphi(G)-1+\sigma, 
    \cr
    &&g(G')=g(G)-1\,, \quad \ell(G')=\ell(G)-2(\sigma+1) \quad \mathrm{and} \quad \omega(G')=\omega(G)-(\sigma+2),
    \end{eqnarray}
    with 
    $\sigma=-1$ (type I, connecting), $\sigma=0$ (type II, rearranging){\footnote{We should underline that the case $\sigma =0$ was overlooked in \cite{MR4450018} and \cite{MR3336566}, however, does not affect the $\ell=0$ result.}}, or $\sigma=1$ (type III).
    \item For non-separating L- or R-dipole (see an illustration in Figure \ref{LRconnecting}) then 
    \begin{eqnarray}
  \label{eq:contNonSepLR}
    && f(G')=f(G)-1+\sigma\,,  \cr
    &&g(G')=g(G)-\frac{1}{2}(\sigma+1)\,, 
    \ell(G')=\ell(G)-(\sigma+3)\,, 
    \mathrm{and} \quad \omega(G')=\omega(G)-(\sigma+2) \,,
\end{eqnarray} 
    with 
    $\sigma=-1$ (connecting), or $\sigma = 1$.
\end{enumerate} 
\end{prop}

\proof
Assume the dipole is a N-, L-, or R-dipole. In the case of a N-dipole, the contraction removes one internal $\oD$-loop and one internal face in the case of L- and  R-dipoles. Furthermore, the number of external faces is unaffected in the case of a N-dipole (the external face structure is the same before and after the contraction), and there is one additional external $\oD$-loop created by the dipole's separating nature. It is the number of external $\oD$-loops and R-faces (resp. L-faces) that remain unaffected in the case of a L-dipole (resp. R-dipole) while one additional external L-face (resp. R-face) is created. As a result, we have the first two relationships in \eqref{eq:contSep}. Using Eqs.\ \eqref{eq:g}, \eqref{eq:ell} from which we can write $\omega=3+\frac{2}{2}v-f+\varphi$, we thus obtain the remaining results in \eqref{eq:contSep}.

Let us now turn to the case of non-separating N-dipole. The contraction eliminates one internal $\oD$-loop while leaving the number of external faces unchanged. Furthermore, there is one more external $\oD$-loop that is either created, deleted, or unaffected. As a result, we obtain the relations in \eqref{eq:contNonSepN}.

Assume we have a non-separating L- or R-dipole. 
The contraction eliminates one internal L- or R-face while maintaining the same number of external $\oD$-loops and L- or R-faces. There is an additional external L- or R-face that is either created or deleted; thus, in both cases the results follow.
\qed

\begin{figure}[htb]
\begin{minipage}[t]{0.8\textwidth}
\centering
\def\svgwidth{0.9\columnwidth}
\tiny{
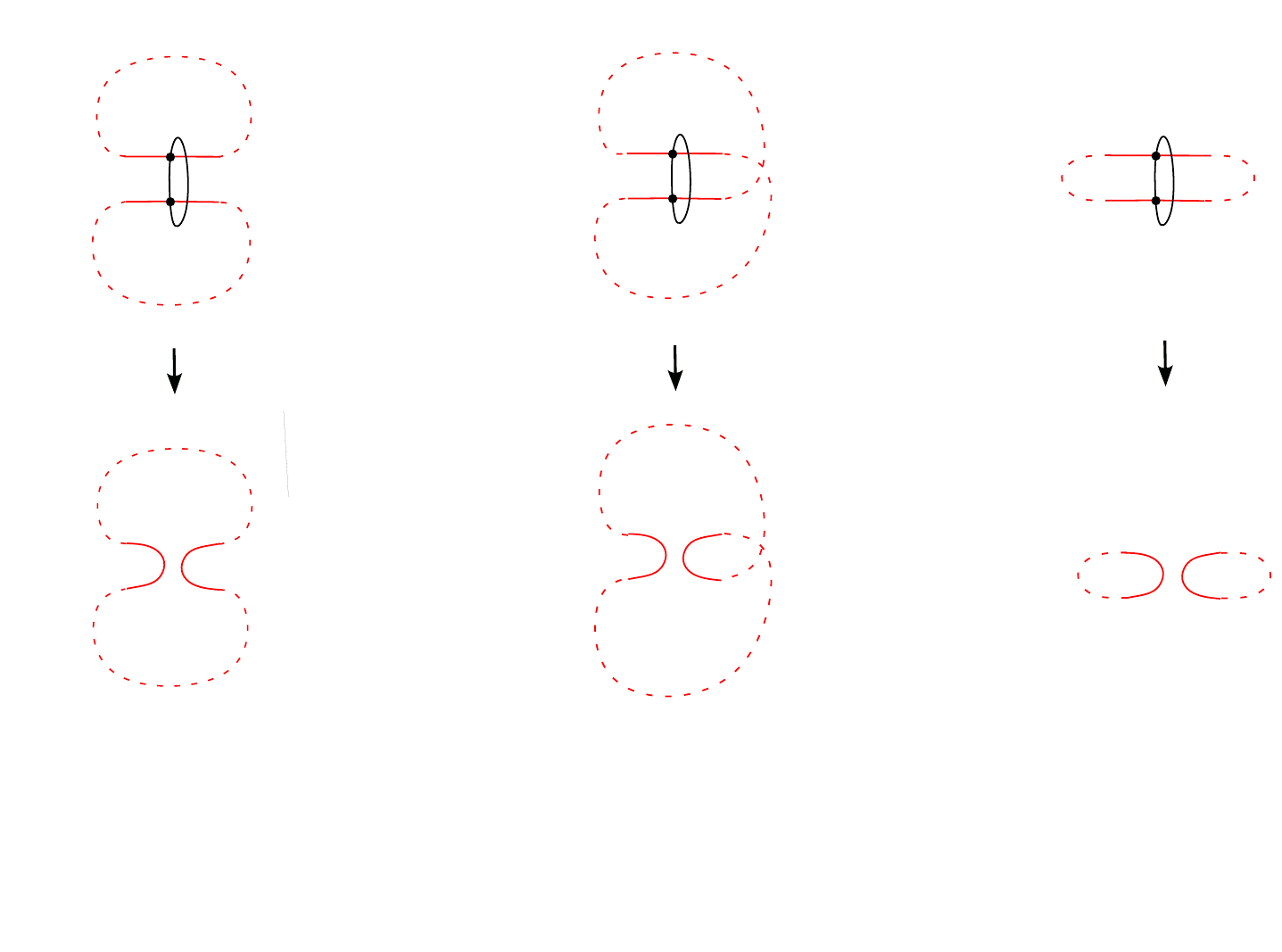
}
\caption{Three types of non-separating N-dipoles. The dotted lines represent the paths of $\oD$-faces. For $\ell=2$, the first two types $\sigma = -1$ (connecting) and $\sigma = 0$ (rearranging) are possible; however, for $\ell=0$ and $\ell=1$, only the first type ($\sigma = -1$) is possible. The third case $\sigma = 1$ is irrelevant for our present study, but will be relevant once one wants to consider $\ell \ge 4$ and $g \ge 1$. In the last row, we show some examples of each type of graph (left for type I, middle for type II, right for type III).}
\label{connectingNdipole3types}
\end{minipage}
\end{figure}

\begin{figure}[htb]
\centering
\begin{minipage}[t]{0.9\textwidth}
\centering
\def\svgwidth{0.9\columnwidth}
\tiny{
\begingroup%
  \makeatletter%
  \providecommand\color[2][]{%
    \errmessage{(Inkscape) Color is used for the text in Inkscape, but the package 'color.sty' is not loaded}%
    \renewcommand\color[2][]{}%
  }%
  \providecommand\transparent[1]{%
    \errmessage{(Inkscape) Transparency is used (non-zero) for the text in Inkscape, but the package 'transparent.sty' is not loaded}%
    \renewcommand\transparent[1]{}%
  }%
  \providecommand\rotatebox[2]{#2}%
  \newcommand*\fsize{\dimexpr\f@size pt\relax}%
  \newcommand*\lineheight[1]{\fontsize{\fsize}{#1\fsize}\selectfont}%
  \ifx\svgwidth\undefined%
    \setlength{\unitlength}{611.90521481bp}%
    \ifx\svgscale\undefined%
      \relax%
    \else%
      \setlength{\unitlength}{\unitlength * \real{\svgscale}}%
    \fi%
  \else%
    \setlength{\unitlength}{\svgwidth}%
  \fi%
  \global\let\svgwidth\undefined%
  \global\let\svgscale\undefined%
  \makeatother%
  \begin{picture}(1,0.73062539)%
    \lineheight{1}%
    \setlength\tabcolsep{0pt}%
    \put(0,0){\includegraphics[width=\unitlength,page=1]{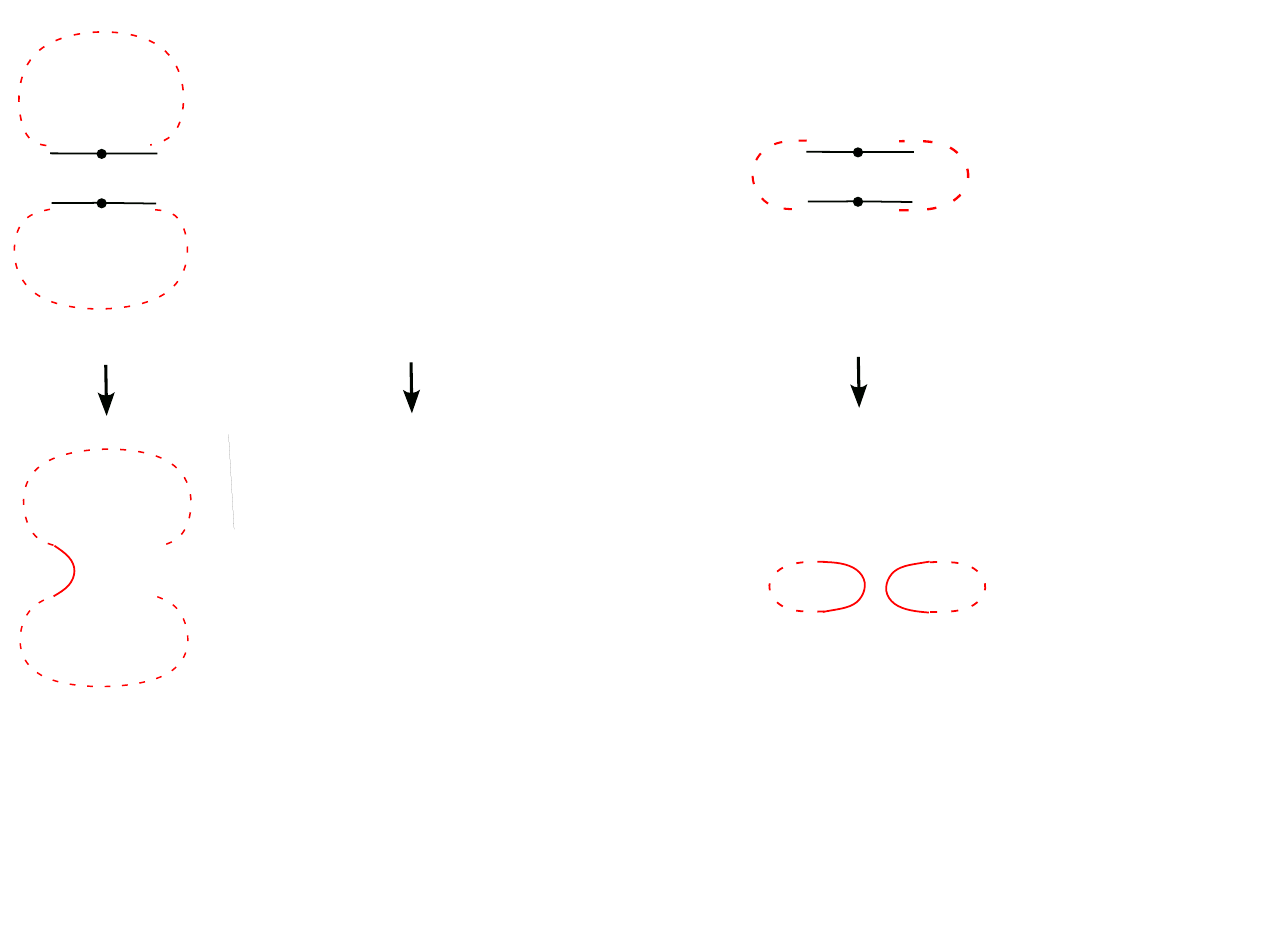}}%
    \put(0.10876683,0.42322732){\color[rgb]{0,0,0}\makebox(0,0)[lt]{\lineheight{1.25}\smash{\begin{tabular}[t]{l}\scalebox{1}{$\Delta \ell =-2, \Delta g = 0$}\end{tabular}}}}%
    \put(0.70202832,0.40308605){\color[rgb]{0,0,0}\makebox(0,0)[lt]{\lineheight{1.25}\smash{\begin{tabular}[t]{l}\scalebox{1}{$\Delta \ell =-4, \Delta g = -1$}\end{tabular}}}}%
    \put(0.05787021,0.71425601){\color[rgb]{0,0,0}\makebox(0,0)[lt]{\lineheight{1.25}\smash{\begin{tabular}[t]{l}\scalebox{1.5}{connecting ($\sigma = -1$)}\end{tabular}}}}%
    \put(0.70264008,0.70841673){\color[rgb]{0,0,0}\makebox(0,0)[lt]{\lineheight{1.25}\smash{\begin{tabular}[t]{l}\scalebox{1.5}{($\sigma = 1$)}\end{tabular}}}}%
    \put(0,0){\includegraphics[width=\unitlength,page=2]{LRconnecting.pdf}}%
  \end{picture}%
\endgroup%

}
\caption{
Two kinds of non-separating L- or R-dipoles; we call the first one with $\sigma = -1$ connecting.
The dotted lines represent the paths that L- or R-face, respectively, for L- or R-dipoles of non-separating type take. The first connecting type $\sigma = -1$ is possible for $\ell=2$, but not for $\ell=0$ and $\ell=1$.
On the bottom, we show some examples. The three graphs on the left has a connecting L- or R-dipole, and the one on the right has an L- or R-dipole of $\sigma = -1$ type.
}
\label{LRconnecting}
\end{minipage}
\end{figure}

The effect of contracting a ladder-vertex in a Feynman graph with ladder-vertices $G$ is now investigated. We continue to denote the two graphs (resp.\ the graph) obtained after contracting a separating (resp.\ non-separating) ladder-vertex in $G$ by $G_1$ and $G_2$ (resp. $G'$). As explained in Section \ref{sec:schemes}, one way to examine how the genus $g$, grade $\ell$, and Gurau degree $\omega$ change when a ladder-vertex is contracted is to first replace it with a ladder of dipoles of the corresponding type. In the case of a B-vertex, this ladder can be of arbitrary length and structure. The contraction of the ladder-vertex is then easily seen to be the combination of two moves: 1) the contraction of one dipole in the corresponding ladder, whose analysis has been given above; and 2) the removal of up to two melonic $2$-point subgraphs that may have been generated by the first move due to the presence of other dipoles in the initial ladder. The second step requires no further discussion because it preserves the genus, Gurau degree, and grade (see Section \ref{sec:MelonFree}). We obtained the following result, which is similar to the results in \cite{MR4450018}. 

\begin{prop}\label{prop:ladder-vertices}
     Let $G$ be a Feynman graph with ladder-vertices which contains a dipole. Let us replace the ladder-vertex with a B-, N-, L-, or R-ladder-vertex, as appropriate. 
     
\begin{enumerate}
    \item In case of a \textit{separating B-, N-, L-, or R-ladder-vertex}, we have the following equations:
    \be \label{eq:LVcontSep}
    g(G_1)+g(G_2)=g(G)\,, \quad \ell(G_1)+\ell(G_2)=\ell(G) \quad \mathrm{and} \quad \omega(G_1)+\omega(G_2)=\omega(G) \, ;
    \ee
    
    \item In case of a \textit{non-separating N-ladder-vertex},
    we have the following equations:
    \be \label{eq:LVcontNonSepN}
    g(G')=g(G)-1\,, \quad \ell(G')=\ell(G)-2(\sigma+1) \quad \mathrm{and} \quad \omega(G')=\omega(G)-(\sigma+2) \, ,
    \ee
    with 
    $\sigma=-1$ (type I, connecting), $\sigma=0$ (type II rearranging){\footnote{We should underline that the case $\sigma =0$ was overlooked in \cite{MR4450018} and \cite{MR3336566}, however, does not affect the $\ell=0$ result.}}, $\sigma=1$ (type III);

    \item Non-separating L-, or R-ladder-vertex, then 
    \be \label{eq:LVcontNonSepLR}
    g(G')=g(G)-\frac{1}{2}(\sigma+1)\,, \quad \ell(G')=\ell(G)-(\sigma+3) \quad \mathrm{and} \quad \omega(G')=\omega(G)-(\sigma+2) \, ,
    \ee
    with $\sigma=\pm1$;

      \item Non-separating B-ladder-vertex, then 
    \be \label{eq:LVcontNonSepB}
    g(G')=g(G)-1\,, \quad \ell(G')=\ell(G)-4 \quad \mathrm{and} \quad \omega(G')=\omega(G)-3 \, .
    \ee
\end{enumerate} 
\end{prop}

\proof

Assume that we have a separating B-, N-, L-, or R-vertex. We begin by replacing the ladder-vertex with a B-, N-, L-, or R-ladder, as appropriate. Because the ladder-vertex is separating, the dipoles in the corresponding ladder must also be separating. As a result, we can apply the result of  Eq.\ \eqref{eq:contSep} to the contraction of a separating dipole. Furthermore, because melonic removals have no effect on the genus, grade, or Gurau degree of a Feynman graph, we obtain the relationships in \eqref{eq:LVcontSep}.

In the case of a non-separating N-vertex, we replace the N-vertex with a N-ladder, 
made out of non-separating N-dipoles.
Using the result of Eq.\ \eqref{eq:contNonSepN} and the properties of melonic removal giving therefore the results in \eqref{eq:LVcontNonSepN}.

We now turn to the case of non-separating L- or R-vertex. The same reasoning as in the previous case leads and using Eq.\ \eqref{eq:contNonSepLR} we obtain the results in \eqref{eq:LVcontNonSepLR}

Let us now discuss the last case of non-separating B-vertex. We can substitute a B-vertex for a B-ladder of any length and structure, with all of its dipoles being non-separating. Furthermore, because of the structure of a B-ladder, one can check that the contraction of a non-separating N-dipole always yields one additional external $\oD$-loop (case of Eq.\ \eqref{eq:contNonSepN} with $\sigma=+1$); and the contraction of a non-separating L- or R-dipole leads to one additional external face (case of Eq.\ \eqref{eq:contNonSepLR} with $\sigma=+1$). As a result, in this final case the results in \eqref{eq:LVcontNonSepB} follow. 
\qed

\medskip

\begin{remark}
\label{remarkuniqueness}
We elaborate the uniqueness of inserting a connecting N and a rearranging N once we specify the cuts.

Consider an ${\mathrm{O}}(D)$ loop. We cut two edges and name the half edges A, B, C, and D as shown in Figure \ref{fig:uniqueconnectingN}.
Then, a priori, we have three possibilities to connect back these half edges via a ${\mathrm N}$-dipole or -ladder-vertex, namely, connecting (i) A and D, (ii) A and C, and (iii) A and B. 
Remark that automatically it means that we connect respectively
(i) B and C, (ii) B and D, and (iii) C and D.

The last choice (iii) will simply bring us back to the original case. 

Once we make the first choice (i), a non-separating N that we insert will be necessarily connecting. Furthermore, the orientations of the edges which were cut dictate whether we insert an odd or even number of N-dipoles.

Similarly, once we make the second choice (ii), a non-separating N that we insert will be necessarily rearranging. Again, the orientations of the edges which were cut dictate whether we insert an odd or even number of N-dipoles.

In conclusion, once we specify the cuts, there is a unique way of inserting a connecting N and a rearranging N.
\end{remark}

\begin{figure}[H]
\begin{minipage}[t]{0.9\textwidth}
\centering
\def\svgwidth{0.8\columnwidth}
\tiny{
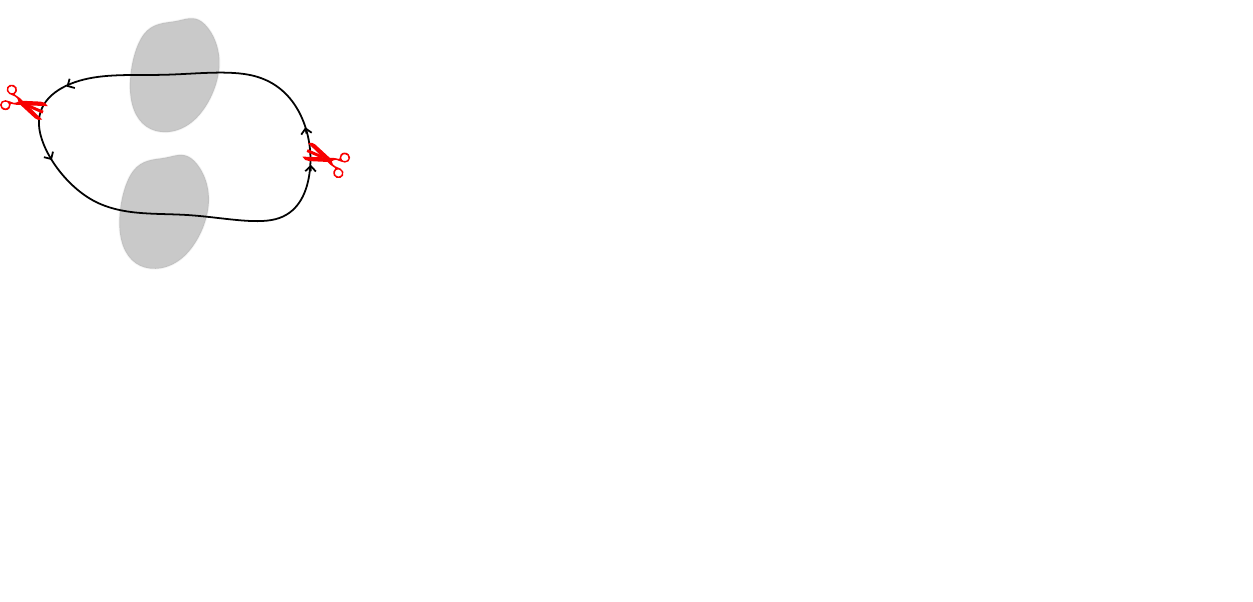
}
\caption{
Uniqueness of inserting a connecting N and a rearranging N, once the cuts are specified.
Grey shaded discs represent any subgraphs.
}
\label{fig:uniqueconnectingN}
\end{minipage}
\end{figure}

\medskip

In the same spirit as \cite{MR4450018}, we now review another local operation that will be useful in analyzing 2PR Feynman graphs. Assume $G$ is a 2PR Feynman graph with ladder-vertices. The graph $G$ thus contains a two-edge-cut $(e,e')$ and has the structure depicted on the left of Figure \ref{fig:flipOp}, where $\tilde G_1$ and $\tilde G_2$ are, by definition, two connected $2$-point non-empty subgraphs.

\begin{definition}
A \emph{flip} on $(e,e')$ is defined as the operation of cutting the two edges $e$ and $e'$ and reconnecting the four half-edges two by two, as shown on the right side of Figure \ref{fig:flipOp}. A flip necessarily divides $G$ into two connected components. The reverse operation of a flip is called a \emph{two-edge-connection insertion}.
\end{definition}

\begin{figure}[htb]
\begin{minipage}[t]{0.9\textwidth}
\centering
\def\svgwidth{0.7\columnwidth}
\tiny{
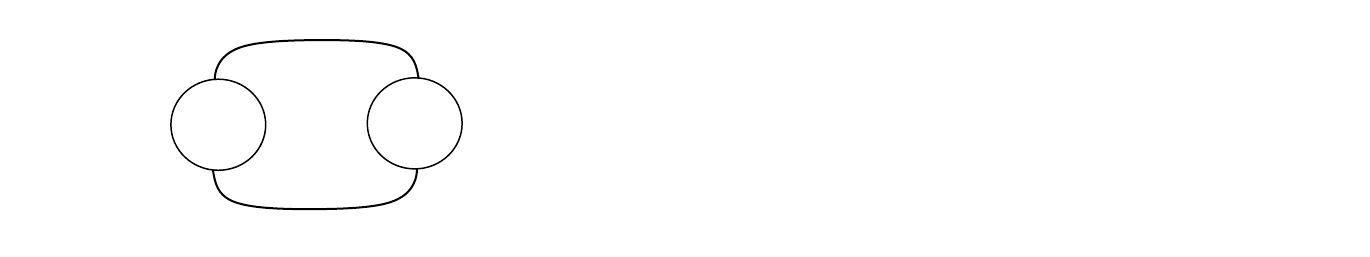
}
\end{minipage}
\caption{\small Flip operation on a ladder-vertex 2PR Feynman graph.
}
\label{fig:flipOp}
\end{figure}

The effect of a flip operation on the genus, grade, and Gurau degree of the graphs is summarized in the following result:
\begin{lemma}\label{lem:2PR} \cite{MR4450018} 
Let $G$ be a Feynman graph with ladder-vertices and suppose that it contains a two-edge-cut $(e,e')$. Then, the flip operation on $(e,e')$ generates two Feynman graphs with ladder-vertices $G_1$ and $G_2$ satisfying
\be \label{eq:flipOp}
g(G_1)+g(G_2)=g(G)\,, \quad \ell(G_1)+\ell(G_2)=\ell(G) \quad \mathrm{and} \quad \omega(G_1)+\omega(G_2)=\omega(G) \, .
\ee
\end{lemma}

\medskip

\subsection{Properties of melon-free Feynman graphs with $\ell=1$ 
or
$\ell=2$
}
\label{sec:propell12}

The connected, melon-free Feynman graphs with $\ell=1$ 
or $\ell=2$ and their corresponding schemes are what we are interested in characterizing. As made clear in Eq.\ \eqref{eq:seven}, they always have an odd number of (standard) vertices for $\ell=1$ and an even number of (standard) vertices for $\ell=2$, respectively. Namely, we observe for $\ell = 1$,
\be
\varphi = g + \frac{v+1}{2}
\,,
\ee
implying the number of standard vertices should be odd. While  for $\ell = 2$,
\be\label{eq:varphi}
\varphi = g + \frac{v}{2}
\,,
\ee
implying the number of standard vertices should be even. The following  lemma provides a more detailed characterization.

We can consider partitions of the total number of edges in terms of the lengths $l_i$ of each $\oD$-loop. We have
\begin{align}\label{partition}
    e = \sum_{i=1}^{\varphi}l_i.
\end{align}
We sometimes refer to a family of graphs by specifying their configuration in terms of $l_i$, written as a $\varphi$-tuple $(l_1,l_2,\ldots,l_{\varphi})$. For example, in section \ref{sec:l3g0}, we will study the family of $(4,4,\ldots,4)$ graphs.
We present some lemmas that will be useful several times throughout this work.

\begin{lemma}\label{2EL}
    Each graph with $l=1$ and $g=0$ has at least one $\oD$-loop with length $2$.
\end{lemma}
\begin{proof}
    We write a partition of the total number of edges $e$ as in \eqref{partition}. Suppose there are no $\oD$-loops with length $2$. Then because $l_i$ is even, the smallest possible $e$ is obtained when $l_i =4$ for every $i$. We can write
    \begin{align}
        e \geq 4\varphi = 4\Big( \frac{v+1}{2}\Big) = 2v + 2 > 2v.
    \end{align}
    Because the graphs we are considering are $4$-regular, we have that $e=2v$.
    This gives a contradiction, hence a $\oD$-loop with length $2$ should be present.
\end{proof}

In a similar manner, we can achieve the following lemma.
 
\begin{lemma}
\label{lem:generall1}
For each graph with $\ell \le 3$ with any $g \ge 1$, it always contains at least one $\oD$-loop of length $2$.
\end{lemma}
\begin{proof}
    Consider a graph with $\ell \le 3$ any $g \ge 1$ that does not contain an $\oD$-loop of length two.  Then because $l_i$ is even, the smallest possible $e$ is obtained when every $\oD$-loop has length $4$. From \eqref{eq:seven}, $\varphi = 1+ g + \frac{v}{2} - \frac{\ell}{2}$. We can write for $g \ge 1$, 
    \begin{align}
        e \geq 4\varphi  = 4(1+ g + \frac{v}{2} - \frac{\ell}{2}) = 4 + 4 g  + 2 v -2 \ell
        \ge 2 v + 2 (4- \ell)\,.
    \end{align}
    If $\ell \le 3$, then the above becomes, $e > 2 v$.
    Because the graphs we are considering are $4$-regular, we have that $e=2v$.
    This gives a contradiction, hence the statement is proven.
\end{proof}

\begin{cor}
\label{cor:N_dip}
Let $G$ be a connected, melon-free 
Feynman graph and $S_G$ its scheme of $\ell=1$ 
or
$\ell=2$. If $g \ge 1$, then there exists a N-dipole in $G$. In particular, there exists a N-dipole, a N-vertex or a B-vertex 
which has at least one N-dipole 
in $S_G$.
\end{cor}

\begin{lemma}
\label{lem:sep-conn}
Let $G$ be a connected, melon-free Feynman graph and $S_G$ its scheme. 
\begin{itemize}
\item{}
If 
$\ell (G) = 1$, any N-dipole in $G$ is separating or connecting; any other dipole is separating.  
Furthermore, in $S_G$: any N-dipole or N-vertex is separating or connecting; any other dipole or ladder-vertex is separating.

\item{}
If  $\ell(G) =2$, any N-dipole in $G$ is separating, connecting (type I), or rearranging (type II), and any L-, or R-dipole is either separating, or connecting. Furthermore, in $S_G$: any N-dipole or N-vertex is separating, connecting (type I) or rearranging (type II), and L-, or R-dipole or L-, or R-vertex is separating or connecting; any other ladder-vertex (namely B-vertex) is separating.
\end{itemize}
\end{lemma}
\begin{proof}
This follows from the variation of the grade under a dipole or a ladder-vertex contraction (see Section \ref{sec:CombMoves}) and from the non-negativity of the grade. Consider a N-dipole in $G$ and suppose it is non-separating. Performing a dipole contraction yields another connected 
Feynman graph $G'$ such that, by Eq.\ \eqref{eq:contNonSepN}, 
\begin{itemize}
\item{}
if $\ell(G)=1$ and $\ell\geq0$, it implies $\sigma=-1$. 
Hence, the N-dipole is connecting.

Consider now a L- or R-dipole in $G$ and suppose it is non-separating. By Eq.\ \eqref{eq:contNonSepLR}, contracting this dipole yields another connected Feynman graph $G'$ such that $\ell(G')\le \ell(G)-2$, which is impossible because $\ell(G) = 1$ and $\ell \geq 0$.

\item{}
Suppose $\ell(G) =2$ and consider the contraction of a non-separating N-dipole resulting in graph $G'$.  Equation\ \eqref{eq:contNonSepN} states that 
$\ell(G')=\ell(G)-2(\sigma+1)$ with $\sigma=\pm1$ or $\sigma= 0 $. 
Since $\ell(G)=2$ and $\ell\geq0$, it implies $\sigma=-1$ or $\sigma = 0$.  
Hence, 
the N-dipole is 
either
connecting (type I) or rearranging (type II).

Consider now L- or R-dipole in $G$ and suppose it is non-separating. By Eq.\ \eqref{eq:contNonSepLR}, contracting this dipole yields another connected Feynman graph $G'$ such that $\ell(G')=\ell(G)-(\sigma+3)$, with $\sigma = \pm1$. 
Since $\ell(G)=2$ and $\ell\geq0$, it implies $\sigma=-1$ and consequently the dipole is connecting.

\end{itemize}

The same reasoning applies for any dipole or ladder-vertex in $S_G$ using a dipole or a ladder-vertex contraction and Eqs.\ \eqref{eq:LVcontNonSepN}-\eqref{eq:LVcontNonSepB}.

\end{proof}

Unless otherwise stated, we assume that the Feynman graphs are always connected in the following. With the  Corollary \ref{cor:N_dip}, 
Lemma \ref{2EL}, 
Lemma  \ref{lem:generall1}, and Lemma \ref{lem:sep-conn}, 
we can clearly manipulate the connected, melon-free Feynman graphs with $\ell=1$ and their schemes by successive insertions or contractions of: 1) connecting N-dipoles or N-vertices; and 2) separating dipoles or ladder-vertices. 
In the same way, we can also manipulate the connected, melon-free Feynman graphs $\ell=2$ and their schemes by successive insertions or contractions of: 
1) connecting (type I) or rearranging (type II) N-dipoles or N-vertices ; 2) separating dipoles or ladder-vertices; and 3) connecting ($\sigma = -1$) L-, or R-dipoles, or L-, or R-vertices.

To lay the groundwork for this construction, it is useful to list the situations in which the contraction of a dipole or a ladder in a $\ell=1$ 
or $\ell = 2$ melon-free Feynman graph generates a melonic $2$-point subgraph. Assume that $G$ is such a graph and that it contains a dipole or ladder $X$.

For $\ell=1$, if we assume that $X$ is non-separating, it is necessarily a connecting N-dipole or N-ladder by Lemma \ref{lem:sep-conn}.

For $\ell=2$, if we assume that $X$ is non-separating, it is necessarily either 
\begin{itemize}
\item 
a connecting ($\sigma = -1$ type I) N-dipole or N-ladder, 
\item 
a rearranging ($\sigma = 0$ type II) N-dipole or N-ladder, or 
\item 
a connecting ($\sigma = -1$) L-, or R-dipole or L-, or R-ladder,
\end{itemize}
by Lemma \ref{lem:sep-conn}.

\medskip

\begin{figure}[H]
     \centering
     \begin{subfigure}[b]{.9\textwidth}
         \begin{minipage}[t]{0.9\textwidth}
\centering
\def\svgwidth{0.65\columnwidth}
\tiny{
\begingroup%
  \makeatletter%
  \providecommand\color[2][]{%
    \errmessage{(Inkscape) Color is used for the text in Inkscape, but the package 'color.sty' is not loaded}%
    \renewcommand\color[2][]{}%
  }%
  \providecommand\transparent[1]{%
    \errmessage{(Inkscape) Transparency is used (non-zero) for the text in Inkscape, but the package 'transparent.sty' is not loaded}%
    \renewcommand\transparent[1]{}%
  }%
  \providecommand\rotatebox[2]{#2}%
  \newcommand*\fsize{\dimexpr\f@size pt\relax}%
  \newcommand*\lineheight[1]{\fontsize{\fsize}{#1\fsize}\selectfont}%
  \ifx\svgwidth\undefined%
    \setlength{\unitlength}{412.70764449bp}%
    \ifx\svgscale\undefined%
      \relax%
    \else%
      \setlength{\unitlength}{\unitlength * \real{\svgscale}}%
    \fi%
  \else%
    \setlength{\unitlength}{\svgwidth}%
  \fi%
  \global\let\svgwidth\undefined%
  \global\let\svgscale\undefined%
  \makeatother%
  \begin{picture}(1,0.08566207)%
    \lineheight{1}%
    \setlength\tabcolsep{0pt}%
    \put(0,0){\includegraphics[width=\unitlength,page=1]{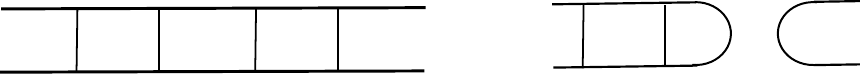}}%
    \put(0.11688382,0.02913336){\color[rgb]{0,0,0}\makebox(0,0)[lt]{\lineheight{1.25}\smash{\begin{tabular}[t]{l}\scalebox{1}{$Y$}\end{tabular}}}}%
    \put(0.7084932,0.03779098){\color[rgb]{0,0,0}\makebox(0,0)[lt]{\lineheight{1.25}\smash{\begin{tabular}[t]{l}\scalebox{1}{$Y$}\end{tabular}}}}%
    \put(0.33044051,0.02913336){\color[rgb]{0,0,0}\makebox(0,0)[lt]{\lineheight{1.25}\smash{\begin{tabular}[t]{l}\scalebox{1}{$X$}\end{tabular}}}}%
    \put(0,0){\includegraphics[width=\unitlength,page=2]{create-melon1.pdf}}%
  \end{picture}%
\endgroup%

}
\end{minipage}
         \caption{
         $X$ is a connecting N, rearranging N, connecting L, connecting R, or separating.
         For separating, $X \in  \{$N-dipole, $\No$, $\Ne$, L-dipole, R-dipole, L, R, B$\}$  in $\ell=1$ and $\ell=2$ graphs. For connecting,  $X \in  \{$N-dipole$, \,\No, \,\Ne\}$ in $\ell =1$ graphs, and $X \in  \{$N-dipole, \, L-dipole,  \, R-dipole \,$\No$, \,$\Ne$, \, L, \, R $\}$
         in $\ell =2$ graphs.
        }
         \label{fig:create-melon1}
     \end{subfigure}
     \begin{subfigure}[b]{0.9\textwidth}
     \begin{minipage}[t]{0.9\textwidth}
\centering
\def\svgwidth{0.8\columnwidth}
\tiny{
\begingroup%
  \makeatletter%
  \providecommand\color[2][]{%
    \errmessage{(Inkscape) Color is used for the text in Inkscape, but the package 'color.sty' is not loaded}%
    \renewcommand\color[2][]{}%
  }%
  \providecommand\transparent[1]{%
    \errmessage{(Inkscape) Transparency is used (non-zero) for the text in Inkscape, but the package 'transparent.sty' is not loaded}%
    \renewcommand\transparent[1]{}%
  }%
  \providecommand\rotatebox[2]{#2}%
  \newcommand*\fsize{\dimexpr\f@size pt\relax}%
  \newcommand*\lineheight[1]{\fontsize{\fsize}{#1\fsize}\selectfont}%
  \ifx\svgwidth\undefined%
    \setlength{\unitlength}{622.70832212bp}%
    \ifx\svgscale\undefined%
      \relax%
    \else%
      \setlength{\unitlength}{\unitlength * \real{\svgscale}}%
    \fi%
  \else%
    \setlength{\unitlength}{\svgwidth}%
  \fi%
  \global\let\svgwidth\undefined%
  \global\let\svgscale\undefined%
  \makeatother%
  \begin{picture}(1,0.35334642)%
    \lineheight{1}%
    \setlength\tabcolsep{0pt}%
    \put(0,0){\includegraphics[width=\unitlength,page=1]{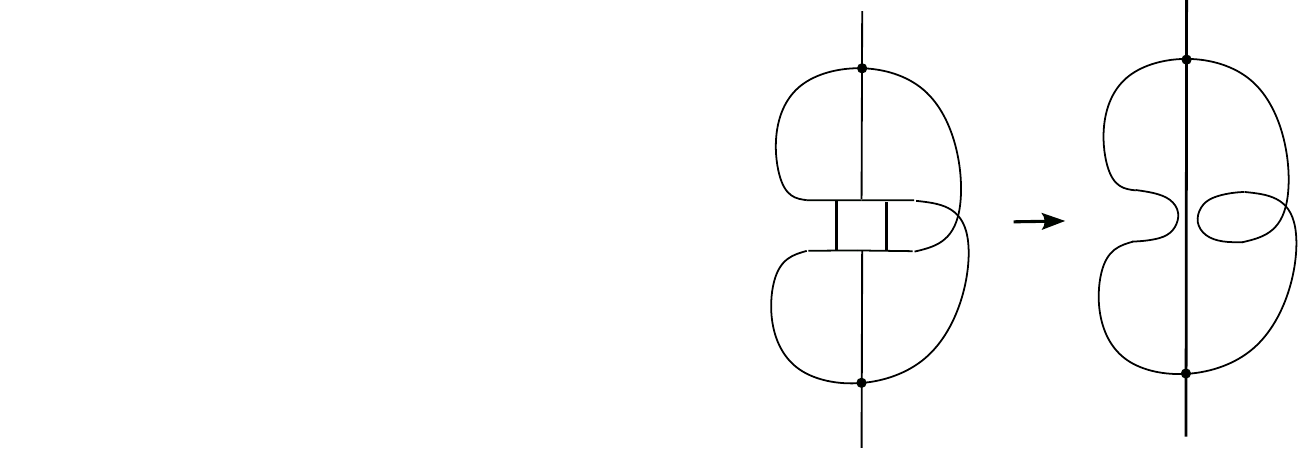}}%
    \put(0.65436266,0.17184922){\color[rgb]{0,0,0}\makebox(0,0)[lt]{\lineheight{1.25}\smash{\begin{tabular}[t]{l}\scalebox{1}{${\rm X}$}\end{tabular}}}}%
    \put(0,0){\includegraphics[width=\unitlength,page=2]{Createmelon4.pdf}}%
    \put(0.056576,0.16419873){\color[rgb]{0,0,0}\makebox(0,0)[lt]{\lineheight{1.25}\smash{\begin{tabular}[t]{l}\scalebox{1}{${\rm X}$}\end{tabular}}}}%
    \put(0,0){\includegraphics[width=\unitlength,page=3]{Createmelon4.pdf}}%
  \end{picture}%
\endgroup%

}
\end{minipage}
         \caption{On the left, $X$ is a connecting N, and $X \in  \{$N-dipole, \,$\No$$\}$ (note that  $X \ne \Ne$)
        in $\ell =1, \,2$ graphs. On the right, on the other hand, in a subgraph with a slightly different structure, $X$
        is a rearranging N, a connecting L, or a connecting R, i.e., 
        $X \in  \{\Ne,\,$ 
        L-dipole, R-dipole, L, R
        $\}$ possible only in $\ell =2$ graphs.
        Shown are only one orientation assignment on edges, and one needs to consider the other assignment as well. 
        }\label{fig:create-melon2}
     \end{subfigure}     
     \begin{subfigure}[b]{0.9\textwidth}
\begin{minipage}[t]{0.9\textwidth}
\centering
\def\svgwidth{0.8\columnwidth}
\tiny{
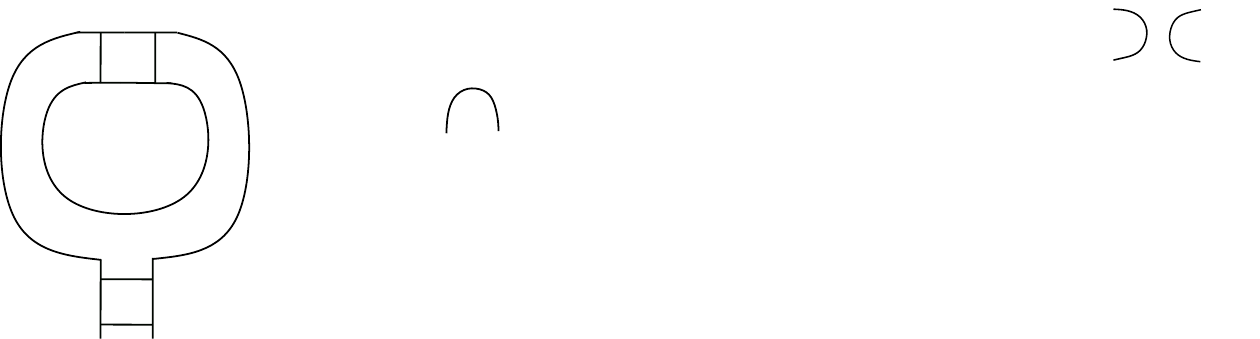
}
\end{minipage}
         \caption{On the left, $X$ is a connecting N, a connecting L, or a connecting R, and more precisely $X = \Ne$
        ($ X \notin \{$N-dipole$, \,\No\}$) in $\ell = 1$ graphs, and 
        $X \in  \{\Ne, \, $L-dipole, R-dipole, L, R$\}$ in $\ell =2$ graphs. On the right,
        on the other hand, in a subgraph with a slightly different structure, $X$
        is a rearranging N, and $X \in  \{ $N-dipole$, \,\No\}$  possible only in $\ell =2$ graphs.
        Shown are only one orientation assignment on edges, and one needs to consider the other assignment as well.        
        }\label{fig:create-melon3}
     \end{subfigure}     

     \caption{\small Configurations in which the contraction of a dipole or ladder $X$ in a 
     $\ell=1$ or $\ell=2$ melon-free Feynman graph, assumed to be connecting (i.e., any $\sigma = -1$) N, L, or R rearranging N,  or separating, generates a melonic $2$-point subgraph (where $Y$ is itself a dipole or a ladder).}
        \label{fig:generate-melon}
\end{figure}

\begin{remark}
When $X$ is contracted in a Feynman graph $G$ with $\ell=1$ or $\ell=2$, a new Feynman graph $G'$ with $\ell=1$ or $\ell=2$ is created. This new graph is \emph{not} melon-free in the following three situations:

\begin{enumerate}
    \item 
    As shown in Figure \ref{fig:create-melon1}, $X$ is connected on one side to another dipole or ladder $Y$ (of any type). For $\ell=1$, $X$ can be a N-dipole, a N$\mathrm{_o}$-, or a N$\mathrm{_e}$-ladder. For $\ell=2$, $X$ can be an N-dipole, an N$\mathrm{_o}$-, or an N$\mathrm{_e}$-ladder, or an L-, or R-dipole, or an  L-, or R-ladder;

    \item 
$X$ is inserted in between two edges of an elementary $2$-point melon, as illustrated in Figure \ref{fig:create-melon2}. Note that the choice of pair of edges on which $X$ is inserted, and therefore their orientation, is fixed by the requirement that either $X$ is connecting (as on the left of Figure \ref{fig:create-melon2}), or a rearranging N,  a connecting L-, or R-ladder (as on the right of Figure \ref{fig:create-melon2}). 
    As a result, $X$ must be an N-dipole or an N$\mathrm{_o}$-ladder (as on the left of Figure \ref{fig:create-melon2}), or else must be an N$\mathrm{_e}$-ladder, L-, or R-dipole, or L-, or R-ladder (as on the right of Figure \ref{fig:create-melon2})
    possible in $\ell=2$ graphs only;

    \item 
    $X$ is forming a $2$-point subgraph on one side of a dipole or ladder $Y$ (of arbitrary type), as illustrated in Figure \ref{fig:create-melon3}. 
    Note that the orientations of the edges are again fixed by the requirement that 
    $X$ is either a connecting N (as on the left of Figure \ref{fig:create-melon3}),  or it is a rearranging N (as on the right of Figure \ref{fig:create-melon3}).
    As a result, $X$ must be an N$\mathrm{_e}$-ladder in either $\ell=0\,,1\,,2$ graphs,  an L-, or R-dipole, or an   L-, or R-ladder (as on the left of Figure \ref{fig:create-melon3}), or must be an N-dipole or an N$\mathrm{_o}$-ladder possible only in $\ell=2$ graphs (as on the right of Figure \ref{fig:create-melon3}).
\end{enumerate}

If we assume instead that $X$ is separating, it can be of any type by Lemma \ref{lem:sep-conn}, and its contraction yields two Feynman graphs $G_1$ and $G_2$, such that $\ell(G_1) + \ell(G_2) = \ell(G)$. 
\end{remark}

We now investigate the structure of melon-free Feynman graphs of any genus with $\ell = 1$ and $\ell=2$.

The upcoming Theorem \ref{thm:ell1ell2planar} will provide all the planar $g=0$, $\ell = 1$  and $\ell=2$ melon-free Feynman graphs by specifying their schemes.

Later on, in Theorem \ref{propo:g1ell1}, we will provide all the $\ell = 1$ melon-free Feynman graphs of $g=1$ by specifying their schemes, and in Theorem \ref{propo:g1ell2}, all the 2PI $\ell=2$ melon-free Feynman graphs of $g=1$ by specifying their 2PI schemes.

Finally, Theorem \ref{thm:induction} offers inductive way of recursively constructing all the $\ell=1$ and $\ell=2$ Feynman graphs of any genus.

\medskip

To begin, let us first recall from \cite{MR4450018} that there is one  scheme of genus zero (given in Figure \ref{fig:Ell0g0.jpg}) and one scheme of genus one (given in Figure \ref{fig:Ell0g1.jpg} with $\ell=0$). 
We point out here an apparent discrepancy, by remarking that more schemes were reported in \cite{MR4450018} simply because the authors in \cite{MR4450018} worked in the class of rooted diagrams, whereas in our current work, we work with unrooted diagrams.

\begin{figure}[H]
\centering
\begin{minipage}[t]{0.9\textwidth}
\centering
\def\svgwidth{0.3\columnwidth}
\tiny{
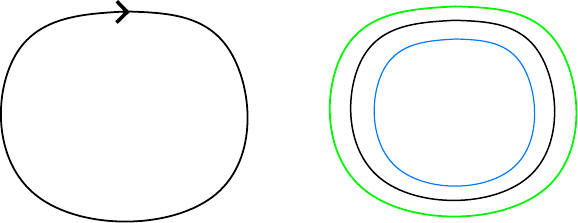
}
\caption{There is only one scheme of $\ell=0$ and $g=0$, the cycle graph.
The left is a colored graph representation, and the right is the stranded notation. Green is L-strand and in blue, we show R-strand.}

\label{fig:Ell0g0.jpg}
\end{minipage}
\end{figure}

\begin{figure}[H]
\centering
\begin{minipage}[t]{0.9\textwidth}
\centering
\def\svgwidth{0.15\columnwidth}
\tiny{
\begingroup%
  \makeatletter%
  \providecommand\color[2][]{%
    \errmessage{(Inkscape) Color is used for the text in Inkscape, but the package 'color.sty' is not loaded}%
    \renewcommand\color[2][]{}%
  }%
  \providecommand\transparent[1]{%
    \errmessage{(Inkscape) Transparency is used (non-zero) for the text in Inkscape, but the package 'transparent.sty' is not loaded}%
    \renewcommand\transparent[1]{}%
  }%
  \providecommand\rotatebox[2]{#2}%
  \newcommand*\fsize{\dimexpr\f@size pt\relax}%
  \newcommand*\lineheight[1]{\fontsize{\fsize}{#1\fsize}\selectfont}%
  \ifx\svgwidth\undefined%
    \setlength{\unitlength}{120.55869876bp}%
    \ifx\svgscale\undefined%
      \relax%
    \else%
      \setlength{\unitlength}{\unitlength * \real{\svgscale}}%
    \fi%
  \else%
    \setlength{\unitlength}{\svgwidth}%
  \fi%
  \global\let\svgwidth\undefined%
  \global\let\svgscale\undefined%
  \makeatother%
  \begin{picture}(1,0.91299295)%
    \lineheight{1}%
    \setlength\tabcolsep{0pt}%
    \put(0,0){\includegraphics[width=\unitlength,page=1]{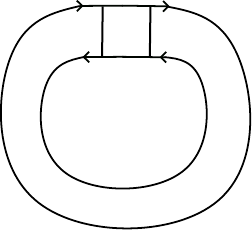}}%
    \put(0.43436027,0.75772678){\color[rgb]{0,0,0}\makebox(0,0)[lt]{\lineheight{1.25}\smash{\begin{tabular}[t]{l}\scalebox{1}{${\rm N}_{\rm e}$}\end{tabular}}}}%
  \end{picture}%
\endgroup%

}
\caption{
There is only one $\ell=0$, $g=1$ scheme, $S_1$, including the opposite orientation assignment on the edges. 
}
\label{fig:Ell0g1.jpg}
\end{minipage}
\end{figure}

By plugging in $\ell=2$ in equation \eqref{eq:seven}, we obtain the following relation
\begin{align}
    \varphi = g + \frac{v}{2}.
\end{align}
This relation implies that the number of vertices should be even.

\begin{lemma}
\label{l2}
For each graph with $\ell=2$, $g=0$, it either contains at least one $\oD$-loop of length $2$, otherwise every $\oD$-loop is of length $4$.
\end{lemma}
\begin{proof}
    Consider a planar graph with $\ell=2$ that does not contain an $\oD$-loop of length two and not every $\oD$-loop has length $4$.  Then because $l_i$ is even, the smallest possible $e$ is obtained when every $\oD$-loop but one has length $4$, and the remaining one has length $6$. We can write
    \begin{align}
        e \geq 4(\varphi-1) + 6 = 4(
        \frac{v}{2}-1) + 6 = 
        2v + 2 > 2v.
    \end{align}
    Because the graphs we are considering are $4$-regular, we have that $e=2v$.
    This gives a contradiction, hence the statement is proved.
\end{proof}

We now investigate the number of schemes based on the grade and the genus of the graph we are considering. Before doing so, the following definition gives a direct way to determine whether two schemes are isomorphic or not.

\begin{definition}[Graph isomorphism]
An isomorphism of graphs $G$ and $H$ is a bijection $f$ between the vertex sets of $G$ and $H$, $f: V(G)\rightarrow V(H)$
such that any two vertices $u$ and $v$ of $G$ are adjacent in $G$ if and only if $f(u)$ and $f(v)$
are adjacent in $H$. 
\end{definition}

\begin{definition}[Combinatorial map isomorphism \cite{MR1851080}]
Two combinatorial maps or graphs embedded on orientable surfaces, known as cyclic graphs, are considered isomorphic if there is a graph isomorphism between their underlying graphs that also preserves the cyclic order of edges emanating from each vertex.
\end{definition}

\begin{definition}[Isomorphic schemes]
\label{def:isomorphism}
Two schemes $S_1$ and $S_2$ are considered isomorphic if they are combinatorial map isomorphic in such a way that each ladder-vertex in $S_1$ corresponds to the same type of ladder-vertex in $S_2$. Additionally, the isomorphism must preserve the orientations of the edges (equivalently face structures).
Specifically, if a face $f_1$ in $S_1$ passes through certain vertices and edges, its counterpart $f_2$ in $S_2$ should traverse the images of those same vertices and edges in $S_2$.
\end{definition}
From this definition, it is clear that if two schemes have faces of the same length connected in identical ways, then if $f$ is an R-face (or L-face or $\oD$-face) in $S_1$, its corresponding face in $S_2$ must also be an R-face (or L-face or $\oD$-face), respectively. An illustration of isomorphic schemes are given by the last two graphs in Figure \ref{fig:App2}.

\medskip

\subsection{Complete characterization of $\ell=1$ and $\ell=2$ with arbitrary $g$}
\label{sec:completeell12}

This subsection recursively enumerates  the Feynman graphs with grades $\ell =1, 2$ and arbitrary genus. 
 The procedure begins with graphs with genus $g=0$,  giving the exact numbers of corresponding schemes  by classifying them as 2PI or 2PR graphs, for $\ell=1, 2$.
 For $g=1$, the same is achieved for $\ell=1$, but only for 2PI in the case of $\ell=2$.
These findings were immediately recursively extended to arbitrary genus for $\ell=1, 2$. We should note that the exact numbers of such schemes are not specified for $g >1$, but we provide an algorithm that allows us to fully recover all of them. Furthermore, we should emphasize that we work under the isomorphism condition given in Definition \ref{def:isomorphism}, namely, that two schemes are the same if they are isomorphic.

\begin{thm}
\label{thm:ell1ell2planar}
Exaustively, there are two schemes of $\ell=1$ and $g=0$ (Figure \ref{fig:Ell1g0.jpg}). They are 2PI and called the infinity graphs.
Furthermore, there are three 2PI schemes (Figure \ref{fig:Sum3}) and 27 2PR schemes of $\ell=2$ and $g=0$ (Figure \ref{fig:Sum4}).
\end{thm}

\begin{figure}[H]
\centering
\begin{minipage}[t]{0.9\textwidth}
\centering
\def\svgwidth{0.4\columnwidth}
\tiny{
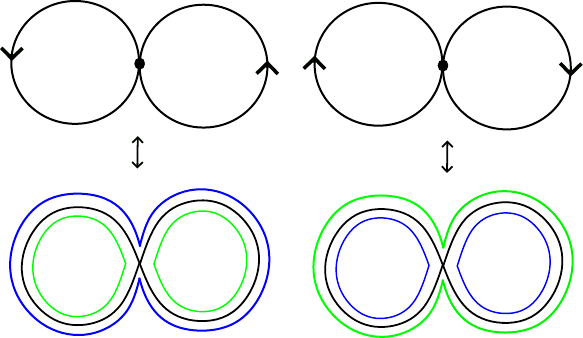
}
\caption{
All $\ell=1$, $g=0$ schemes.
There are two schemes of $\ell=1$ and $g=0$ in total.
We call this type the infinity graph.}
\label{fig:Ell1g0.jpg}
\end{minipage}
\end{figure}

\begin{figure}[H]
\centering
\begin{minipage}[t]{0.9\textwidth}
\centering
\def\svgwidth{0.7\columnwidth}
\tiny{
\begingroup%
  \makeatletter%
  \providecommand\color[2][]{%
    \errmessage{(Inkscape) Color is used for the text in Inkscape, but the package 'color.sty' is not loaded}%
    \renewcommand\color[2][]{}%
  }%
  \providecommand\transparent[1]{%
    \errmessage{(Inkscape) Transparency is used (non-zero) for the text in Inkscape, but the package 'transparent.sty' is not loaded}%
    \renewcommand\transparent[1]{}%
  }%
  \providecommand\rotatebox[2]{#2}%
  \newcommand*\fsize{\dimexpr\f@size pt\relax}%
  \newcommand*\lineheight[1]{\fontsize{\fsize}{#1\fsize}\selectfont}%
  \ifx\svgwidth\undefined%
    \setlength{\unitlength}{411.60459984bp}%
    \ifx\svgscale\undefined%
      \relax%
    \else%
      \setlength{\unitlength}{\unitlength * \real{\svgscale}}%
    \fi%
  \else%
    \setlength{\unitlength}{\svgwidth}%
  \fi%
  \global\let\svgwidth\undefined%
  \global\let\svgscale\undefined%
  \makeatother%
  \begin{picture}(1,0.2262219)%
    \lineheight{1}%
    \setlength\tabcolsep{0pt}%
    \put(0,0){\includegraphics[width=\unitlength,page=1]{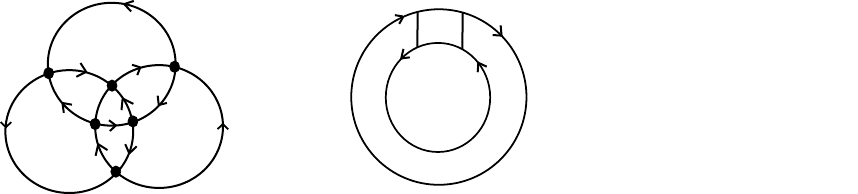}}%
    \put(0.50741566,0.18641892){\color[rgb]{0,0,0}\makebox(0,0)[lt]{\lineheight{1.25}\smash{\begin{tabular}[t]{l}\scalebox{1.2}{L}\end{tabular}}}}%
    \put(0,0){\includegraphics[width=\unitlength,page=2]{Sum3.pdf}}%
    \put(0.89154544,0.18569552){\color[rgb]{0,0,0}\makebox(0,0)[lt]{\lineheight{1.25}\smash{\begin{tabular}[t]{l}\scalebox{1.2}{R}\end{tabular}}}}%
    \put(0,0){\includegraphics[width=\unitlength,page=3]{Sum3.pdf}}%
  \end{picture}%
\endgroup%

}
\caption{
All $\ell=2$, $g=0$ 2PI schemes.
There are three 2PI schemes of $\ell=2$ and $g=0$ in total. Let us call the leftmost type $S_{1}^{2,0}$, the middle and the right type $S_2^{2,0}$.}
\label{fig:Sum3}
\end{minipage}
\end{figure}

\begin{figure}[H]
\centering
\begin{minipage}[t]{0.95\textwidth}
\centering
\def\svgwidth{0.9\columnwidth}
\tiny{
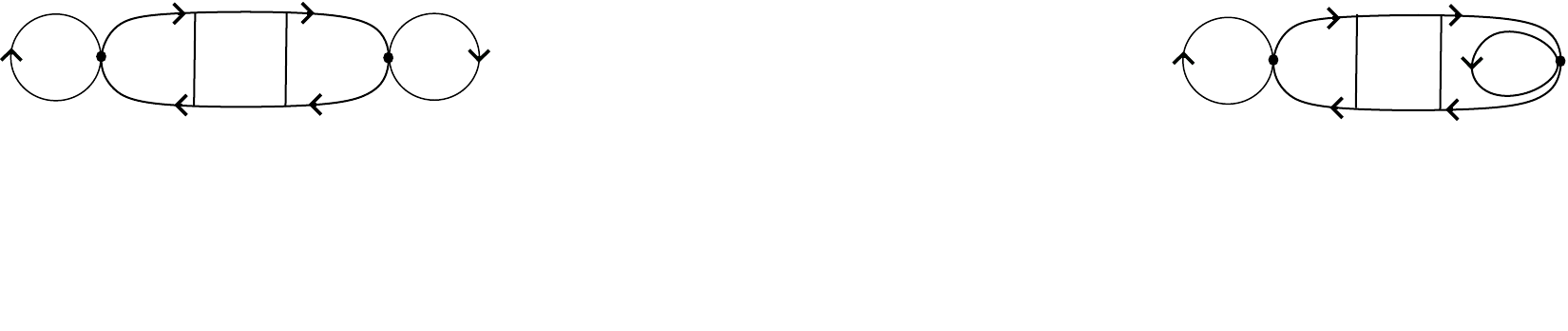
}
\caption{
All $\ell=2$, $g=0$ 2PR schemes.
There are 27 2PR schemes of $\ell=2$ and $g=0$ in total.
Let us call this type (all 27 of them) $S_3^{2,0}$.
$Y \in \{$L-dipole$, $R-dipole$, {\rm  L}, {\rm  R}, {\rm N}_{\rm e}, {\rm B}\}$ and $\widetilde {\rm N}_{\rm o} \in \{ $N-dipole$, {\rm N}_{\rm o}\}$.
}
\label{fig:Sum4}
\end{minipage}
\end{figure}

\begin{proof}  Let us ignore the orientations of the edges at the moment. 
We will recover them at the end of the proof.
Consider a  graph 
with  $\ell=1$, $g=0$. Because of 
Lemma \ref{2EL}
there is an $\oD$-loop of length $2$ present. If the graph has one vertex, there is only one possibility. It is the cycle graph with one tadpole added, also referred to as the infinity graph, depicted in Figure \ref{fig:infinity}.

\begin{figure}[H]
\centering
\begin{minipage}[t]{0.9\textwidth}
\centering
\def\svgwidth{0.15\columnwidth}
\tiny{
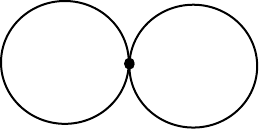
}
\caption{
The infinity graph.}
\label{fig:infinity}
\end{minipage}
\end{figure}

If it has more than one vertex and is connected, then the $\oD$-loop of length 2 necessarily has to go through two vertices and is contained in an N-dipole and the graph will be of the form as the graph on the left hand side of Figure \ref{2ELremoval}.
The planarity condition necessarily yields this form of the graph.
One can easily prove that\footnote{By planarity of the graphs in Figure \ref{2ELremoval}, $g=0$ before and after. $v$ will decrease by 2 whereas $\varphi$ will increase by 2, leaving $\ell$ invariant.} removing the $\oD$-loop of length 2 and its associated two vertices, as depicted in Figure \ref{2ELremoval},  does not change $\ell$ or $g$. Therefore, we obtain again a planar ($g=0$) graph with $\ell=1$. Consequently, the resulting graph contains an $\oD$-loop of length $2$. We can again remove from the resulting graph the $\oD$-loop of length 2 without changing $\ell$ or $g$. By continuing this procedure of removing an $\oD$-loop with length 2 and its associated two vertices while keeping $\ell$ and $g$ invariant, we end up with a planar graph with $\ell=1$ and only one vertex. 
\begin{figure}[H]
\begin{minipage}[t]{0.9\textwidth}
\centering
\def\svgwidth{0.4\columnwidth}
\tiny{
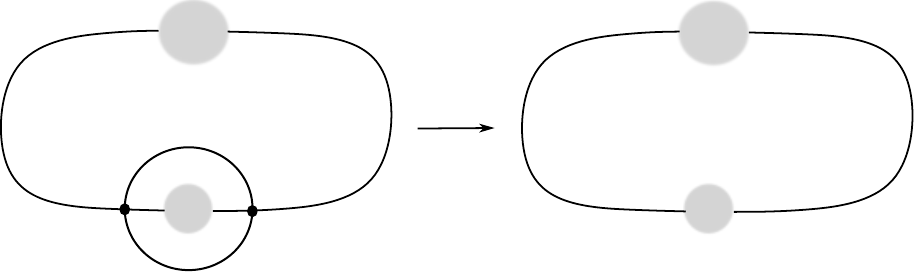
}
\end{minipage}
\caption{Removal of an $\oD$-loop of length $2$. Grey discs represent any subgraphs.}
\label{2ELremoval}
\end{figure}

By reversing this reasoning, one can conclude that a graph with $\ell=1$, $g=0$ can be constructed from the infinity graph by adding $\oD$-loops of length $2$.
In this case, these $\oD$-loops of length $2$ are necessarily part of a melonic two-point function, owing to the fact that only in one of the grey discs in Figure \ref{2ELremoval}, there exists a tadpole and the other grey disc only contains melons.
In conclusion, the resulting graph is the infinity graph decorated with melons.
Therefore, in schemes, we just proved that the infinity graph shown in Figure \ref{fig:infinity} is the only $\ell=1$, $g=0$ scheme{\footnote{In \cite{MR3336566}, they have proven also that the infinity graph is the only $\ell=1$ $g=0$ scheme.}}.
If we recover the edge orientation assignment in the infinity graph, we obtain two possibilities for the infinity graph which are depicted in Figure \ref{fig:Ell1g0.jpg}.

\medskip 

Consider a planar graph with $\ell=2$. If an $\oD$-loop of length two exists, we can repeat part of the reasoning from the case for $\ell=1$ and remove $\oD$-loops of length two iteratively while keeping $\ell$ and $g$ invariant. The resulting graph contains no $\oD$-loops of length 2 and thus should only contain  $\oD$-loops of length 4, as explained in Lemma \ref{l2}.
By reversing this argument, we can conclude that every graph is built from  $(4,4, \ldots, 4)$ graphs by sequentially adding  $\oD$-loops of length 2.

In order to discuss the possible configurations of $(4,4,\ldots,4)$ graphs, we first consider the possible configurations of a single  $\oD$-loop of length $4$ appearing in a planar graph.
Suppose that the $\oD$-loop of length 4 contains two vertices. Then, 
ignoring orientations of the edges for the moment,
we have only two possibilities, depicted in  Figure \ref{l2basic}. They can be constructed from the cycle graph by adding two tadpoles.
\begin{figure}[H]
\centering
\begin{minipage}[t]{0.9\textwidth}
\centering
\def\svgwidth{0.4\columnwidth}
\tiny{
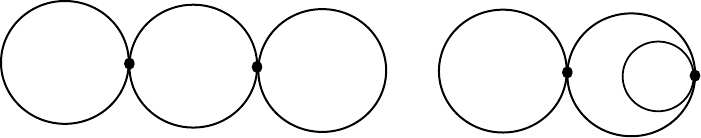
}
\end{minipage}
\caption{Simplest planar graphs of the configuration $(4,4, ... ,4)$ with $\ell=2$ with two vertices.}
\label{l2basic}
\end{figure}
Secondly, suppose that the $\oD$-loop contains three vertices. Then we can distinguish two cases given as the first two subgraphs (a) and (b) in Figure \ref{4EL}. They can be thought of as the cycle graph with one tadpole added. 
Finally, suppose that the $\oD$-loop contains four vertices. Then the only possible $\oD$-loop one can form is the subgraph drawn on the right (c) in Figure \ref{4EL}. Note that these $\oD$-loops are only subgraphs of a general graph, hence we let the vertices with their respective free half-edges drawn.

\begin{figure}[H]
\centering
\begin{minipage}[t]{0.9\textwidth}
\centering
\def\svgwidth{0.4\columnwidth}
\tiny{
\begingroup%
  \makeatletter%
  \providecommand\color[2][]{%
    \errmessage{(Inkscape) Color is used for the text in Inkscape, but the package 'color.sty' is not loaded}%
    \renewcommand\color[2][]{}%
  }%
  \providecommand\transparent[1]{%
    \errmessage{(Inkscape) Transparency is used (non-zero) for the text in Inkscape, but the package 'transparent.sty' is not loaded}%
    \renewcommand\transparent[1]{}%
  }%
  \providecommand\rotatebox[2]{#2}%
  \newcommand*\fsize{\dimexpr\f@size pt\relax}%
  \newcommand*\lineheight[1]{\fontsize{\fsize}{#1\fsize}\selectfont}%
  \ifx\svgwidth\undefined%
    \setlength{\unitlength}{373.42771515bp}%
    \ifx\svgscale\undefined%
      \relax%
    \else%
      \setlength{\unitlength}{\unitlength * \real{\svgscale}}%
    \fi%
  \else%
    \setlength{\unitlength}{\svgwidth}%
  \fi%
  \global\let\svgwidth\undefined%
  \global\let\svgscale\undefined%
  \makeatother%
  \begin{picture}(1,0.24903406)%
    \lineheight{1}%
    \setlength\tabcolsep{0pt}%
    \put(0,0){\includegraphics[width=\unitlength,page=1]{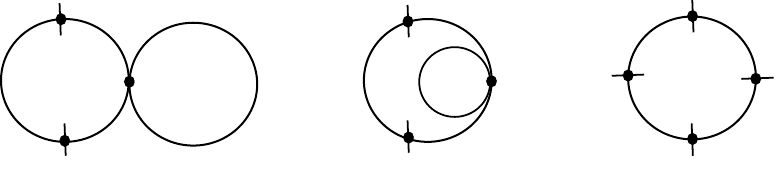}}%
    \put(0.13512584,0.01362642){\color[rgb]{0,0,0}\transparent{0.98699999}\makebox(0,0)[lt]{\lineheight{1.25}\smash{\begin{tabular}[t]{l}\scalebox{1}{$(a)$}\end{tabular}}}}%
    \put(0.52502006,0.00911609){\color[rgb]{0,0,0}\transparent{0.98699999}\makebox(0,0)[lt]{\lineheight{1.25}\smash{\begin{tabular}[t]{l}\scalebox{1}{$(b)$}\end{tabular}}}}%
    \put(0.8655686,0.01137149){\color[rgb]{0,0,0}\transparent{0.98699999}\makebox(0,0)[lt]{\lineheight{1.25}\smash{\begin{tabular}[t]{l}\scalebox{1}{$(c)$}\end{tabular}}}}%
  \end{picture}%
\endgroup%

}
\end{minipage}
\caption{Building blocks of $(4,4,\ldots,4)$ graphs.
}
\label{4EL}
\end{figure}
The $\oD$-loops of length $4$ just described and dipicted in Figure \ref{4EL} can be thought of as the building blocks from which one constructs a general $(4,4,\ldots,4)$ graph. 
We distinguish two cases below, by heavily relying on the planarity condition.

If the graph does not contain a tadpole, it must only contain $\oD$-loops of length $4$ as (c) of Figure \ref{4EL}. If an $\oD$-loop contains two non-adjacent vertices of another $\oD$-loop, there is only one possibility, in order not to break planarity.
The graph is depicted on the left side of Figure \ref{l2noTP}. 
The only other possibility when no tadpole is present, is the graph obtained by linking the $\oD$-loops together to form a closed chain. It can be easily redrawn with the use of an L- or R-ladder-vertex on the right side of Figure \ref{l2noTP}.
\begin{figure}[H]
\centering
\begin{minipage}[t]{1\textwidth}
\centering
\def\svgwidth{1\columnwidth}
\tiny{
\begingroup%
  \makeatletter%
  \providecommand\color[2][]{%
    \errmessage{(Inkscape) Color is used for the text in Inkscape, but the package 'color.sty' is not loaded}%
    \renewcommand\color[2][]{}%
  }%
  \providecommand\transparent[1]{%
    \errmessage{(Inkscape) Transparency is used (non-zero) for the text in Inkscape, but the package 'transparent.sty' is not loaded}%
    \renewcommand\transparent[1]{}%
  }%
  \providecommand\rotatebox[2]{#2}%
  \newcommand*\fsize{\dimexpr\f@size pt\relax}%
  \newcommand*\lineheight[1]{\fontsize{\fsize}{#1\fsize}\selectfont}%
  \ifx\svgwidth\undefined%
    \setlength{\unitlength}{985.28530812bp}%
    \ifx\svgscale\undefined%
      \relax%
    \else%
      \setlength{\unitlength}{\unitlength * \real{\svgscale}}%
    \fi%
  \else%
    \setlength{\unitlength}{\svgwidth}%
  \fi%
  \global\let\svgwidth\undefined%
  \global\let\svgscale\undefined%
  \makeatother%
  \begin{picture}(1,0.15601586)%
    \lineheight{1}%
    \setlength\tabcolsep{0pt}%
    \put(0,0){\includegraphics[width=\unitlength,page=1]{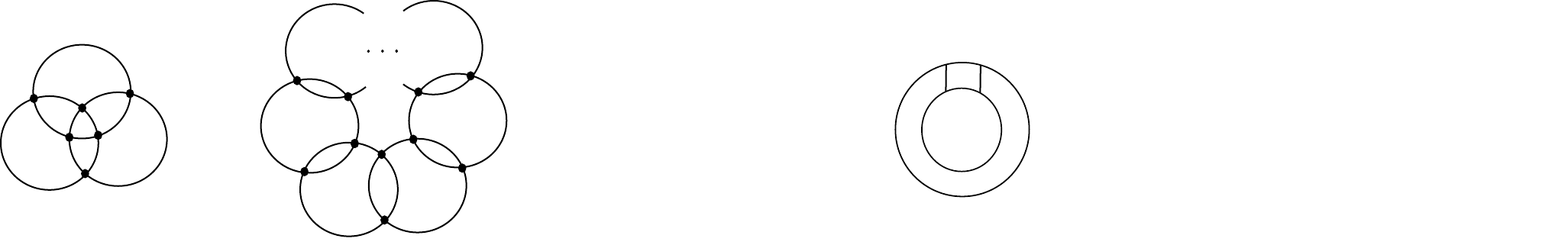}}%
    \put(0.60940622,0.10167334){\color[rgb]{0,0,0}\makebox(0,0)[lt]{\lineheight{1.25}\smash{\begin{tabular}[t]{l}\scalebox{1.2}{${\rm L}$}\end{tabular}}}}%
    \put(0,0){\includegraphics[width=\unitlength,page=2]{l2noTP.pdf}}%
    \put(0.9524601,0.0981506){\color[rgb]{0,0,0}\makebox(0,0)[lt]{\lineheight{1.25}\smash{\begin{tabular}[t]{l}\scalebox{1.2}{${\rm R}$}\end{tabular}}}}%
    \put(0,0){\includegraphics[width=\unitlength,page=3]{l2noTP.pdf}}%
    \put(0.88097908,0.07122554){\color[rgb]{0,0,0}\makebox(0,0)[lt]{\lineheight{1.25}\smash{\begin{tabular}[t]{l}\scalebox{2}{$\cong$}\end{tabular}}}}%
    \put(0,0){\includegraphics[width=\unitlength,page=4]{l2noTP.pdf}}%
    \put(0.53645312,0.07635411){\color[rgb]{0,0,0}\makebox(0,0)[lt]{\lineheight{1.25}\smash{\begin{tabular}[t]{l}\scalebox{2}{$\cong$}\end{tabular}}}}%
  \end{picture}%
\endgroup%

}
\end{minipage}
\caption{Graphs with $\ell=2$ and no tadpoles. 
The third and the fourth graphs and their corresponding schemes are obtained from the second graph from the left decorated with orientations on edges.}
\label{l2noTP}
\end{figure}

We now describe how any $(4,4,\ldots,4)$ graph with tadpoles can be constructed. 
Throughout the construction procedure below, we keep relying on the planarity condition.
The graph contains a tadpole, hence it contains an $\oD$-loop of either  (a) or (b).
We now have two vertices, each with two free half-edges.
One can add to this $\oD$-loop of (a) or (b), another $\oD$-loop of the same kind.
Then there are no free half-edges left and the procedure is finished as illustrated in Figure \ref{l2procedureshort}.

\begin{figure}[H]
\centering
\begin{minipage}[t]{0.7\textwidth}
\centering
\def\svgwidth{0.5\columnwidth}
\tiny{
\begingroup%
  \makeatletter%
  \providecommand\color[2][]{%
    \errmessage{(Inkscape) Color is used for the text in Inkscape, but the package 'color.sty' is not loaded}%
    \renewcommand\color[2][]{}%
  }%
  \providecommand\transparent[1]{%
    \errmessage{(Inkscape) Transparency is used (non-zero) for the text in Inkscape, but the package 'transparent.sty' is not loaded}%
    \renewcommand\transparent[1]{}%
  }%
  \providecommand\rotatebox[2]{#2}%
  \newcommand*\fsize{\dimexpr\f@size pt\relax}%
  \newcommand*\lineheight[1]{\fontsize{\fsize}{#1\fsize}\selectfont}%
  \ifx\svgwidth\undefined%
    \setlength{\unitlength}{156.73228707bp}%
    \ifx\svgscale\undefined%
      \relax%
    \else%
      \setlength{\unitlength}{\unitlength * \real{\svgscale}}%
    \fi%
  \else%
    \setlength{\unitlength}{\svgwidth}%
  \fi%
  \global\let\svgwidth\undefined%
  \global\let\svgscale\undefined%
  \makeatother%
  \begin{picture}(1,0.2304909)%
    \lineheight{1}%
    \setlength\tabcolsep{0pt}%
    \put(0,0){\includegraphics[width=\unitlength,page=1]{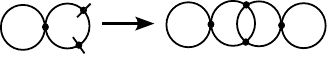}}%
    \put(0.11949367,0.01426479){\color[rgb]{0,0,0}\transparent{0.99626899}\makebox(0,0)[lt]{\lineheight{1.25}\smash{\begin{tabular}[t]{l}\scalebox{1}{$(a)$}\end{tabular}}}}%
  \end{picture}%
\endgroup%

}
\end{minipage}
\caption{
A construction of $(4,4,\ldots,4)$-graph with tadpoles.
Here, we start with the subgraph (a) in Figure \ref{4EL}, and add another (a).
Presently, there are no specifications on the orientations on the edges, therefore, the dipoles contained in this Feynman graph can either be L or R.
}
\label{l2procedureshort}
\end{figure}

Instead, one can also add the $\oD$-loop of length four without a tadpole as drawn as (c) in Figure \ref{4EL} 
Then there are still free half-edges and the procedure continues. 
A priori, there are two possibilities, but due to the planarity condition, adding (c) in the way that is illustrated in Figure \ref{l2procedureforbidden} is not allowed. Whatever $\oD$-loop you add next (either (a), (b), (c), or an $\oD$-loop of length 2), the resulting graph will necessarily be non-planar.
\begin{figure}[H]
\centering
\begin{minipage}[t]{0.7\textwidth}
\centering
\def\svgwidth{0.5\columnwidth}
\tiny{
\begingroup%
  \makeatletter%
  \providecommand\color[2][]{%
    \errmessage{(Inkscape) Color is used for the text in Inkscape, but the package 'color.sty' is not loaded}%
    \renewcommand\color[2][]{}%
  }%
  \providecommand\transparent[1]{%
    \errmessage{(Inkscape) Transparency is used (non-zero) for the text in Inkscape, but the package 'transparent.sty' is not loaded}%
    \renewcommand\transparent[1]{}%
  }%
  \providecommand\rotatebox[2]{#2}%
  \newcommand*\fsize{\dimexpr\f@size pt\relax}%
  \newcommand*\lineheight[1]{\fontsize{\fsize}{#1\fsize}\selectfont}%
  \ifx\svgwidth\undefined%
    \setlength{\unitlength}{156.73228707bp}%
    \ifx\svgscale\undefined%
      \relax%
    \else%
      \setlength{\unitlength}{\unitlength * \real{\svgscale}}%
    \fi%
  \else%
    \setlength{\unitlength}{\svgwidth}%
  \fi%
  \global\let\svgwidth\undefined%
  \global\let\svgscale\undefined%
  \makeatother%
  \begin{picture}(1,0.2304909)%
    \lineheight{1}%
    \setlength\tabcolsep{0pt}%
    \put(0,0){\includegraphics[width=\unitlength,page=1]{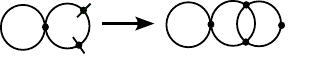}}%
    \put(0.11949367,0.01426479){\color[rgb]{0,0,0}\transparent{0.99626899}\makebox(0,0)[lt]{\lineheight{1.25}\smash{\begin{tabular}[t]{l}\scalebox{1}{$(a)$}\end{tabular}}}}%
    \put(0,0){\includegraphics[width=\unitlength,page=2]{l2procedureforbidden.pdf}}%
  \end{picture}%
\endgroup%

}
\end{minipage}
\caption{
A forbidden configuration which will necessarily end up with a non-planar graph.
}
\label{l2procedureforbidden}
\end{figure}

Therefore, we are only left with one specific way  in adding an $\oD$-loop without a tadpole (i.e., (c)) as drawn in the first step in Figure \ref{l2procedure}.
We continue the procedure by adding  (c)'s.
The procedure ends by adding an $\oD$-loop with a tadpole (i.e., (a) or (b)) similarly to the minimal case in Figure \ref{l2procedureshort}.
We end up with a ladder containing L- and/or R-dipoles with a tadpole on either end.
An example of this procedure is displayed in Figure \ref{l2procedure}, where first, two (c)'s are added,  then subsequently  an (a) or (b) is added to end the procedure.
In this Figure \ref{l2procedure}, it is also shown how the resulting graph can be expressed in terms of dipoles.
\begin{figure}[H]
\centering
\begin{minipage}[t]{1\textwidth}
\centering
\def\svgwidth{1\columnwidth}
\tiny{
\begingroup%
  \makeatletter%
  \providecommand\color[2][]{%
    \errmessage{(Inkscape) Color is used for the text in Inkscape, but the package 'color.sty' is not loaded}%
    \renewcommand\color[2][]{}%
  }%
  \providecommand\transparent[1]{%
    \errmessage{(Inkscape) Transparency is used (non-zero) for the text in Inkscape, but the package 'transparent.sty' is not loaded}%
    \renewcommand\transparent[1]{}%
  }%
  \providecommand\rotatebox[2]{#2}%
  \newcommand*\fsize{\dimexpr\f@size pt\relax}%
  \newcommand*\lineheight[1]{\fontsize{\fsize}{#1\fsize}\selectfont}%
  \ifx\svgwidth\undefined%
    \setlength{\unitlength}{669.42244883bp}%
    \ifx\svgscale\undefined%
      \relax%
    \else%
      \setlength{\unitlength}{\unitlength * \real{\svgscale}}%
    \fi%
  \else%
    \setlength{\unitlength}{\svgwidth}%
  \fi%
  \global\let\svgwidth\undefined%
  \global\let\svgscale\undefined%
  \makeatother%
  \begin{picture}(1,0.08051473)%
    \lineheight{1}%
    \setlength\tabcolsep{0pt}%
    \put(0,0){\includegraphics[width=\unitlength,page=1]{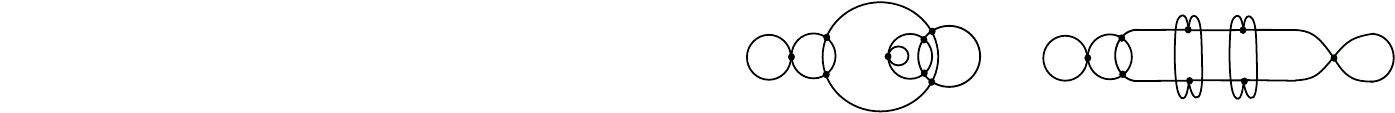}}%
    \put(0.70772274,0.03497277){\color[rgb]{0,0,0}\makebox(0,0)[lt]{\lineheight{1.25}\smash{\begin{tabular}[t]{l}\scalebox{2}{$\cong$}\end{tabular}}}}%
    \put(0,0){\includegraphics[width=\unitlength,page=2]{l2procedure.pdf}}%
    \put(0.02676848,0.00342079){\color[rgb]{0,0,0}\transparent{0.99626899}\makebox(0,0)[lt]{\lineheight{1.25}\smash{\begin{tabular}[t]{l}\scalebox{1}{$(a)$}\end{tabular}}}}%
  \end{picture}%
\endgroup%

}
\end{minipage}
\caption{Example of construction of $(4,4,\ldots,4)$-graph with tadpoles.
Presently, there are no specifications on the orientations on the edges, therefore, the dipoles contained in this Feynman graph can either be L or R.
}
\label{l2procedure}
\end{figure}

One can therefore construct general planar graphs with $\ell=2$ from the $(4,4,\ldots,4)$ graphs by adding $\oD$-loops of length $2$ that do not change the genus. In the case where there are no tadpoles present 
(e.g., graphs in Figure \ref{l2noTP}) these $\oD$-loops will be part of a melonic two-point function, because of planarity.
In the case where tadpoles are present in both grey discs in Figure \ref{2ELremoval} (the rightmost graph of Figure \ref{l2procedure} is such an example), adding such non-melonic $\oD$-loops, amounts to adding N-dipoles to the ladder (which can be possibly empty). 
(See Figure \ref{l2procedurenonplanar} for an example.)
\begin{figure}[H]
\begin{minipage}[t]{0.9\textwidth}
\centering
\def\svgwidth{0.6\columnwidth}
\tiny{
\begingroup%
  \makeatletter%
  \providecommand\color[2][]{%
    \errmessage{(Inkscape) Color is used for the text in Inkscape, but the package 'color.sty' is not loaded}%
    \renewcommand\color[2][]{}%
  }%
  \providecommand\transparent[1]{%
    \errmessage{(Inkscape) Transparency is used (non-zero) for the text in Inkscape, but the package 'transparent.sty' is not loaded}%
    \renewcommand\transparent[1]{}%
  }%
  \providecommand\rotatebox[2]{#2}%
  \newcommand*\fsize{\dimexpr\f@size pt\relax}%
  \newcommand*\lineheight[1]{\fontsize{\fsize}{#1\fsize}\selectfont}%
  \ifx\svgwidth\undefined%
    \setlength{\unitlength}{391.34797344bp}%
    \ifx\svgscale\undefined%
      \relax%
    \else%
      \setlength{\unitlength}{\unitlength * \real{\svgscale}}%
    \fi%
  \else%
    \setlength{\unitlength}{\svgwidth}%
  \fi%
  \global\let\svgwidth\undefined%
  \global\let\svgscale\undefined%
  \makeatother%
  \begin{picture}(1,0.2134916)%
    \lineheight{1}%
    \setlength\tabcolsep{0pt}%
    \put(0,0){\includegraphics[width=\unitlength,page=1]{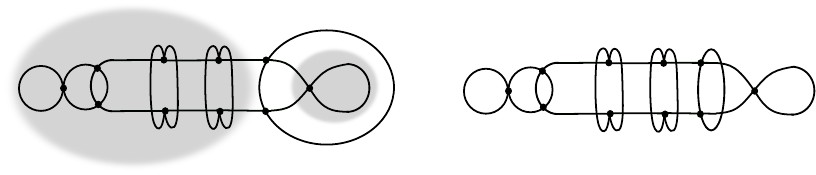}}%
    \put(0.49041048,0.08561098){\color[rgb]{0,0,0}\makebox(0,0)[lt]{\lineheight{1.25}\smash{\begin{tabular}[t]{l}\scalebox{2}{$\cong$}\end{tabular}}}}%
  \end{picture}%
\endgroup%

}
\end{minipage}
\caption{
Example of adding a non-melonic $\oD$-loop of length 2 to a $(4,4,\ldots,4)$ graph. We identified the Figure \ref{2ELremoval}'s grey discs on the left graph for illustration.
}
\label{l2procedurenonplanar}
\end{figure}
All possible schemes of graphs with $\ell=2$, $g=0$ are given in Figure \ref{fig:Sum4}, after recovering the orientation assignment on edges. 
\end{proof}

By plugging in $\ell=1$, $g=1$ in equation \eqref{eq:seven}, we obtain the following relation
\begin{align}
    \varphi = \frac{v+3}{2}.
\end{align}
This relation implies that the number of vertices should be odd.

\begin{thm}
\label{propo:g1ell1}
For $\ell=1$ and $g=1$, there are four 2PI schemes as given in Figure \ref{fig:Sum1} and 36 2PR schemes  as given in Figure \ref{fig:Sum2}. 
\end{thm}

\begin{figure}[H]
\centering
\begin{minipage}[t]{0.6\textwidth}
\centering
\def\svgwidth{0.6\columnwidth}
\tiny{
\begingroup%
  \makeatletter%
  \providecommand\color[2][]{%
    \errmessage{(Inkscape) Color is used for the text in Inkscape, but the package 'color.sty' is not loaded}%
    \renewcommand\color[2][]{}%
  }%
  \providecommand\transparent[1]{%
    \errmessage{(Inkscape) Transparency is used (non-zero) for the text in Inkscape, but the package 'transparent.sty' is not loaded}%
    \renewcommand\transparent[1]{}%
  }%
  \providecommand\rotatebox[2]{#2}%
  \newcommand*\fsize{\dimexpr\f@size pt\relax}%
  \newcommand*\lineheight[1]{\fontsize{\fsize}{#1\fsize}\selectfont}%
  \ifx\svgwidth\undefined%
    \setlength{\unitlength}{258.22397547bp}%
    \ifx\svgscale\undefined%
      \relax%
    \else%
      \setlength{\unitlength}{\unitlength * \real{\svgscale}}%
    \fi%
  \else%
    \setlength{\unitlength}{\svgwidth}%
  \fi%
  \global\let\svgwidth\undefined%
  \global\let\svgscale\undefined%
  \makeatother%
  \begin{picture}(1,0.42505904)%
    \lineheight{1}%
    \setlength\tabcolsep{0pt}%
    \put(0,0){\includegraphics[width=\unitlength,page=1]{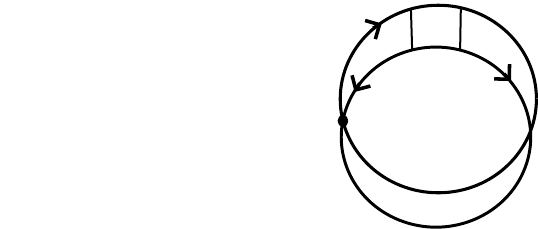}}%
    \put(0.77467852,0.35255159){\color[rgb]{0,0,0}\makebox(0,0)[lt]{\lineheight{1.25}\smash{\begin{tabular}[t]{l}\scalebox{1.2}{$\widetilde{\rm N}_{\rm o}$}\end{tabular}}}}%
    \put(0,0){\includegraphics[width=\unitlength,page=2]{Sum1.pdf}}%
    \put(0.14186461,0.36038502){\color[rgb]{0,0,0}\makebox(0,0)[lt]{\lineheight{1.25}\smash{\begin{tabular}[t]{l}\scalebox{1.2}{$\widetilde{\rm N}_{\rm o}$}\end{tabular}}}}%
    \put(0,0){\includegraphics[width=\unitlength,page=3]{Sum1.pdf}}%
  \end{picture}%
\endgroup%

}
\caption{
All $\ell=1$, $g=1$ 2PI schemes.
There are four schemes of 2PI $\ell=1$ and $g=1$ in total.
$\widetilde {\rm N}_{\rm o}\in\{ $N-dipole$, {\rm N}_{\rm o}\}$.}
\label{fig:Sum1}
\end{minipage}
\end{figure}
\begin{figure}[H]
\begin{minipage}[t]{1\textwidth}
\centering
\def\svgwidth{1\columnwidth}
\tiny{
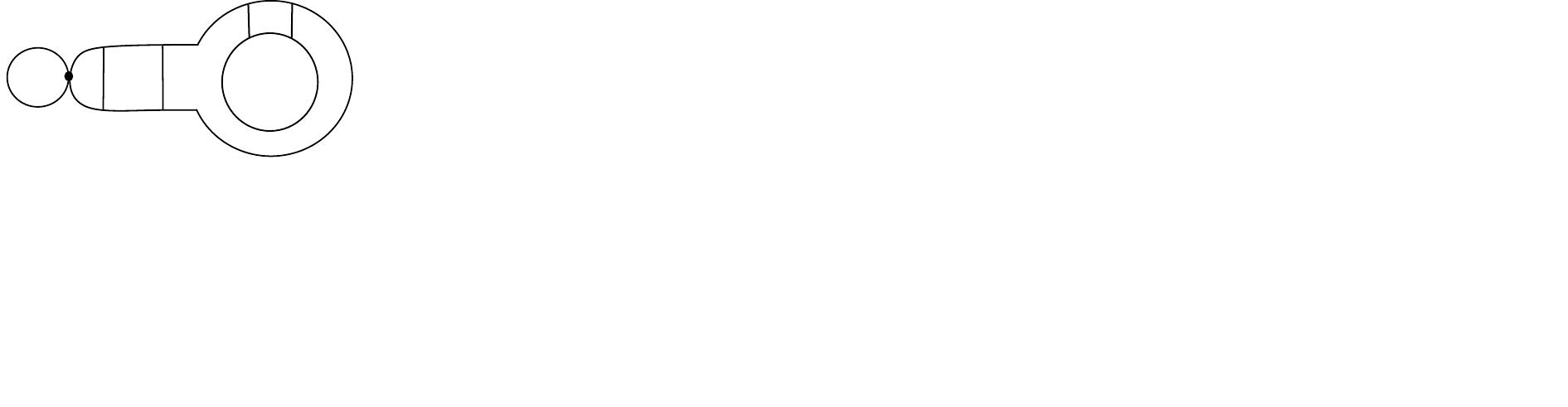
}
\caption{
All $\ell=1$, $g=1$ 2PR schemes.
There are 36 schemes of $\ell=1$ and $g=1$ in total.
${\rm Y} \in 
\{$L-dipole$, \, $R-dipole$,\,{\rm L}, {\rm  R}, {\rm N}_{\rm e}, {\rm B}\}$. $\widetilde {\rm N}_{\rm o}\in\{ $N-dipole$, {\rm N}_{\rm o}$\}.
}
\label{fig:Sum2}
\end{minipage}
\end{figure}

\begin{proof}
Consider melon-free Feynman graph $G$ of $\ell=1$ genus one and $S_G$ its scheme.

{\bf($\boldsymbol{\ell=1}$, 2PI.)}
We begin by assuming that $G$ is 2PI. According to 
Corollary \ref{cor:N_dip}, there must be a connecting (type I) N-dipole in $G$. Because of Lemma \ref{lem:sep-conn}, this N-dipole must be connecting (type I)
(a separating one would violate the 2PI condition).
Furthermore, the maximal extension of this N-dipole in $G$ is either the N-dipole itself or a maximal N-ladder (a B-ladder is necessarily separating by Lemma \ref{lem:sep-conn}), which corresponds to a N-dipole or a N-vertex in $S_G$. The contraction of the maximal extension of this connecting N-dipole with $\Delta \ell = 0$ and $\Delta g = -1$ yields a $\ell=1$ Feynman graph $G'$ of genus zero using  Eq.\ \eqref{eq:contNonSepN} or \eqref{eq:LVcontNonSepN} with $\sigma=-1$ (by the connecting property). 

Now, if $G'$ is not melon-free, we must be in one of the configurations shown in Figures \ref{fig:create-melon1}, \ref{fig:create-melon2}, \ref{fig:create-melon3}. However, the configuration \ref{fig:create-melon1} is excluded by maximality of $X$, while both configurations \ref{fig:create-melon2} and \ref{fig:create-melon3} yield 2PR graphs $G$. 
Therefore, $G'$ must be melon-free.

Let us analyze further $\ell=1$ and 2PI graphs $G$ keeping in mind that $G'$ must be melon-free.
If $G'$ is melon-free, it must be the infinity graph which can be either clockwise or counter-clockwise; and the orientation of the edges imposes that the contraction involved is an N-dipole or an N$_\mathrm{o}$-ladder. As a result, $S_G$ can be one of the graphs in Figure {\ref{fig:2PIg1}}.

\begin{figure}[H]
\centering
\begin{minipage}[t]{0.8\textwidth}
\centering
\def\svgwidth{0.7\columnwidth}
\tiny{
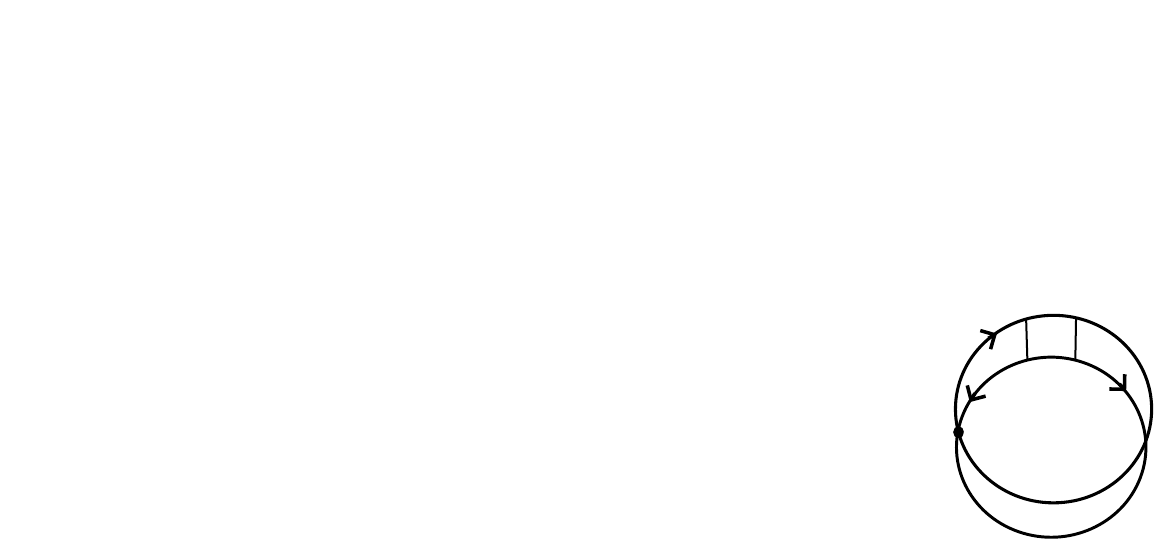
}
\caption{The schematic instruction on how to obtain the only 2PI $\ell=1$, $g=1$ schemes given in Figure \ref{fig:Sum1}, starting with $\ell=1$ $g=0$ schemes $G'$ on the left. One cuts the edges where scissors are placed, and insert a connecting N-dipole, or a connecting N-ladder-vertex accordingly. The red dotted lines indicate the paths of  $\oD$-strands. $\widetilde {\rm N}_{\rm o} \in \{ $N-dipole$, {\rm N}_{\rm o}\}$.
Notice that for illustration, we have used two different infinity graphs (clockwise and counter-clockwise), however, once two cuts are performed, their resulting graphs are identical. Nevertheless, after the insertion of a connecting N, we will obtain all four distinct schemes drawn above.
}
\label{fig:2PIg1}
\end{minipage}
\end{figure}

\medskip

{\bf($\boldsymbol{\ell=1}$, 2PR.)} 
We will now assume that $G$ is 2PR. 
We distinguish two cases. 

In the first case, $G$ contains a separating dipole. We call $X$ its maximal extension and consider its contraction. We call the two resulting graphs $G_1$ and $G_2$. 

In the second case, $G$ does not contain a separating dipole. It still contains a two-edge-cut because it is 2PR. We perform the flip operation as described in Lemma \ref{lem:2PR} and call the resulting graphs $G_1$ and $G_2$. 

In both cases the following equations hold , $\ell(G_1) + \ell(G_2) = \ell(G)$ and $g(G_1) + g(G_2) = g(G)$, as discussed in Propositions \ref{prop:dipolecases} and \ref{prop:ladder-vertices} in the first case and in Lemma \ref{lem:2PR} in the second case. Since $g(G) =1$, we can assume without loss of generality that  $g(G_1)=0$ and $g(G_2) = 1$.

 A priori, there are two options:
 \begin{itemize}
     \item [(a)]
     $g(G_1) = 0$, $\ell(G_1) = 0$ and $g(G_2) = 1$, $\ell(G_2) = 1$. Therefore, $G_1$ is the cycle graph, 
     possibly decorated with melons.
     Since $G$ is melon-free, it is easy to see that $G_1$ should be melon-free as well because of the maximality of $X$ in the first case and because of the non-existence of a separating dipole in the second case. Thus $G_1$ is the melon-free cycle graph. This implies that $G$ is not melon-free in the first case, while in the second case this configuration (where $\tilde G_1 = \emptyset$ in Figure \ref{fig:flipOp}) is not permitted by the definition of a two-edge-cut. In conclusion, option (a) cannot occur.

     \item [(b)]
     $g(G_1) = 0$, $\ell(G_1) = 1$ and $g(G_2) = 1$, $\ell(G_2) = 0$. Thus, $G_1$ is the infinity graph of either clockwise or counterclockwise, 
     possibly decorated with melons,
     and there is also only one possible $G_2$ which is $S_1$ (given in Figure \ref{fig:Ell0g1.jpg}),  
     possibly decorated with melons.
 \end{itemize}
Therefore, in the end, we  omit case (a) and we only have one possibility, which is (b) above. Let us examine the case (b) further below.

Suppose $G_1$ or $G_2$ contains a melonic subgraph.

\begin{itemize}
    \item 
    In the first case, this would violate either the maximality of $X$ 
    (see Figure \ref{fig:create-melon1})
    or the assumption that $G$ is melon-free.
    \item 
    In the second case, this would violate either the assumption that $G$ does not contain a separating dipole or the assumption that $G$ is melon-free.
\end{itemize}

Therefore, we conclude that $G_1$ nor $G_2$ cannot contain a melonic subgraph.

Then, the only possibilities are given in Figure \ref{2prEll1g1} below.
\begin{figure}[H]
\begin{center}
\begin{minipage}[t]{0.95\textwidth}
\centering
\def\svgwidth{1\columnwidth}
\tiny{
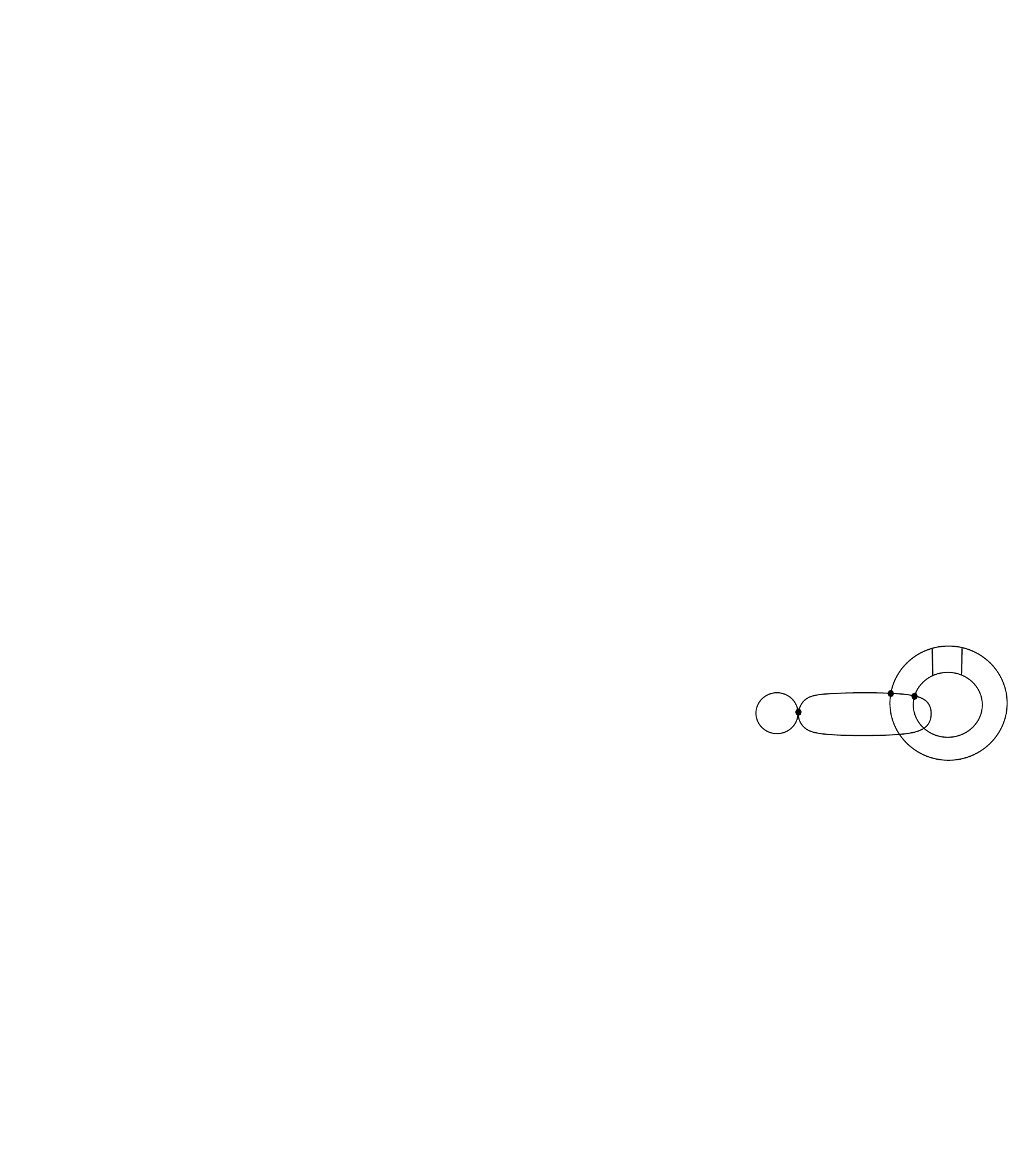
}
\caption{
$\ell=1$, $g=1$ 2PR schemes of $G$ (on the right), which are generated by inserting a separating dipole or separating ladder into $G_1$ (infinity graphs given in Fig.\,\ref{fig:Ell1g0.jpg}) and $G_2 = S_1$ given in Fig. \ref{fig:Ell0g1.jpg}.
$\widetilde {\rm N}_{\rm o} \in \{ $N-dipole$,\, {\rm N}_{\rm o}\}$, 
$Y \in \{$L-dipole$, \,  $R-dipole$, {\rm L}, \, {\rm R}, \,{\rm N}_{\rm e} \}$.
Note that we have taken advantage of the flexibility in where cuts can be made (scissor placement) to represent graphs in our chosen form.
See Figure \ref{fig:isomorphismwithcuts} in Appendix \ref{sec:appiso} for the isomorphisms of the graphs with cuts.
}
\label{2prEll1g1}
\end{minipage}
\end{center}
\end{figure}

This concludes that for $\ell=1$ and $g=1$, we have 36 2PR schemes listed in Figure \ref{fig:Sum2}, in addition to the only 2PI four schemes shown in Figure \ref{fig:Sum1}. 
\end{proof}

\medskip

\begin{thm}
\label{propo:g1ell2}
For $\ell=2$ and $g=1$, there are 42 2PI schemes as given in Figure \ref{fig:Sum5} in addition to more 2PR schemes.
\end{thm}

\begin{figure}[H]
\centering
\begin{minipage}[t]{0.9\textwidth}
\centering
\def\svgwidth{0.9\columnwidth}
\tiny{
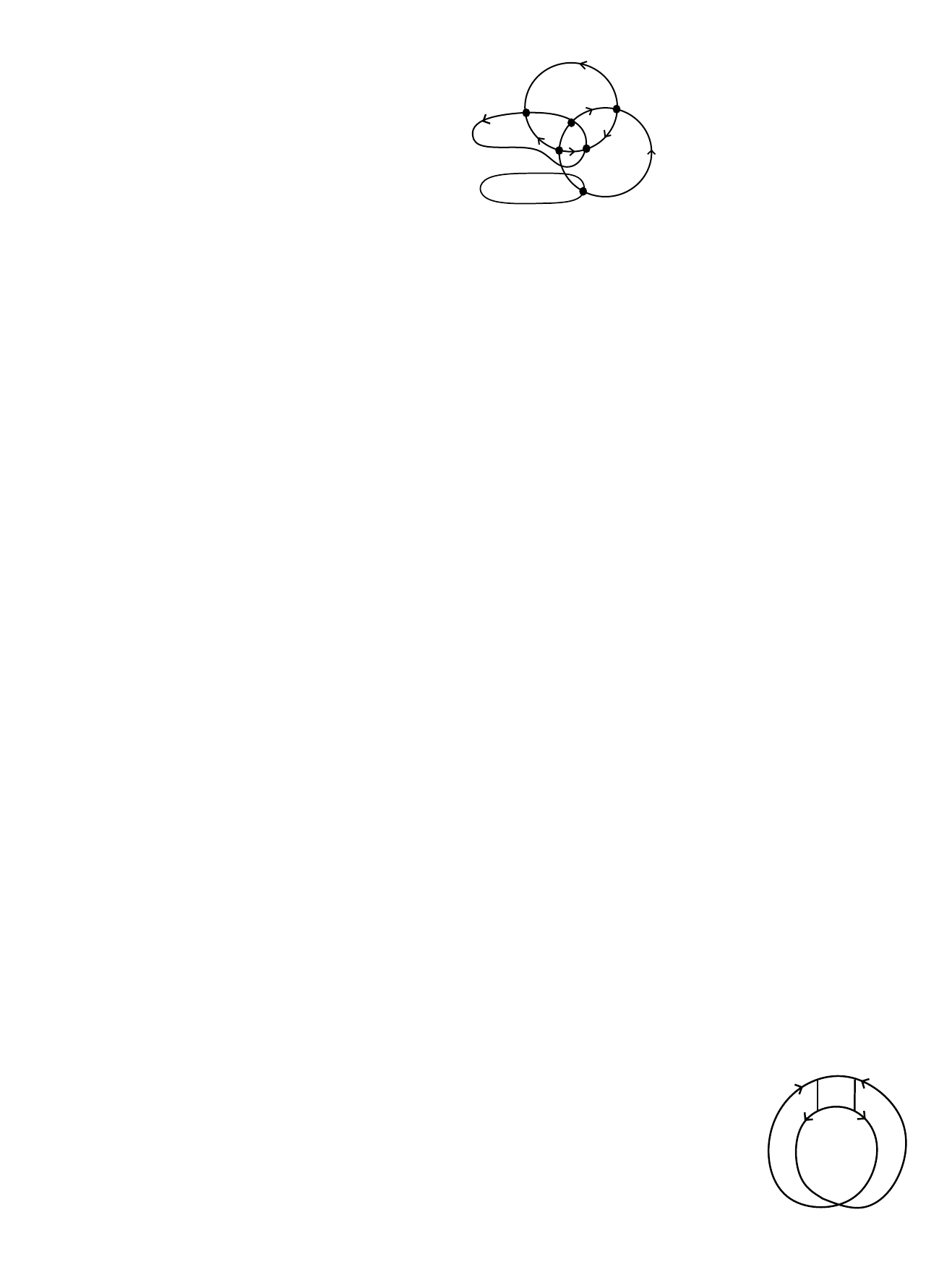
}
\caption{
All $\ell=2$, $g=1$ 2PI schemes.
There are 42 2PI schemes of $\ell=2$ and $g=1$ in total.
$\widetilde {\rm L} \in \{$L-dipole$, \, $L$\}$, 
$\widetilde {\rm R} \in \{$R-dipole$, \, $R$\}$, 
$\widetilde {\rm N}_{\rm o} \in \{$N-dipole$, {\rm N}_{\rm o}\}$.
}
\label{fig:Sum5}
\end{minipage}
\end{figure}

\begin{proof}
Let $G$ be a 
melon-free Feynman graph of $\ell=2$ genus one and $S_G$ its scheme.

{\bf($\boldsymbol{\ell=2}$, 2PI.)}
We first assume that $G$ is 2PI.
By Lemma \ref{lem:sep-conn}, there necessarily exists a non-separating (either connecting (type I, $\sigma = -1$) or rearranging (type II, $\sigma = 0$)) N-dipole
in $G$ (a separating one would break the 2PI condition).
Furthermore, the maximal extension of this 
N-dipole in $G$ is either the 
N-dipole itself or a maximal N-ladder respectively (a B-ladder is necessarily separating by Lemma \ref{lem:sep-conn}), which translates into a N-dipole or a 
N-vertex in $S_G$, respectively. 
Using Eq.\ \eqref{eq:contNonSepN} or \eqref{eq:LVcontNonSepN},
\begin{itemize}
\item [(1)]
the contraction of the maximal extension of this connecting ($\sigma=-1$) N-dipole of type I with $\Delta \ell = 0$ and $\Delta g = -1$ yields a $\ell=2$ Feynman graph $G'$ of genus zero,
\item [(2)] 
whereas the contraction of the maximal extension of the rearranging ($\sigma=0$) N-dipole of type II with $\Delta \ell = -2$ and $\Delta g = -1$ yields a $\ell=0$ Feynman graph $G'$ of genus zero.
\end{itemize}

For both (1) and (2) above, if $G'$ is not melon-free, 
we must be in one of the configurations shown in Figures \ref{fig:create-melon1}, \ref{fig:create-melon2}, 
or
\ref{fig:create-melon3}.
However, the configuration \ref{fig:create-melon1} 
is excluded by maximality of $X$, while  
configurations 
\ref{fig:create-melon2}
and
\ref{fig:create-melon3} 
yield  2PR $G$. 
Therefore, $G'$ must be melon-free.

\medskip

Now, let us further examine the above cases (1) and (2) separately keeping in mind that $G'$ is melon-free.

\begin{itemize}
\item[(1)] 
 If $G'$ (can be either 2PI or 2PR a priori) is melon-free, it must be either $S^{2,0}_1$, $S^{2,0}_2$, or $S^{2,0}_3$ given in Figures \ref{fig:Sum3} and \ref{fig:Sum4}.
 We explictily show in Figures
 \ref{fig:App1}, 
 \ref{fig:App2} and \ref{fig:App21},
 and
 \ref{fig:App3} 
 all possible operations of inserting a connecting N 
 (either N-dipole, ${\rm N}_{\rm e}$-ladder-vertex or ${\rm N}_{\rm o}$-ladder-vertex)
 to an $\ell=2$, $g=0$ 
 scheme respectively $S^{2,0}_1$ (2PI), $S^{2,0}_2$ of L-type (2PI), $S^{2,0}_2$ of R-type (2PI), and $S^{2,0}_3$ (2PR) to produce 2PI schemes of $\ell=2$ and $g=1$.
In all of the figures  \ref{fig:App1}, 
 \ref{fig:App2} and \ref{fig:App21}, 
 and
 \ref{fig:App3},
the red scissors indicate where we cut and insert a connecting N, whereas the red dotted lines indicate the $\oD$-loops.
Based on Remark \ref{remarkuniqueness}, listing all possible cuts is sufficient, however, there can be isomorphisms between the resulting schemes, and we need to identify them.
Refer Appendix \ref{sec:appiso} for isomorphisms of schemes.
We finally list all the distinct schemes generated by such operations in Figure \ref{fig:Sum5}.

\begin{figure}[H]
\begin{minipage}[t]{0.8\textwidth}
\centering
\def\svgwidth{0.8\columnwidth}
\tiny{
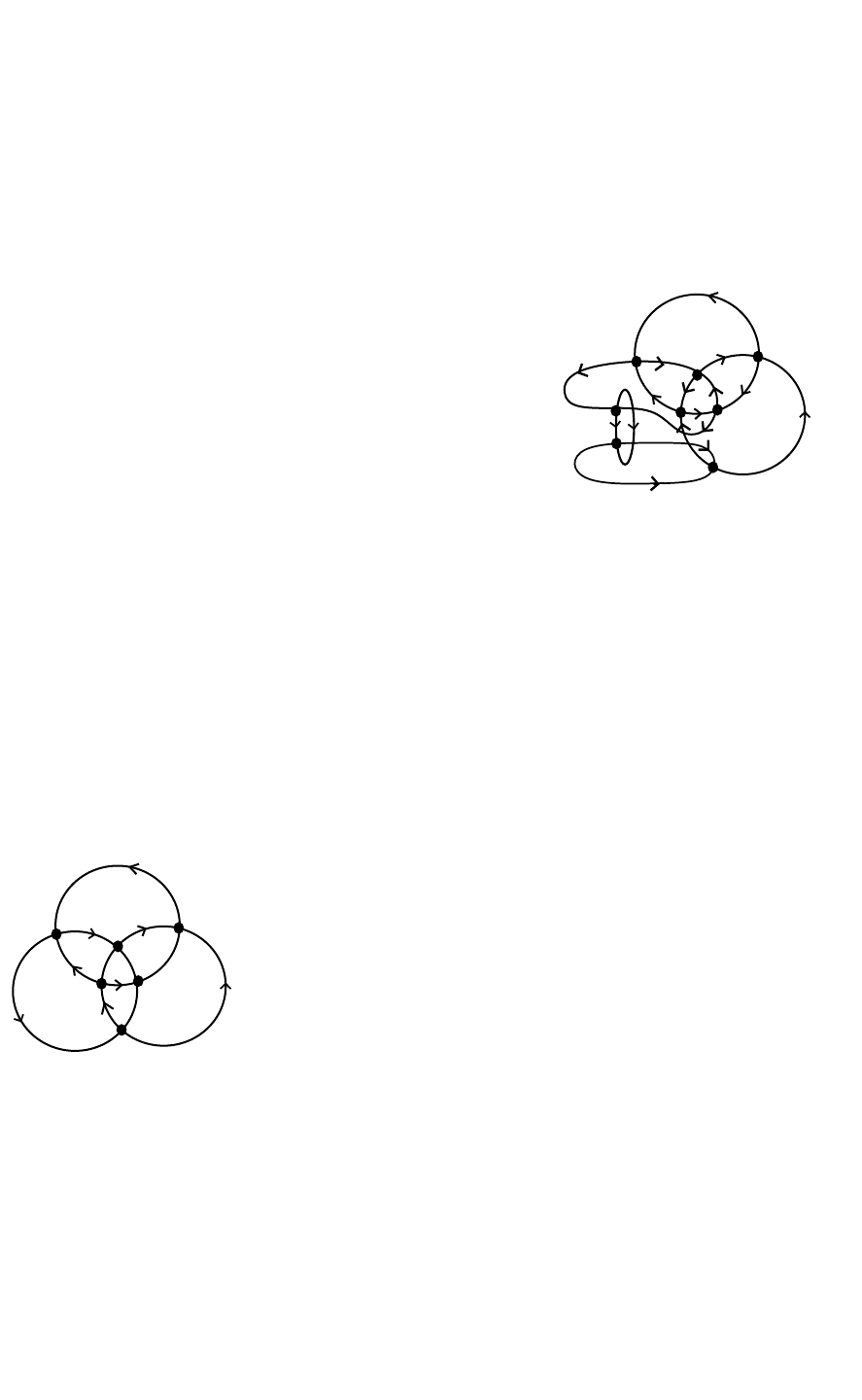
}
\caption{
We obtain an $\ell=2$, $g=1$ 2PI scheme by inserting a compatible connecting N to an $\ell=2$, $g=0$ 2PI scheme, $S^{2,0}_1$.
$\widetilde {\rm N}_{\rm o} \in \{ $N-dipole$, {\rm N}_{\rm o}\}$.
Note that if we insert only one N-dipole as in the top two rows, the resulting graphs happen to be isomorphic.
See Figure \ref{fig:app1isomorphism}.
See Figure \ref{fig:App1isomophismwithcuts}
for additional isomorphism of graphs related to the graphs listed here.
}
\label{fig:App1}
\end{minipage}
\end{figure}

\begin{figure}[H]
\begin{center}
\begin{minipage}[t]{0.9\textwidth}
\centering
\def\svgwidth{0.7\columnwidth}
\tiny{
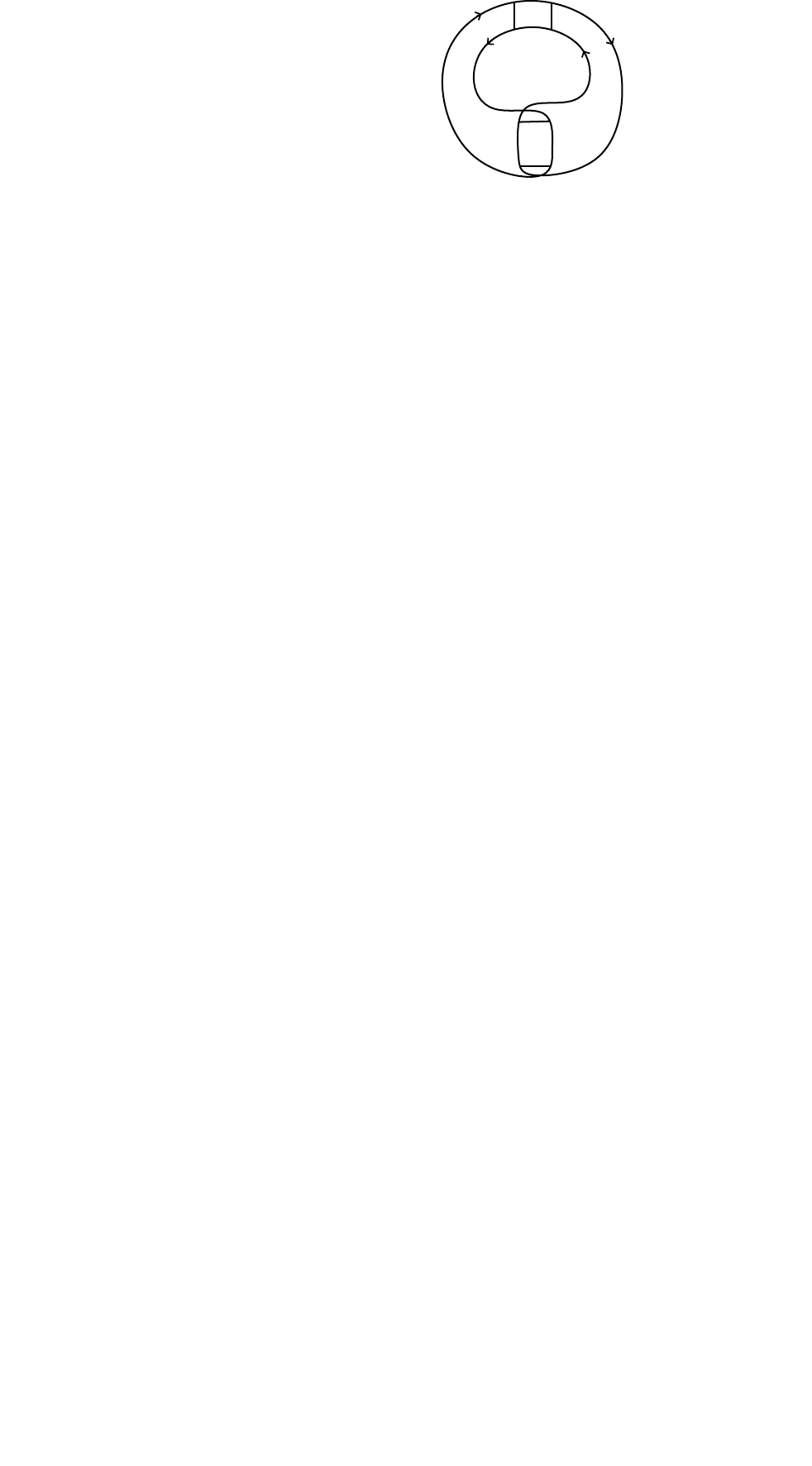
}
\end{minipage}
\caption{
All possible procedures of how we obtain $\ell=2$, $g=1$ 2PI schemes by inserting a connecting N to $\ell=2$, $g=0$ 2PI scheme, $S^{2,0}_2$ of L-type in Fig. \ref{fig:Sum3}.
See Fig. \ref{fig:App2isomorphismwithcuts} for isomorphisms of $S^{2,0}_2$ of L-type with two cuts identified.
$\widetilde {\rm L} \in \{ $L-dipole$, {\rm L}$\}.
On the second row, in principle, we can also start with $\emptyset$ in replacement with $\widetilde {\rm L}$ in $S^{2,0}_2$, however, then, this graph with two identified cuts is isomorphic to the graph on the top row with two L-rungs in ${\rm L}$ (see Fig. \ref{fig:App2isomorphismwithcutsREDUCED}).
Similarly, on the third, forth, and fifth rows, in principle, we can also start with $\emptyset$ in replacement with $\widetilde {\rm L}$ in $S^{2,0}_2$, however, the resulting schemes will be 2PR, so we will not list them here.
In the last two rows, although it is not obvious, two operations will yield the isomorphic scheme
}
\label{fig:App2}
\end{center}
\end{figure}

\begin{figure}[H]
\begin{center}
\begin{minipage}[t]{0.9\textwidth}
\centering
\def\svgwidth{0.7\columnwidth}
\tiny{
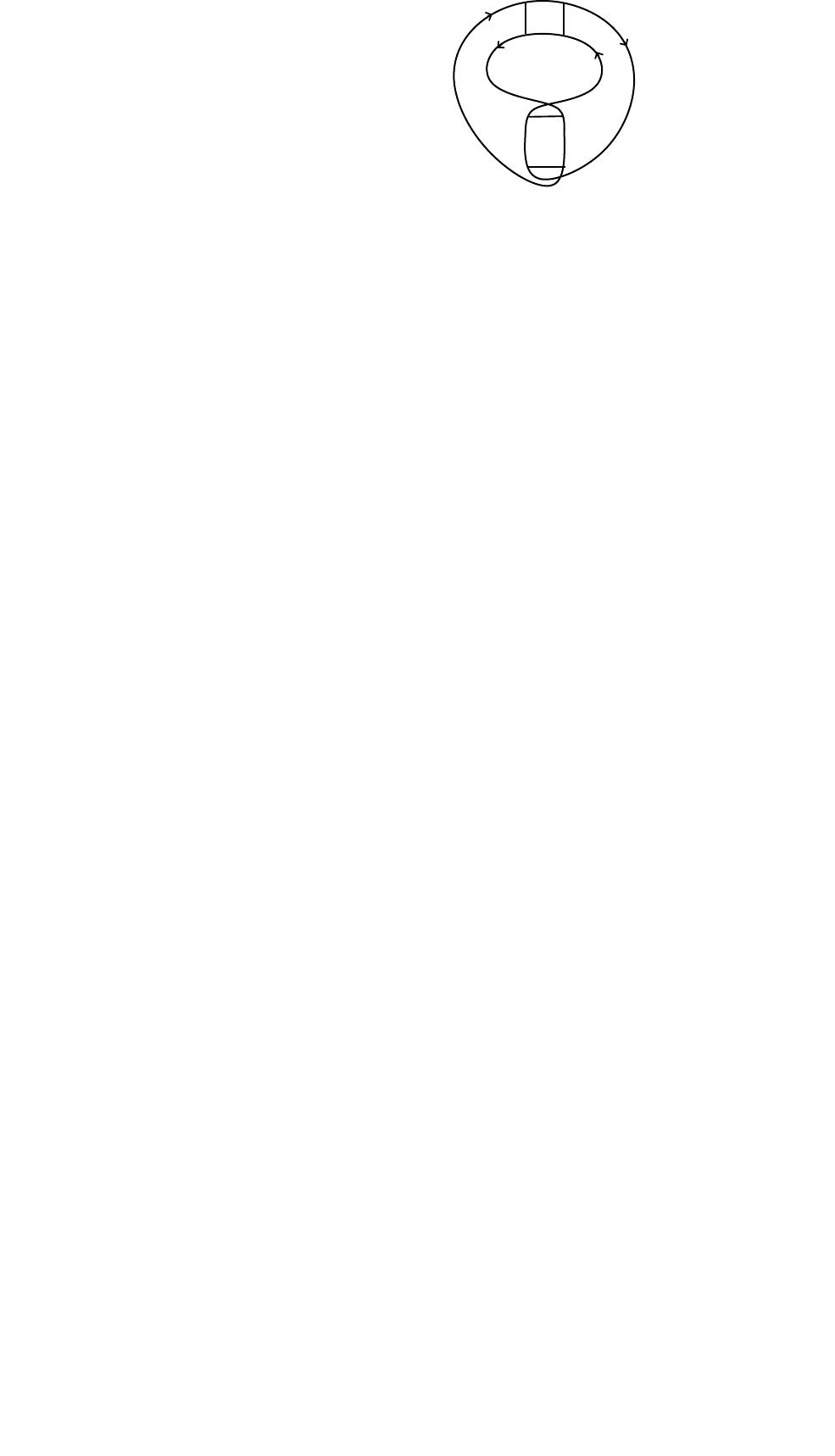
}
\caption{
All possible procedures of how we obtain $\ell=2$, $g=1$ 2PI schemes by inserting a connecting N to $\ell=2$, $g=0$ 2PI scheme, $S^{2,0}_2$ of R-type in Fig. \ref{fig:Sum3}.
See Fig. \ref{fig:App2isomorphismwithcuts} for isomorphisms of $S^{2,0}_2$ of R-type with two cuts identified.
$\widetilde {\rm R} \in \{ $R-dipole$, {\rm R}$\}.
On the second row, in principle, we can also start with $\emptyset$ in replacement with $\widetilde {\rm R}$ in $S^{2,0}_2$, however, then, this graph with two identified cuts is isomorphic to the graph on the top row with two R-rungs in ${\rm R}$ (see Fig. \ref{fig:App2isomorphismwithcutsREDUCED}).
Similarly, on the third, forth and fifth rows, in principle, we can also start with $\emptyset$ in replacement with $\widetilde {\rm R}$ in $S^{2,0}_2$, however, the resulting schemes will be 2PR, so we will not list them here.
In the last two rows, although it is not obvious, two operations will yield the isomorphic scheme.
}
\label{fig:App21}
\end{minipage}
\end{center}
\end{figure}

\begin{figure}[H]
\begin{minipage}[t]{1\textwidth}
\centering
\def\svgwidth{0.95\columnwidth}
\tiny{
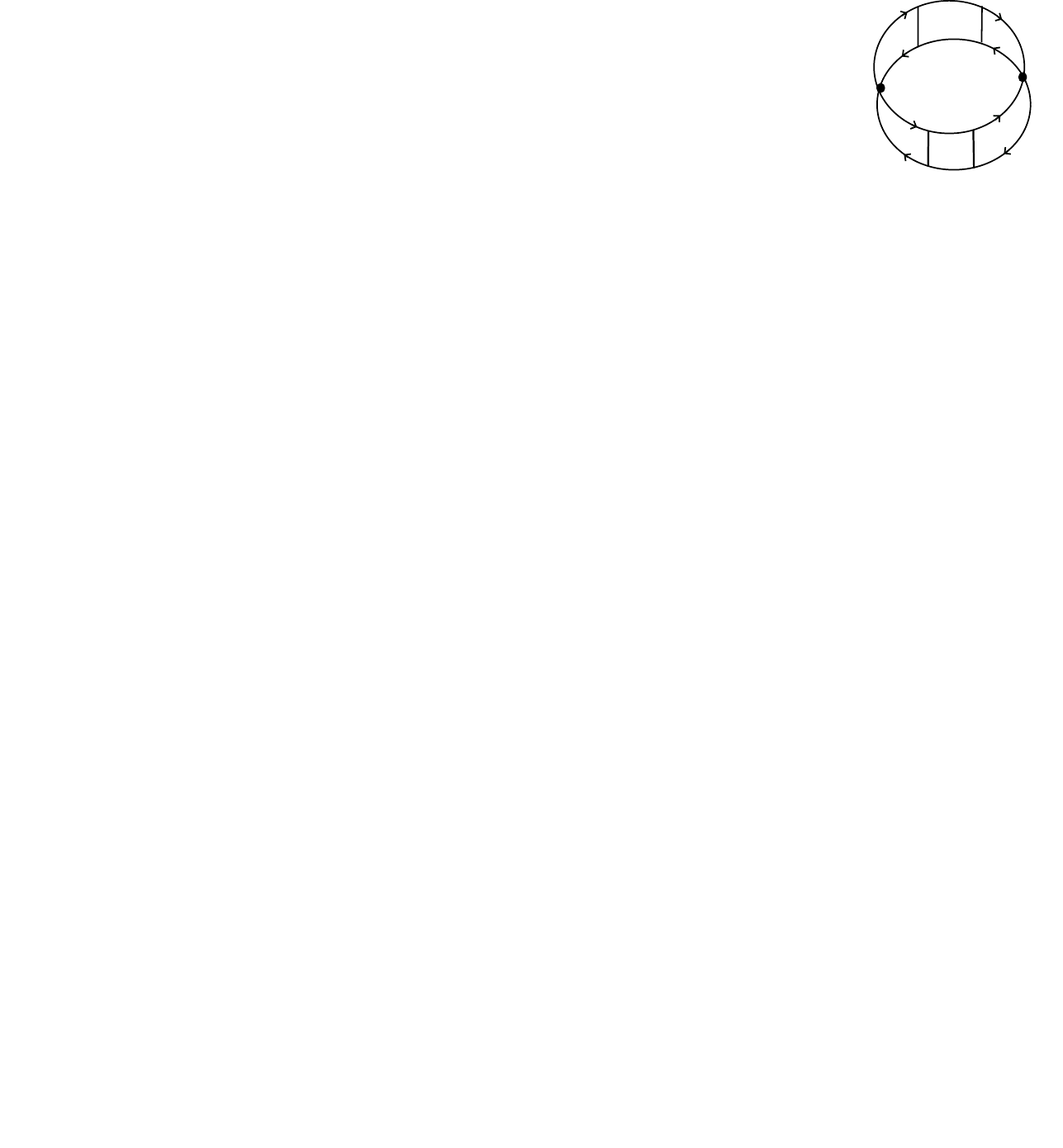
}
\caption{
We show how to obtain an $\ell=2$, $g=1$ 2PI scheme by inserting a connecting N to an $\ell=2$, $g=0$ 2PR scheme, $S_3^{2,0}$ (2PR) given in Figure \ref{fig:Sum4}. 
$\widetilde {\rm N}_{\rm o} \in \{ $N-dipole$, {\rm N}_{\rm o}\}$.
}
\label{fig:App3}
\end{minipage}
\end{figure}

\item[(2)]
If $G'$ is melon-free, it must be the cycle graph.

We explicitly demonstrate in Figure \ref{fig:App4} how to generate the  scheme of $\ell=2$ $g=1$ by inserting a rearranging N.

\begin{figure}[H]
\begin{minipage}[t]{0.8\textwidth}
\centering
\def\svgwidth{0.9\columnwidth}
\tiny{
\begingroup%
  \makeatletter%
  \providecommand\color[2][]{%
    \errmessage{(Inkscape) Color is used for the text in Inkscape, but the package 'color.sty' is not loaded}%
    \renewcommand\color[2][]{}%
  }%
  \providecommand\transparent[1]{%
    \errmessage{(Inkscape) Transparency is used (non-zero) for the text in Inkscape, but the package 'transparent.sty' is not loaded}%
    \renewcommand\transparent[1]{}%
  }%
  \providecommand\rotatebox[2]{#2}%
  \newcommand*\fsize{\dimexpr\f@size pt\relax}%
  \newcommand*\lineheight[1]{\fontsize{\fsize}{#1\fsize}\selectfont}%
  \ifx\svgwidth\undefined%
    \setlength{\unitlength}{356.90053624bp}%
    \ifx\svgscale\undefined%
      \relax%
    \else%
      \setlength{\unitlength}{\unitlength * \real{\svgscale}}%
    \fi%
  \else%
    \setlength{\unitlength}{\svgwidth}%
  \fi%
  \global\let\svgwidth\undefined%
  \global\let\svgscale\undefined%
  \makeatother%
  \begin{picture}(1,0.23923288)%
    \lineheight{1}%
    \setlength\tabcolsep{0pt}%
    \put(0,0){\includegraphics[width=\unitlength,page=1]{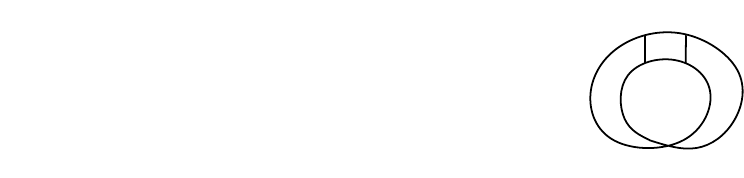}}%
    \put(0.87406073,0.16547321){\color[rgb]{0,0,0}\makebox(0,0)[lt]{\lineheight{1.25}\smash{\begin{tabular}[t]{l}\scalebox{1.2}{$\widetilde {\rm N}_{\rm o}$}\end{tabular}}}}%
    \put(0,0){\includegraphics[width=\unitlength,page=2]{App4.pdf}}%
    \put(0.48869639,0.1087605){\color[rgb]{0,0,0}\makebox(0,0)[lt]{\lineheight{1.25}\smash{\begin{tabular}[t]{l}\scalebox{1.2}{${\rm N}$}\end{tabular}}}}%
    \put(0,0){\includegraphics[width=\unitlength,page=3]{App4.pdf}}%
    \put(0.33485602,0.10460606){\color[rgb]{0,0,0}\makebox(0,0)[lt]{\lineheight{1.25}\smash{\begin{tabular}[t]{l}\scalebox{2}{$+$}\end{tabular}}}}%
    \put(0.67190864,0.10961175){\color[rgb]{0,0,0}\makebox(0,0)[lt]{\lineheight{1.25}\smash{\begin{tabular}[t]{l}\scalebox{2}{$=$}\end{tabular}}}}%
    \put(0,0){\includegraphics[width=\unitlength,page=4]{App4.pdf}}%
  \end{picture}%
\endgroup%

}
\caption{
We obtain an $\ell=2$, $g=1$ 2PI scheme by inserting a rearranging N to the $\ell=0$, $g=0$ cycle graph. 
The red scissors indicate where we cut and insert a connecting N, and the red dotted line indicates the $\oD$ path. 
$\widetilde {\rm N}_{\rm o} \in \{ $N-dipole$, {\rm N}_{\rm o}\}$.
}
\label{fig:App4}
\end{minipage}
\end{figure}

\end{itemize}

\medskip

{\bf($\boldsymbol{\ell=2}$, 2PR.)}
We now assume that $G$ is 2PR.

We distinguish two cases. 

In the first case, $G$ contains a separating dipole. We call $X$ its maximal extension and consider its contraction. We call the two resulting graphs $G_1$ and $G_2$. 

In the second case, $G$ does not contain a separating dipole. It still contains a two-edge-cut because it is 2PR. We perform the flip operation as described in Lemma \ref{lem:2PR} and call the resulting graphs $G_1$ and $G_2$. 

In both cases the following equations hold , $\ell(G_1) + \ell(G_2) = \ell(G)$ and $g(G_1) + g(G_2) = g(G)$, as discussed in Propositions \ref{prop:dipolecases} and \ref{prop:ladder-vertices} in the first case and in Lemma \ref{lem:2PR} in the second case. Since $g(G) =1$, we can assume without loss of generality that  $g(G_1)=0$ and $g(G_2) = 1$.

 A priori, there are three possibilities:
 \begin{itemize}
     \item [(a)]
     $g(G_1) = 0$, $\ell(G_1) = 0$ and $g(G_2) = 2$, $\ell(G_2) = 2$. Hence, $G_1$ is the cycle graph, 
     possibly decorated with melons.
    Since $G$ is melon-free, it is easy to see that $G_1$ should be melon-free as well because of the maximality of $X$ in the first case and because of the non-existence of a separating dipole in the second case. Thus $G_1$ is the melon-free cycle graph. This implies that $G$ is not melon-free in the first case, while in the second case this configuration (where $\tilde G_1 = \emptyset$ in Figure \ref{fig:flipOp}) is not permitted by the definition of a two-edge-cut. In conclusion, option (a) cannot occur.

     \item [(b)]
     $g(G_1) = 0$, $\ell(G_1) = 2$ and $g(G_2) = 1$, $\ell(G_2) = 0$. Therefore, $G_1$ is one of corresponding graphs from 
     Figures \ref{fig:Sum3} and \ref{fig:Sum4}. The graph
     $G_2$ is $S_1$ in Figure \ref{fig:Ell0g1.jpg}, 
     possibly decorated with melons.

     \item[(c)]
     $g(G_1) = 0$, $\ell(G_1) = 1$ and $g(G_2) = 1$, $\ell(G_2) = 1$. Thus, $G_1$ is the infinity graph, 
     possibly decorated with melons,
     and $G_2$ is either corresponding 2PI graphs given in Figure \ref{fig:Sum1}, or corresponding 2PR graphs given in Figure \ref{fig:Sum2}.

 \end{itemize}
Therefore, a posteriori, (a) must be omitted, and we only have two possibilities (b) and (c) of above. 
In both of these cases (b) and (c), let us analyze the situation further.

Suppose $G_1$ or $G_2$ contains a melonic subgraph.

\begin{itemize}
    \item 
    In the first case, this would violate either the maximality of $X$ 
    (see Figure \ref{fig:create-melon1})
    or the assumption that $G$ is melon-free.
    \item 
    In the second case, this would violate either the assumption that $G$ does not contain a separating dipole or the assumption that $G$ is melon-free.
\end{itemize}

Therefore, we conclude that $G_1$ nor $G_2$ cannot contain a melonic subgraph. 
Then, we show specific examples below of 2PR $\ell=2$ graphs.
For (b), Figure \ref{fig:2PRex1} shows an example, and more examples can be found in Figure \ref{fig:2PRex1app} in Appendix \ref{sec:app2PR}.
For (c), Figure \ref{fig:2PRex2} illustrates an example, and another example can be found in Figure \ref{fig:2PRex2app} in Appendix \ref{sec:app2PR}.

\begin{figure}[H]
\begin{minipage}[t]{1\textwidth}
\centering
\def\svgwidth{1\columnwidth}
\tiny{
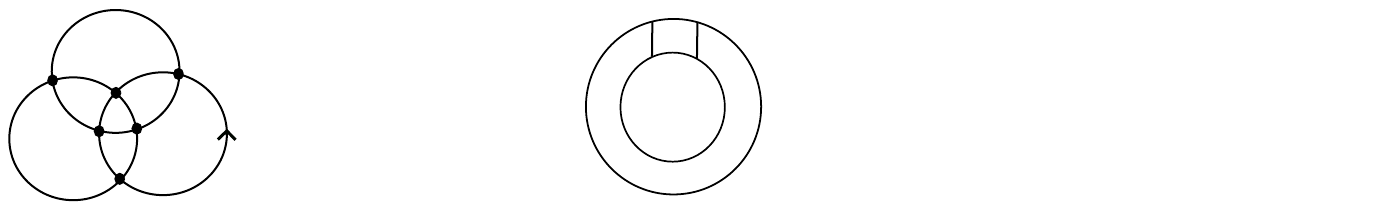
}
\caption{
2PR schemes obtained from combining a $\ell=2$ and $g = 0$ scheme with a $\ell=0$ and $g=1$ scheme via separating dipole or a ladder-vertex or a two-edge-connection.
$\widetilde {\rm N}_{\rm o} \in \{ $N-dipole$,\, {\rm N}_{\rm o}\}$.
See more in Figure \ref{fig:2PRex1app}.
}
\label{fig:2PRex1}
\end{minipage}
\end{figure}
\begin{figure}[H]
\begin{minipage}[t]{0.95\textwidth}
\centering
\def\svgwidth{1\columnwidth}
\tiny{
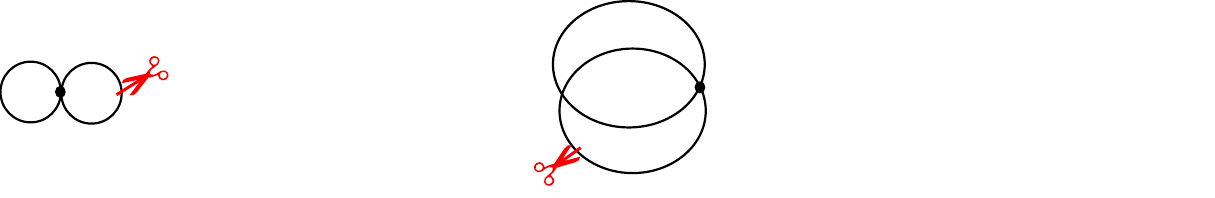
}
\caption{
2PR schemes obtained from combining a $\ell=1$ and $g = 0$  infinity graph with a $\ell=1$ and $g=1$ scheme via separating dipole or a ladder-vertex.  
Orientations on edges are ignored in this illustration.
$ X\in\{$N-dipole$, \,
$L-dipole$, \,$R-dipole$,
\, {\rm N}_{\rm e}, {\rm N}_{\rm o}, 
{\rm L}, {\rm R}, 
{\rm B}\}$.
See more in Figure \ref{fig:2PRex2app}.
}
\label{fig:2PRex2}
\end{minipage}
\end{figure}
This concludes that in the end, we show some  2PR graphs for $\ell=2$ and $g=1$ in Figures \ref{fig:2PRex1} and \ref{fig:2PRex2},  in addition to the only 2PI ones (42 of them) shown in Figure \ref{fig:Sum5}.

\end{proof}

From the families of  $\ell = 1$ and $\ell=2$ melon-free Feynman graphs of genus $g' < g$, the following theorem provides a method for constructing all of the $\ell=1$ and $\ell=2$ melon-free Feynman graphs of any genus $g$.   

\begin{thm}
\label{thm:induction}
{\bf$\underline{\boldsymbol{\ell=1}}$.}
Let $G$ be a connected and melon-free Feynman graph, such that  $\ell(G) = 1$ and $g(G) = g \geq 2$.

{\bf($\boldsymbol{\ell=1}$, 2PI.)}
Suppose first that $G$ is 2PI. 
Then it can be obtained by insertion of a connecting N-dipole or N-ladder into a: $\ell=1$, connected and melon-free Feynman graph of genus $g-1$.

{\bf($\boldsymbol{\ell=1}$, 2PR.)}
Suppose instead that $G$ is 2PR. 
Then, $G$ can be obtained by insertion of a separating dipole, a separating ladder, or a two-edge-connection in between:  connected and melon-free Feynman graphs $G_1$ and $G_2$ with 
$\ell(G_1) = 1$ and $\ell (G_2) = 0$, such that  $0   \le g(G_1)  < g$, $0 <  g(G_2) \le  g$
and $g(G_1) + g(G_2) = g$.

\medskip

{\bf$\underline{\boldsymbol{\ell=2}}$.}
Let us now suppose $G$ be a connected and melon-free Feynman graph, such that $\ell(G) = 2$ and $g(G) = g \geq 2$.

{\bf($\boldsymbol{\ell=2}$, 2PI.)}
Suppose first that $G$ is 2PI. Then it can be obtained by insertion of either 
\begin{itemize}
    \item 
    a connecting N-dipole or N-ladder into a: $\ell=2$, connected and melon-free Feynman graph of genus $g-1$,
    \item
    a rearranging N-dipole or N-ladder into a: $\ell=0$, connected 
    and melon-free 
    Feynman graph of genus $g-1$.

\end{itemize}

{\bf($\boldsymbol{\ell=2}$, 2PR.)}
Suppose instead that $G$ is 2PR. 
Then one of these two conditions holds:
\begin{itemize}
    \item 
    $G$ can be obtained by insertion of a separating dipole, a separating ladder, or a two-edge-connection in-between:  connected and melon-free Feynman graphs $G_1$ and $G_2$ with $\ell(G_1) = 2$ and $\ell (G_2) = 0$, such that   $0 \le g(G_1) \le g-1$, $0  < g(G_2) \le g$, and $g(G_1) + g(G_2) = g$.
 
    \item 
    $G$ can be obtained by insertion of a separating dipole, a separating ladder, or a two-edge-connection in between: $\ell=1$, connected, and melon-free Feynman graphs $G_1$ and $G_2$, such that $g(G_1) + g(G_2) = g$.
\end{itemize}
\end{thm}

\begin{proof} 
We prove by generalizing the arguments of Theorem \ref{propo:g1ell1} and Theorem \ref{propo:g1ell2}.
\medskip

{\bf$\underline{\boldsymbol{\ell=1}}$.}
Let us consider $\ell(G)=1$ first.

{\bf($\boldsymbol{\ell=1}$, 2PI.)}
If $G$ is 2PI, there exists a connecting N-dipole in $G$, which can be maximally extended. Contracting the resulting N-dipole or N-ladder $X$ yields a 
$\ell=1$ Feynman graph $G'$ of genus $g(G') = g-1$. Furthermore, $G'$ is necessarily melon-free because the three configurations of Figure \ref{fig:generate-melon} are all excluded: configuration \ref{fig:create-melon1} because $X$ is maximal, 
configurations \ref{fig:create-melon2} and \ref{fig:create-melon3} because $G$ should be 2PI but \ref{fig:create-melon2} and \ref{fig:create-melon3} necessarily give $G$ to be 2PR.
This proves the first part of the proposition.

{\bf($\boldsymbol{\ell=1}$, 2PR.)}
Let us suppose that $G$ is 2PR, and discuss two distinct situations:

-Suppose first that $G$ does not contain a separating dipole. Since $G$ is 2PR, it must still contain a two-edge-cut $(e,e')$. By Lemma \ref{lem:2PR}, performing a flip on $(e,e')$ yields two Feynman graphs $G_1$ and $G_2$ such that 
$\ell(G_1) =1$, $\ell(G_2) = 0$ without the loss of generality.
However, we cannot have $g(G_2) =0 $, because then it is not allowed by the definition of the two-edge-cut. Therefore, 
$g(G_1) < g$ and $g(G_1) + g(G_2) = g$.
As in Theorem~\ref{propo:g1ell1}, $G_1$ and $G_2$ are necessarily melon-free, 
otherwise it would imply that there was a separating dipole in $G$ 
or that $G$ was not melon-free to begin with.

-Assume instead that $G$ contains a separating dipole. As in Theorem \ref{propo:g1ell1}, we proceed to contracting its maximal extension $X$, itself a separating dipole or a separating ladder. The two resulting 
Feynman graphs $G_1$ with $\ell(G_1) = 1$ and $G_2$  with $\ell(G_2) = 0$ are again melon-free. However, again $g(G_2) =0$ is excluded because then, it means that $G$ has a melon.
Therefore, they obey: $g(G_1) < g$ and $g(G_1) + g(G_2) = g$.

This concludes the proof for $\ell(G) = 1$.

\medskip

{\bf$\underline{\boldsymbol{\ell=2}}$.}
Let us now consider $\ell(G)=2$.

{\bf($\boldsymbol{\ell=2}$, 2PI.)}
If $G$ is 2PI, there exists either
\begin{itemize}
\item 
a connecting N-dipole in $G$, which can be maximally extended, or
\item 
a rearranging N-dipole in $G$, which can be maximally extended.
\end{itemize}
Furthermore, the consequences of the both cases above follow.
Contracting the resulting connecting (resp. rearranging) N-dipole or N-ladder (in Figure \ref{fig:generate-melon} they may be denoted as $X$ (resp. $X$ or $\tilde X$))  yields a  Feynman graph $G'$ of genus $g(G') = g-1 \ge 1$ with $\ell=2$ (resp. $\ell=0$). Furthermore, $G'$ is necessarily melon-free because the three configurations of Figure \ref{fig:generate-melon} are all excluded: configuration \ref{fig:create-melon1} because $X$ is maximal, 
configurations \ref{fig:create-melon2} and \ref{fig:create-melon3} because $G$ should be 2PI by assumption but \ref{fig:create-melon2} and \ref{fig:create-melon3} necessarily give $G$ to be 2PR.

\medskip

{\bf($\boldsymbol{\ell=2}$, 2PR.)}
Let us suppose that $G$ is 2PR. We shall distinguish two cases:

\begin{itemize}
\item
$G$ does not contain a separating dipole. Since $G$ is 2PR, it must still contain a two-edge-cut $(e,e')$. By Lemma \ref{lem:2PR}, performing a flip on $(e,e')$ yields two Feynman graphs $G_1$ and $G_2$.

\item
$G$ contains a separating dipole. As in Theorem \ref{propo:g1ell2}, we proceed to contracting its maximal extension $X$, itself a separating dipole or a separating ladder, which will result in the two resulting Feynman graphs $G_1$ and $G_2$.
\end{itemize}

In both of the cases above, without a loss of generality, $G_1$ and $G_2$ will obey either cases:
\begin{itemize}
\item 
$\ell(G_1) = 2$, $\ell (G_2) =0$. However, we cannot have $g(G_2) =0 $, 
because 

- for the case $G$ does not contain a separating dipole, it is not allowed by the definition of the two-edge-cut. 

- for the case $G$ contains a separating dipole, it means that $G$ had a melon.

Therefore,   $g(G_1) < g$  and $g(G_1) + g(G_2) = g$.

\item
$\ell (G_1) = 1$, $\ell (G_2) = 1$. And, $g(G_1) + g(G_2) = g$.
\end{itemize}
\end{proof}

\medskip

Theorem \ref{thm:induction} is a key result of this paper. 
It provides a constructive way of generating all the $\ell=1$ and $\ell=2$ Feynman graphs, order by order in the genus and starting at genus one, by application of the following algorithms.

\paragraph{{\color{black}{Algorithmic construction of all}} 
{$\ell=1, 2$} {\color{black}{graphs.}}}

\begin{enumerate}
    \item Assume that the set 
    $\hat{E}_{g', \, \ell=1}$ (resp. $\hat{E}_{g', \, \ell=2}$)
     melon-free Feynman graphs of genus $g'\leq g-1$ has been constructed. Additionally, note that in \cite{MR4450018}, the authors
     have already constructed the set $\hat{E}_{g', \, \ell=0}$,  $\forall g' \ge 0$, 
     which is necessary to construct all $\ell = 2$ graphs; 
    
    \item 
    Construct the set $\hat{E}^\mathrm{2PI}_{g, \, \ell=1}$ (resp. $\hat{E}^\mathrm{2PI}_{g, \, \ell=2}$) of 2PI melon-free Feynman graphs of genus $g$, by inserting a connecting N-dipole or N-ladder into any element of $\hat{E}_{g-1, \, \ell=1}$, 
    (resp. by inserting either 
    \begin{itemize}
    \item
    a connecting N-dipole or N-ladder into any element of $\hat{E}_{g-1, \, \ell=2}$, or 
    \item
    a rearranging N-dipole or N-ladder into any element of $\hat{E}_{g-1, \, \ell=0}$,)
    \end{itemize}
    in any possible way that yields a 2PI graph. 
    Pay attention that it is possible that the graph before the insertion is 2PR, yielding a 2PI graph.
    Therefore, it is necessary to start from elements of $\hat{E}$ as opposed to $\hat{E}^\mathrm{2PI}$ for both $\ell =1$ and $\ell=2$;

    \item
    Construct the set $\hat{E}^\mathrm{2PR}_{g, \, \ell = 1}$ (resp. $\hat{E}^\mathrm{2PR}_{g, \, \ell = 2}$) of 2PR  melon-free Feynman graphs by inserting 
    a separating dipole, a separating ladder, or a two-edge-connection 
    in between any element of $\hat{E}_{g_1, \ell_1=1}$ and $\hat{E}_{g_2, \ell_2 = 0}$ such that $g_1 + g_2 = g$ and $g_1 < g$ 
    (resp. in between any element of either 
    \begin{itemize}
        \item $\hat{E}_{g_1, \ell_1=2}$ and $\hat{E}_{g_2, \ell_2 = 0}$ such that $g_1 + g_2 = g$ and $g_1 < g$, or
        \item $\hat{E}_{g_1, \ell_1=1}$ and $\hat{E}_{g_2, \ell_2 = 1}$ such that $g_1 + g_2 = g$)
    \end{itemize}  in any possible way.

   \item 
Obtain the set 
$\hat{E}_{g, \, \ell=1} = \hat{E}^\mathrm{2PR}_{g, \, \ell=1} \cup \hat{E}^\mathrm{2PI}_{g, \, \ell=1}$ 
(resp. $\hat{E}_{g, \, \ell=2} = \hat{E}^\mathrm{2PR}_{g, \, \ell=2} \cup \hat{E}^\mathrm{2PI}_{g, \, \ell=2}$ )
of all melon-free Feynman graphs of genus $g$;
    
    \item Finally, insert arbitrary melonic 
    Feynman graphs on the edges of any element in 
    $\hat{E}_{g, \, \ell=1}$ (resp. $\hat{E}_{g, \, \ell=2}$ )
    to obtain the set 
    $E_{g, \, \ell=1}$ (resp. $E_{g, \, \ell=2}$)
    Feynman graphs (see the fourth bullet in Remark \ref{rk:melon}).
\end{enumerate}

\medskip

\section{Classification for $\ell=3$ with $g=0$}
\label{sec:l3g0}

This discussion follows a similar approach as in the planar case for $\ell =1$ and $\ell =2$.
By plugging $\ell=3$, $g=0$ in equation \eqref{eq:seven}, we obtain the following relation
\begin{align}
    \varphi = \frac{v-1}{2}.
\end{align}
This relation implies that the number of vertices should be odd.

\begin{lemma}
\label{l3loops}
For each graph with $\ell=3$, $g=0$, it either contains at least one $\oD$-loop of length $2$ or has $\oD$-loop configuration $(4,4,\ldots,4,6)$.
\end{lemma}
\begin{proof}
    The proof is very similar to the proof of 
    Lemma \ref{l2}.
    Consider a planar graph with $\ell=3$ that does not contain an $\oD$-loop of length two and does not have $\oD$-loop configuration $(4,4,\ldots,4,6)$.  Then because  $l_i$ is even, the smallest possible $e$ is obtained when $e= 4\varphi$ or $e= 4\varphi + 4$. These expressions of $e$ give successively the following two equations:
    \begin{eqnarray}
        &&e = 4\varphi = 4(\frac{v-1}{2}) = 2v - 2 < 2v,\cr
        &&e = 4\varphi + 4 = 4(\frac{v-1}{2}) + 4 = 2v + 2 > 2v.
    \end{eqnarray}
    Because the graphs we are considering are $4$-regular, we have that $e=2v$.
    This gives a contradiction in both cases, hence the statement is proved.
\end{proof}

In the following, we should distinguish different embeddings of a tadpole as illustrated in Figure \ref{tadpoleembedding}.

\begin{figure}[H]
\centering
\begin{minipage}[t]{0.4\textwidth}
\centering
\def\svgwidth{0.6\columnwidth}
\tiny{
\begingroup%
  \makeatletter%
  \providecommand\color[2][]{%
    \errmessage{(Inkscape) Color is used for the text in Inkscape, but the package 'color.sty' is not loaded}%
    \renewcommand\color[2][]{}%
  }%
  \providecommand\transparent[1]{%
    \errmessage{(Inkscape) Transparency is used (non-zero) for the text in Inkscape, but the package 'transparent.sty' is not loaded}%
    \renewcommand\transparent[1]{}%
  }%
  \providecommand\rotatebox[2]{#2}%
  \newcommand*\fsize{\dimexpr\f@size pt\relax}%
  \newcommand*\lineheight[1]{\fontsize{\fsize}{#1\fsize}\selectfont}%
  \ifx\svgwidth\undefined%
    \setlength{\unitlength}{313.50510244bp}%
    \ifx\svgscale\undefined%
      \relax%
    \else%
      \setlength{\unitlength}{\unitlength * \real{\svgscale}}%
    \fi%
  \else%
    \setlength{\unitlength}{\svgwidth}%
  \fi%
  \global\let\svgwidth\undefined%
  \global\let\svgscale\undefined%
  \makeatother%
  \begin{picture}(1,0.56160678)%
    \lineheight{1}%
    \setlength\tabcolsep{0pt}%
    \put(0,0){\includegraphics[width=\unitlength,page=1]{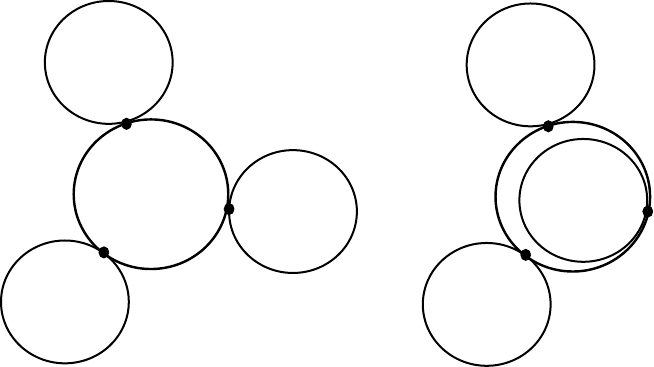}}%
    \put(0.43119614,0.23108651){\color[rgb]{0.01960784,0,0.01960784}\makebox(0,0)[lt]{\lineheight{1.25}\smash{\begin{tabular}[t]{l}$e$\end{tabular}}}}%
    \put(0.87378955,0.24818251){\color[rgb]{0.01960784,0,0.01960784}\makebox(0,0)[lt]{\lineheight{1.25}\smash{\begin{tabular}[t]{l}$e$\end{tabular}}}}%
  \end{picture}%
\endgroup%

}
\caption{Example of distinct tadpole (labelled by $e$) embeddings.}
\label{tadpoleembedding}
\end{minipage}
\end{figure}

\begin{thm}
\label{thm:ell3g02PI}
For $\ell=3$ and $g=0$, there are four 2PI schemes as listed in Figure \ref{l32PI}.
Up to tadpole embeddings and orientations of the edges, all the 2PR graphs are listed in Figure \ref{l3all2PR}.
\end{thm}

\begin{figure}[H]
\centering
\begin{minipage}[t]{0.6\textwidth}
\centering
\def\svgwidth{0.45\columnwidth}
\tiny{
\begingroup%
  \makeatletter%
  \providecommand\color[2][]{%
    \errmessage{(Inkscape) Color is used for the text in Inkscape, but the package 'color.sty' is not loaded}%
    \renewcommand\color[2][]{}%
  }%
  \providecommand\transparent[1]{%
    \errmessage{(Inkscape) Transparency is used (non-zero) for the text in Inkscape, but the package 'transparent.sty' is not loaded}%
    \renewcommand\transparent[1]{}%
  }%
  \providecommand\rotatebox[2]{#2}%
  \newcommand*\fsize{\dimexpr\f@size pt\relax}%
  \newcommand*\lineheight[1]{\fontsize{\fsize}{#1\fsize}\selectfont}%
  \ifx\svgwidth\undefined%
    \setlength{\unitlength}{263.51413036bp}%
    \ifx\svgscale\undefined%
      \relax%
    \else%
      \setlength{\unitlength}{\unitlength * \real{\svgscale}}%
    \fi%
  \else%
    \setlength{\unitlength}{\svgwidth}%
  \fi%
  \global\let\svgwidth\undefined%
  \global\let\svgscale\undefined%
  \makeatother%
  \begin{picture}(1,1.09685773)%
    \lineheight{1}%
    \setlength\tabcolsep{0pt}%
    \put(0,0){\includegraphics[width=\unitlength,page=1]{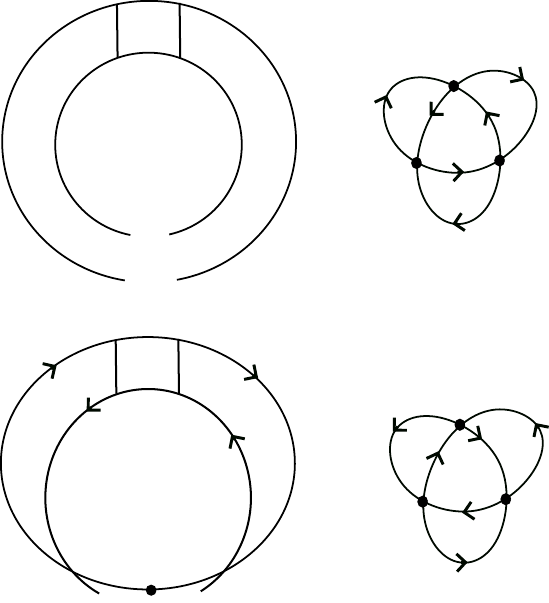}}%
    \put(0.25595916,1.02391743){\color[rgb]{0,0,0}\makebox(0,0)[lt]{\lineheight{1.25}\smash{\begin{tabular}[t]{l}\scalebox{1.2}{${\rm L}$}\end{tabular}}}}%
    \put(0,0){\includegraphics[width=\unitlength,page=2]{l32PI.pdf}}%
    \put(0.25143939,0.4114841){\color[rgb]{0,0,0}\makebox(0,0)[lt]{\lineheight{1.25}\smash{\begin{tabular}[t]{l}\scalebox{1.2}{${\rm R}$}\end{tabular}}}}%
    \put(0,0){\includegraphics[width=\unitlength,page=3]{l32PI.pdf}}%
  \end{picture}%
\endgroup%

}
\caption{All possible 2PI schemes for $\ell=3$, $g=0$. 
Note that if we replace the L-vertex (resp. R-vertex) with an L-dipole (resp. R-dipole) in the top (resp. bottom) left scheme then it is isomorphic to the right graph on the top (resp. bottom).
  }
\label{l32PI}
\end{minipage}
\end{figure}

\begin{figure}[H]
\centering
\begin{minipage}[t]{1\textwidth}
\centering
\def\svgwidth{0.9\columnwidth}
\tiny{
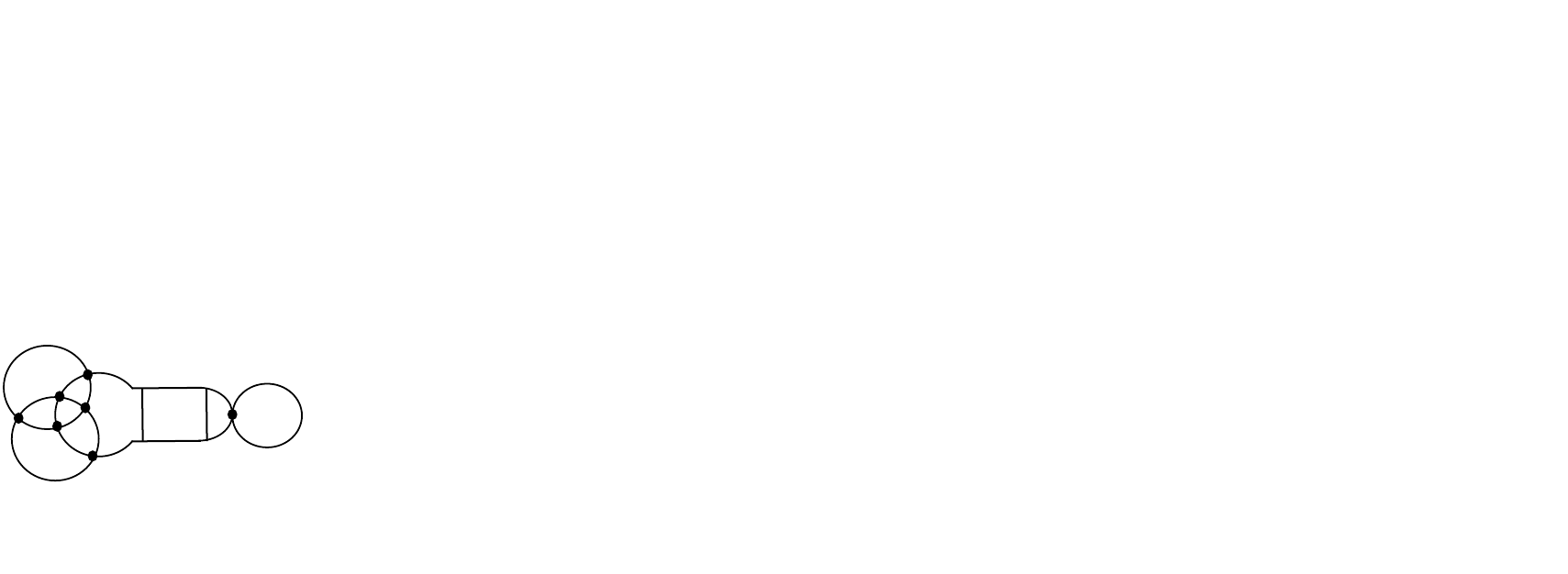
}
\caption{All possible 2PR graphs
for $\ell=3$, $g=0$, up to tadpole embedding and up to orientations of edges.
Note that the drawing on the top row  with $\check{\mathrm I}$ (as defined in Figure \ref{ION}) suggests the existence of a series of $\check{\mathrm I}$ rungs of two or more. 
Depending on the orientations of the edges assigned, they can be either L-, or R-ladder vertex.
The notation $Y$ in this figure can include any non-empty collection of elements of $\{\check {\mathrm I }, \check {\mathrm O }, \check {\mathrm N }\}$ (described in Figure \ref{ION}) in any quantity, combination, and order.
Depending on the orientation of the edges, $Y$ can become ${\mathrm L}$-, ${\mathrm R}$-, or ${\mathrm B}$-ladder vertices.
}
\label{l3all2PR}
\end{minipage}
\end{figure}

\begin{figure}[H]
\centering
\begin{minipage}[t]{0.6\textwidth}
\centering
\def\svgwidth{0.4\columnwidth}
\tiny{
\begingroup%
  \makeatletter%
  \providecommand\color[2][]{%
    \errmessage{(Inkscape) Color is used for the text in Inkscape, but the package 'color.sty' is not loaded}%
    \renewcommand\color[2][]{}%
  }%
  \providecommand\transparent[1]{%
    \errmessage{(Inkscape) Transparency is used (non-zero) for the text in Inkscape, but the package 'transparent.sty' is not loaded}%
    \renewcommand\transparent[1]{}%
  }%
  \providecommand\rotatebox[2]{#2}%
  \newcommand*\fsize{\dimexpr\f@size pt\relax}%
  \newcommand*\lineheight[1]{\fontsize{\fsize}{#1\fsize}\selectfont}%
  \ifx\svgwidth\undefined%
    \setlength{\unitlength}{136.89469826bp}%
    \ifx\svgscale\undefined%
      \relax%
    \else%
      \setlength{\unitlength}{\unitlength * \real{\svgscale}}%
    \fi%
  \else%
    \setlength{\unitlength}{\svgwidth}%
  \fi%
  \global\let\svgwidth\undefined%
  \global\let\svgscale\undefined%
  \makeatother%
  \begin{picture}(1,0.50025462)%
    \lineheight{1}%
    \setlength\tabcolsep{0pt}%
    \put(0,0){\includegraphics[width=\unitlength,page=1]{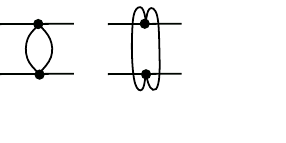}}%
    \put(0.1212544,0.02323573){\color[rgb]{0,0.03529412,0}\makebox(0,0)[lt]{\lineheight{1.25}\smash{\begin{tabular}[t]{l}\scalebox{1.5}{$\check{\mathrm I}$}\end{tabular}}}}%
    \put(0,0){\includegraphics[width=\unitlength,page=2]{ION.pdf}}%
    \put(0.49345538,0.01888559){\color[rgb]{0,0.03529412,0}\makebox(0,0)[lt]{\lineheight{1.25}\smash{\begin{tabular}[t]{l}\scalebox{1.5}{$\check{\mathrm O}$}\end{tabular}}}}%
    \put(0.83662468,0.02196178){\color[rgb]{0,0.03529412,0}\makebox(0,0)[lt]{\lineheight{1.25}\smash{\begin{tabular}[t]{l}\scalebox{1.5}{$\check{\mathrm N}$}\end{tabular}}}}%
  \end{picture}%
\endgroup%

}
\caption{
We define dipoles
$\check {\mathrm I }$, $\check {\mathrm O }$, and $\check {\mathrm N }$, without orientations of edges.
}
\label{ION}
\end{minipage}
\end{figure}

\proof
In the proof, we ignore the tadpole embedding and the orientations of the edges, as they can be recovered paying attention to the definition of scheme isomorphism given in Definition \ref{def:isomorphism}. 
Additionally, we heavily rely on the planarity condition in the similar way as we did in the earlier proofs.

Consider a planar graph with $\ell=3$.
If an $\oD$-loop of length two is present, we can remove iteratively the $\oD$-loops of length two, while keeping $\ell$ and $g$ invariant, as described in Figure \ref{2ELremoval}. 
The resulting graph does not have any $\oD$-loops of length $2$ and hence should have $\oD$-loop configuration $(4,4,\ldots,4,6)$, as explained in 
Lemma \ref{l3loops}.
By reversing this argument we can conclude that every graph is constructed from $(4,4,\ldots,4,6)$ graphs by consecutively adding $\oD$-loops of length of $2$. All possible $\oD$-loops of length $4$ are discussed in the proof of Theorem \ref{thm:ell1ell2planar}.
We now study the $\oD$-loops of length $6$, which are found empirically (a different proof is addressed in Theorem \ref{prop:enumerate}).

An $\oD$-loop of length six contains at least three vertices. In this minimal case, there are three possibilities depicted in  Figure \ref{6EL3v}.
The first two graphs are similar to the ones displayed in Figure \ref{l2basic}. They can be constructed from the cycle graph by adding three tadpoles. The third graph stands out, as it does not have a tadpole.
\begin{figure}[H]
\centering
\begin{minipage}[t]{0.9\textwidth}
\centering
\def\svgwidth{0.8\columnwidth}
\tiny{
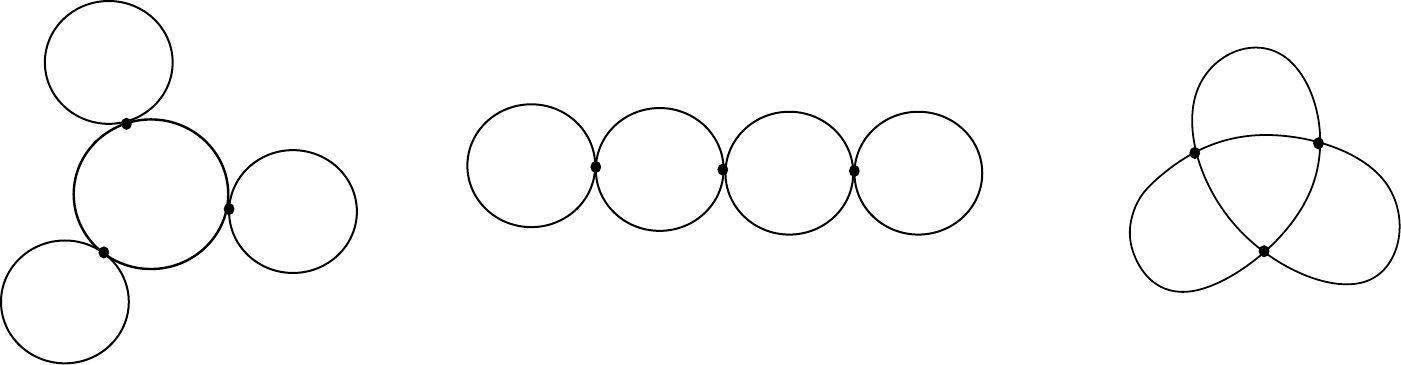
}
\caption{$\oD$-loops of length $6$ containing $3$ vertices.
}
\label{6EL3v}
\end{minipage}
\end{figure}

Secondly, suppose the $\oD$-loop contains four vertices. Then, the only possible graphs are the ones in Figure \ref{6EL4v}, up to tadpole embeddings.
Because they must contain a tadpole, they can be formed from the first two $\oD$-loops of length $4$ in Figure \ref{4EL} by adding a tadpole.
\begin{figure}[H]
\centering
\begin{minipage}[t]{0.9\textwidth}
\centering
\def\svgwidth{0.5\columnwidth}
\tiny{
\begingroup%
  \makeatletter%
  \providecommand\color[2][]{%
    \errmessage{(Inkscape) Color is used for the text in Inkscape, but the package 'color.sty' is not loaded}%
    \renewcommand\color[2][]{}%
  }%
  \providecommand\transparent[1]{%
    \errmessage{(Inkscape) Transparency is used (non-zero) for the text in Inkscape, but the package 'transparent.sty' is not loaded}%
    \renewcommand\transparent[1]{}%
  }%
  \providecommand\rotatebox[2]{#2}%
  \newcommand*\fsize{\dimexpr\f@size pt\relax}%
  \newcommand*\lineheight[1]{\fontsize{\fsize}{#1\fsize}\selectfont}%
  \ifx\svgwidth\undefined%
    \setlength{\unitlength}{313.3927521bp}%
    \ifx\svgscale\undefined%
      \relax%
    \else%
      \setlength{\unitlength}{\unitlength * \real{\svgscale}}%
    \fi%
  \else%
    \setlength{\unitlength}{\svgwidth}%
  \fi%
  \global\let\svgwidth\undefined%
  \global\let\svgscale\undefined%
  \makeatother%
  \begin{picture}(1,0.41142206)%
    \lineheight{1}%
    \setlength\tabcolsep{0pt}%
    \put(0,0){\includegraphics[width=\unitlength,page=1]{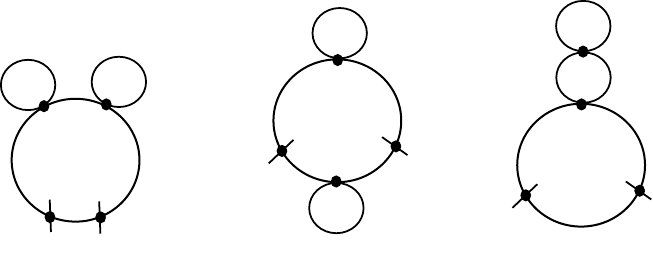}}%
    \put(0.09691142,0.00660267){\color[rgb]{0,0,0}\transparent{0.98699999}\makebox(0,0)[lt]{\lineheight{1.25}\smash{\begin{tabular}[t]{l}\scalebox{1}{$(a)$}\end{tabular}}}}%
    \put(0.49595836,0.00561669){\color[rgb]{0,0,0}\transparent{0.98699999}\makebox(0,0)[lt]{\lineheight{1.25}\smash{\begin{tabular}[t]{l}\scalebox{1}{$(b)$}\end{tabular}}}}%
    \put(0.87600307,0.00751689){\color[rgb]{0,0,0}\transparent{0.98699999}\makebox(0,0)[lt]{\lineheight{1.25}\smash{\begin{tabular}[t]{l}\scalebox{1}{$(c)$}\end{tabular}}}}%
  \end{picture}%
\endgroup%

}
\caption{$\oD$-loops of length $6$ containing $4$ vertices.
}
\label{6EL4v}
\end{minipage}
\end{figure}

Thirdly, suppose the $\oD$-loop contains five vertices. Then, when the $\oD$-loop contains a tadpole, it can be formed by adding a tadpole to the $\oD$-loop on the right of Figure \ref{4EL}. The resulting $\oD$-loop is displayed on the left, i.e., (a) in Figure \ref{6EL5v}. 
When there are no tadpoles present, the only possibilities are given in the middle (b) and on the right of Figure \ref{6EL5v}. 
In the case for planar graphs with $\ell=3$, these latter two subgraphs are equivalent.

\begin{figure}[H]
\centering
\begin{minipage}[t]{0.9\textwidth}
\centering
\def\svgwidth{0.6\columnwidth}
\tiny{
\begingroup%
  \makeatletter%
  \providecommand\color[2][]{%
    \errmessage{(Inkscape) Color is used for the text in Inkscape, but the package 'color.sty' is not loaded}%
    \renewcommand\color[2][]{}%
  }%
  \providecommand\transparent[1]{%
    \errmessage{(Inkscape) Transparency is used (non-zero) for the text in Inkscape, but the package 'transparent.sty' is not loaded}%
    \renewcommand\transparent[1]{}%
  }%
  \providecommand\rotatebox[2]{#2}%
  \newcommand*\fsize{\dimexpr\f@size pt\relax}%
  \newcommand*\lineheight[1]{\fontsize{\fsize}{#1\fsize}\selectfont}%
  \ifx\svgwidth\undefined%
    \setlength{\unitlength}{388.98067925bp}%
    \ifx\svgscale\undefined%
      \relax%
    \else%
      \setlength{\unitlength}{\unitlength * \real{\svgscale}}%
    \fi%
  \else%
    \setlength{\unitlength}{\svgwidth}%
  \fi%
  \global\let\svgwidth\undefined%
  \global\let\svgscale\undefined%
  \makeatother%
  \begin{picture}(1,0.26465043)%
    \lineheight{1}%
    \setlength\tabcolsep{0pt}%
    \put(0,0){\includegraphics[width=\unitlength,page=1]{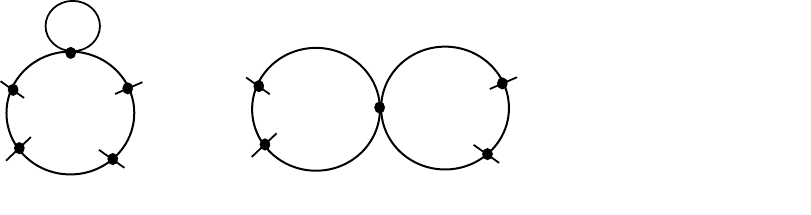}}%
    \put(0.0644527,0.0065421){\color[rgb]{0,0,0}\transparent{0.98699999}\makebox(0,0)[lt]{\lineheight{1.25}\smash{\begin{tabular}[t]{l}\scalebox{1}{$(a)$}\end{tabular}}}}%
    \put(0.45022604,0.00574771){\color[rgb]{0,0,0}\transparent{0.98699999}\makebox(0,0)[lt]{\lineheight{1.25}\smash{\begin{tabular}[t]{l}\scalebox{1}{$(b)$}\end{tabular}}}}%
    \put(0,0){\includegraphics[width=\unitlength,page=2]{6EL5v.pdf}}%
  \end{picture}%
\endgroup%

}
\caption{$\oD$-loops of length $6$ containing $5$ vertices. 
}
\label{6EL5v}
\end{minipage}
\end{figure}

Finally, suppose the $\oD$-loop contains six vertices.
Then the only possible $\oD$-loop one can form is the subgraph drawn in 
Figure 
\ref{6EL6v}.
\begin{figure}[H]
\centering
\begin{minipage}[t]{0.9\textwidth}
\centering
\def\svgwidth{0.12\columnwidth}
\tiny{
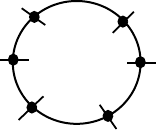
}
\caption{An $\oD$-loop of length $6$ containing $6$ vertices.
}
\label{6EL6v}
\end{minipage}
\end{figure}

We now have described all the possible configurations for $\oD$-loops with length $6$. The graphs in question are constructed by adding $\oD$-loops of length $4$ to the $\oD$-loop with length $6$, as was done in the case for $\ell =2$ in Figure \ref{l2procedure}. 
$X$ in the following figures (\ref{6EL4vcont}, \ref{6EL5vcont},  and \ref{6EL6vcont}), 
denotes 
any non-empty collection of elements from $\{\check {\mathrm I }, \check {\mathrm O }\}$ in any amount, combination, and sequence.
We distinguish three cases. 
\\

(I)
Suppose that the $\oD$-loop of length six is of the forms as in Figure \ref{6EL4v}. Then the procedure of adding $\oD$-loops of length $4$ is identical to the case where $\ell=2$, as described in Figure \ref{l2procedure} as an example.
As discussed earlier, once orientations of edges are assigned, it results in a ladder with only L- and/or R-dipoles ending in a tadpole. 
The possible graphs, before assigning orientations of edges, are given in Figure \ref{6EL4vcont}.
\begin{figure}[H]
\centering
\begin{minipage}[t]{0.9\textwidth}
\centering
\def\svgwidth{1\columnwidth}
\tiny{
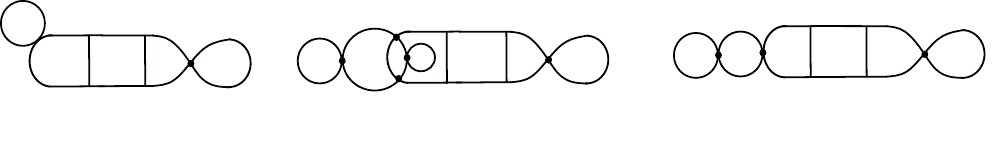
}
\caption{ 
Graphs ignoring tadpole embeddings and the orientations on edges with $\ell=3$ and $g=0$ and $\oD$-loop structure  $(4,4,\ldots,4,6)$, where the $\oD$-loop of length $6$ contains $4$ vertices. 
(a'), (b'), and (c') graphs are derived from (a), (b), and (c) respectively in Figure \ref{6EL4v}.
}
\label{6EL4vcont}
\end{minipage}
\end{figure}

(II)
Suppose that the $\oD$-loop of length six is of the forms as in Figure \ref{6EL5v}. Then the procedure of adding $\oD$-loops of length $4$ is nearly identical to the case where $\ell=2$, except that now we have four free vertices rather than two.
Note that we can only add an $\oD$-loop in a certain way, namely to pair the blue with blue and the red with red vertices without mixing them 
(as labeled in Figure \ref{6EL5vmore} and Figure \ref{6EL5vcont})
due to the planarity condition (see in Figure \ref{6EL5vforbidden} a forbidden way of inserting an $\oD$-loop in (b) of Figure \ref{6EL5v}).
\begin{figure}[H]
\centering
\begin{minipage}[t]
{0.45\textwidth}
\centering
\def\svgwidth{0.25\columnwidth}
\tiny{
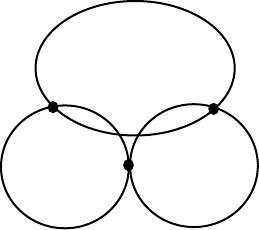
}
\caption{
This graph is non-planar.
}
\label{6EL5vforbidden}
\end{minipage}
\end{figure}
Therefore, respecting the planarity condition (i.e., pairing only blue (resp. red) with blue (resp. red) vertices), either one can add only $\oD$-loops of length four each without a tadpole that make a closed chain 
(see Figure \ref{6EL5vmore} and the left column of the resulting graphs in Figure \ref{6EL5vcont}), or one performs the procedure explained for Figure \ref{l2procedure}, ending up with creating ladders containing only  L- or R-dipoles ending with a tadpole. 
The resulting latter graphs can be found on the right column of the resulting graphs presented in Figure \ref{6EL5vcont}.
\begin{figure}[H]
\centering
\begin{minipage}[t]{1\textwidth}
\centering
\def\svgwidth{1\columnwidth}
\tiny{
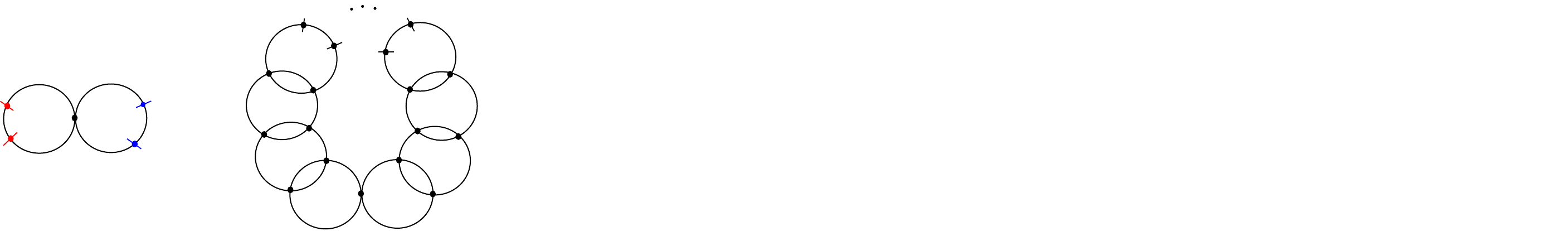
}
\caption{
Adding only $\oD$-loops of length four each without a tadpole to the $\mathrm{O}(D)$-loop structure (b) in Figure \ref{6EL5v} makes a closed chain.
The resulting graphs after recovering the orientations on edges are either two of the right side, corresponding to the two schemes above and they are identified to be 2PI.
}
\label{6EL5vmore}
\end{minipage}
\end{figure}
\begin{figure}[H]
\centering
\begin{minipage}[t]{0.9\textwidth}
\centering
\def\svgwidth{0.8\columnwidth}
\tiny{
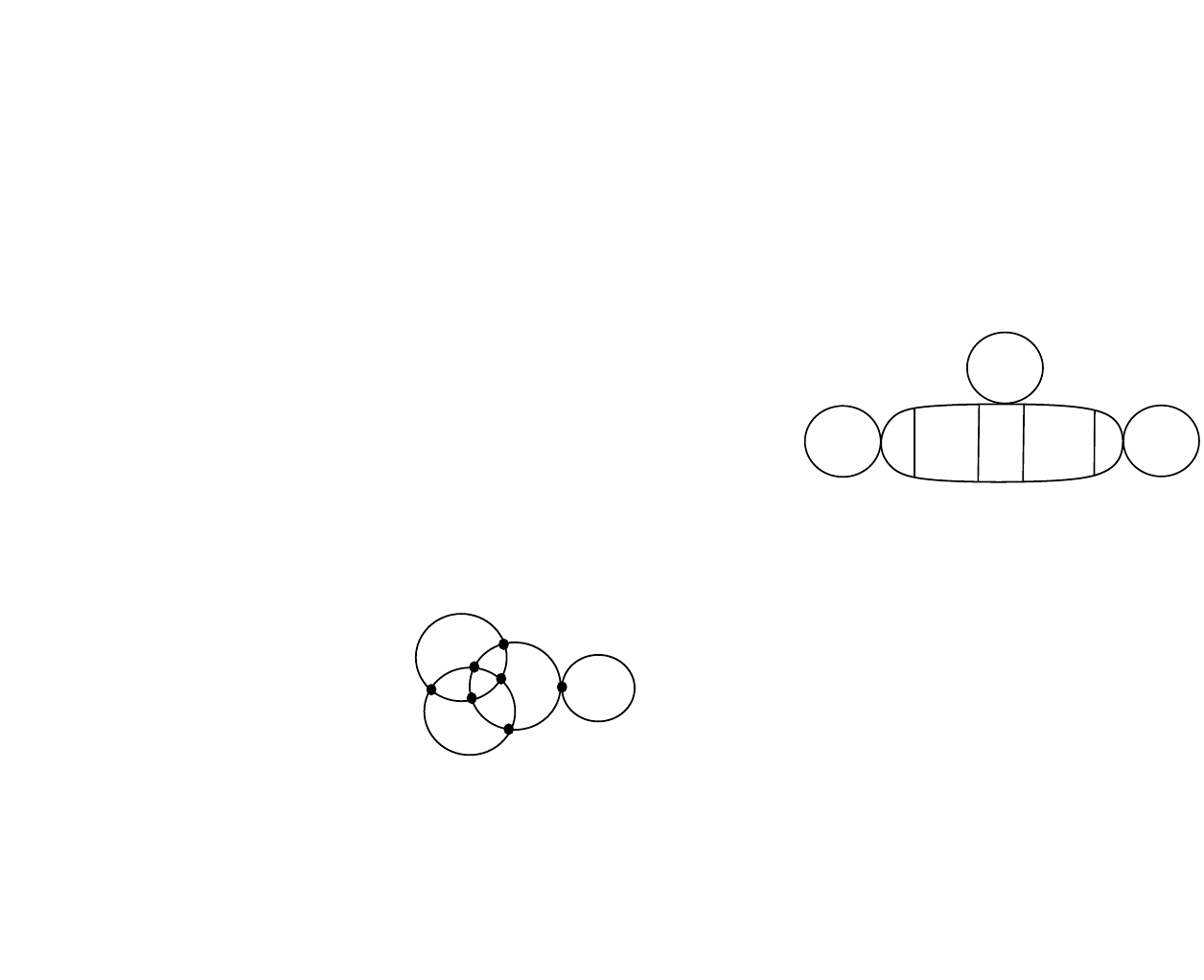
}
\caption{
Graphs (without tadpole embeddings nor orientation assignment on edges) with $\ell=3$ and $g=0$ and $\oD$-loop structure  $(4,4,\ldots,4,6)$, where the $\oD$-loop of length $6$ contains $5$ vertices. 
All cases presented here are 2PR.
Colors indicate the pairing of vertices belonging to the same $\mathrm{O}(D)$-loop when they are added in the procedure described in the text.
On the top two rows, the drawing with $\check{\mathrm I}$'s indicates that there are two or more $\check{\mathrm I}$'s in sequence. 
Depending on the orientations of the edges, it will either become L- or R-ladder vertex.
}
\label{6EL5vcont}
\end{minipage}
\end{figure}

(III)
Suppose now the case where the $\oD$-loop of length $6$ is as in Figure \ref{6EL6v}. 
Then, there are only three possibilities,
to add $\oD$-loops of length $4$. 
In Figure \ref{6EL6vcolour}  we indicate the vertices and corresponding half-edges that can be contained in the same $\oD$-loop of length $4$ with the same color.

\begin{figure}[H]
\centering
\begin{minipage}[t]{0.9\textwidth}
\centering
\def\svgwidth{0.5\columnwidth}
\tiny{
\begingroup%
  \makeatletter%
  \providecommand\color[2][]{%
    \errmessage{(Inkscape) Color is used for the text in Inkscape, but the package 'color.sty' is not loaded}%
    \renewcommand\color[2][]{}%
  }%
  \providecommand\transparent[1]{%
    \errmessage{(Inkscape) Transparency is used (non-zero) for the text in Inkscape, but the package 'transparent.sty' is not loaded}%
    \renewcommand\transparent[1]{}%
  }%
  \providecommand\rotatebox[2]{#2}%
  \newcommand*\fsize{\dimexpr\f@size pt\relax}%
  \newcommand*\lineheight[1]{\fontsize{\fsize}{#1\fsize}\selectfont}%
  \ifx\svgwidth\undefined%
    \setlength{\unitlength}{278.25166237bp}%
    \ifx\svgscale\undefined%
      \relax%
    \else%
      \setlength{\unitlength}{\unitlength * \real{\svgscale}}%
    \fi%
  \else%
    \setlength{\unitlength}{\svgwidth}%
  \fi%
  \global\let\svgwidth\undefined%
  \global\let\svgscale\undefined%
  \makeatother%
  \begin{picture}(1,0.27342055)%
    \lineheight{1}%
    \setlength\tabcolsep{0pt}%
    \put(0,0){\includegraphics[width=\unitlength,page=1]{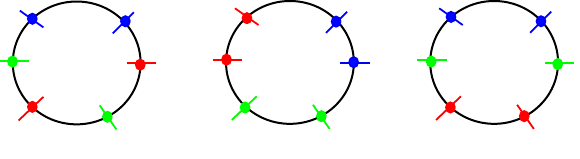}}%
    \put(0.11210742,0.00803501){\color[rgb]{0,0,0}\transparent{0.98699999}\makebox(0,0)[lt]{\lineheight{1.25}\smash{\begin{tabular}[t]{l}\scalebox{1}{$(a)$}\end{tabular}}}}%
    \put(0.48557365,0.01231538){\color[rgb]{0,0,0}\transparent{0.98699999}\makebox(0,0)[lt]{\lineheight{1.25}\smash{\begin{tabular}[t]{l}\scalebox{1}{$(b)$}\end{tabular}}}}%
    \put(0.84191819,0.01445557){\color[rgb]{0,0,0}\transparent{0.98699999}\makebox(0,0)[lt]{\lineheight{1.25}\smash{\begin{tabular}[t]{l}\scalebox{1}{$(c)$}\end{tabular}}}}%
  \end{picture}%
\endgroup%

}
\caption{Possible ways of adding $\oD$-loops of length $4$ to this particular $\oD$-loop of length $6$.
}
\label{6EL6vcolour}
\end{minipage}
\end{figure}

In the first case, i.e., (a) in Figure \ref{6EL6vcolour},
there is only one possible configuration for the $\oD$-loops added at the red and green half-edges. 
We then add $\oD$-loops of length $4$ to the the blue half-edges, following the same procedure as described for Figure \ref{l2procedure}.
We end up with (3) in Figure \ref{6EL6vcont}.

In the second case, i.e., (b) depicted in Figure \ref{6EL6vcolour}, we carry out the procedure of adding $\oD$-loops of length $4$ to the red, green and blue half-edges respectively. 
We have two possibilities.
First possibility is each procedure ends by adding an  $\oD$-loop of length $4$ with a tadpole. 
Then we end up with (4) in Figure \ref{6EL6vcont}. 
Otherwise, one of the procedures (say, starting with red) ends with an $\oD$-loop with a tadpole, while the other two procedures (say, starting at green, and at blue) connect 
together to make a closed chain, similar as to what was depicted in Figure \ref{l2noTP}. 
We then end up with (2) in Figure \ref{6EL6vcont}.

Lastly, the third case, i.e., (c) in Figure \ref{6EL6vcolour} is slightly more subtle. 
We again perform the procedure of adding $\oD$-loops of length $4$ to the red, green and blue half-edges respectively. 
Again, each procedure of adding $\oD$-loops of length $4$ can end by adding an $\oD$-loop with a tadpole. 
Then we end up with (5) in Figure \ref{6EL6vcont}. 
Now we consider how two procedures of adding $\oD$-loops of length $4$ starting at two different sets of colored half-edges can connect to make a closed chain as in the previous case. 
In this case we are more restricted how such a closed chain can occur. 
Indeed, the procedure of adding $\oD$-loops starting at the blue half-edges cannot be connected together with the procedure starting at the red half-edges, as this would make it impossible to add an  $\oD$-loop to the green half-edges while keeping the graph planar. 
However, the procedure of adding $\oD$-loops starting at the green half-edges can be connected together with the procedure started at either the red or blue half-edges. 
One can check we then end up with a scheme equivalent to (1) in Figure \ref{6EL6vcont}.

\begin{figure}[H]
\centering
\begin{minipage}[t]{0.95\textwidth}
\centering
\def\svgwidth{1\columnwidth}
\tiny{
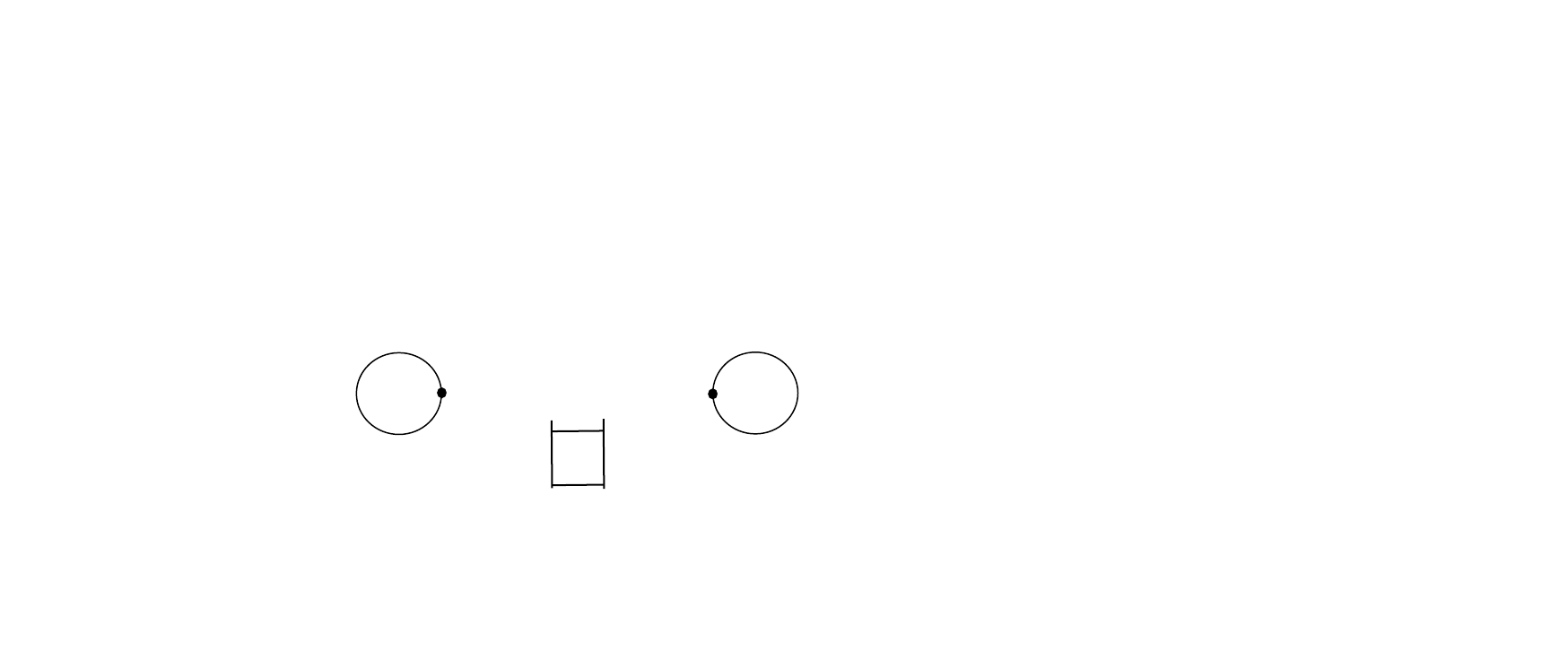
}
\caption{
Graphs (without tadpole embeddings nor orientations on edges) with $\ell=3$ and $g=0$ and $\oD$-loop structure  $(4,4,\ldots,4,6)$, where the $\oD$-loop of length $6$ contains $6$ vertices. 
}
\label{6EL6vcont}
\end{minipage}
\end{figure}

The next step of this proof is similar to the proof for $\ell=2$.
As discussed before, one can construct general planar graphs with $\ell=3$ from the $(4,4,\ldots,4,6)$ graphs by adding $\oD$-loops of length $2$ that do not change the genus. In the case where there are no tadpoles present, such $\oD$-loops will be melons. In the case where tadpoles are present, adding such non-melonic $\oD$-loops, amounts to adding N-dipoles to the, possibly empty, ladder.

To summarize, all the 2PR schemes with $g=0$ and $\ell=3$, ignoring the orientations of the edges, are given in Figure \ref{l3all2PR}.
Recovering the orientations of the edges for 2PI schemes we list all the 2PI schemes for $g=0$ and $\ell=3$ in Figure \ref{l32PI}.

\qed

\medskip

\section{Toward classification of $\ell> 3$ with $g=0$} 
\label{sec:higherlg0}

This section will cover the enumeration of 2PI melon-free graphs with grade $\ell$ greater than $3$. More precisely, we find all graphs with $\oD$-loop configurations $(4, 4, \dots, 4, 2\ell)$ for $\ell=4,5, 6$. 
The illustrations are provided in figures, read from left to right as the first, second, third (etc.) graph.
We finally deduce how to enumerate the 2PI melon-free graphs with a single $\oD$-loop and arbitrary grade after connecting to alternating knot diagrams.

Let us denote by $p_i$ (an even number) and $p(n,m)=(p_1, \ldots, p_m)$ a partition of $n$ in $m$ elements with $n=\sum_{i=1}^m p_i$ and by $p(n,m)^{>k}$ a partition in which each $p_i$ is greater than $k$. 
\begin{thm}
\label{thm:remi}
Let $G$ be a graph with grade $\ell$ and $\varphi$ number of $\oD$-loops.
\begin{enumerate}
    \item If $\ell>\varphi + 2$ then $G$ has $\oD$-loop configurations $p(2(l+2i),i+1)$; $i=0, \ldots, \ell-4$;
    \item Else, $G$ contains an $\oD$-loop of length $2$ or has $\oD$-loop configurations of the form $(4, 4, \ldots, 4, p(2(\ell+2(i-1)),i)^{>2})$; $i=1, \ldots, \ell-2$.
\end{enumerate}
\end{thm}
\begin{proof}
\begin{enumerate}
    \item For $\ell>\varphi + 2$, $\varphi$ is clearly at most $\ell-3$ and therefore any $\oD$-loop configuration is a partition of $E$ in 1, 2, $\ldots$ or $\ell-3$ elements. Then $\varphi=i+1$ in the first equation in \eqref{eq:seven} gives $E=2v=2\ell-4+4(i+1)=2(\ell+2i)$ ending the proof of the first item.

    \item We now assume that $\ell\leq \varphi + 2$ and $G$ contains no $\oD$-loop of length $2$. The smallest value of $E$ is obtained if every $\oD$-loop but the fewest possible has length $4$. The maximum number $m$ of $\oD$-loops with length greater than $4$ is at most $\ell-2$. Otherwise if $m>\ell-2$ then $E=4(\varphi-m)+\sum_{t=1}^ml_i$ with $l_i\geq 6$ for all $i$. Then $E=2v\geq 4(\varphi-m)+6m=4\varphi+2m>4+2v-2\ell+2\ell-4=2v$ which gives a contradiction. 
    
    If the $\oD$-loop configuration 
   has only one element $p_k >4$, then it is of the form $(4, 4, \ldots, 4, 2\ell)$. In fact $E=2v=4(\varphi-1)+k=4(\frac{v}{2}-\frac{\ell}{2})+k=2v-2\ell+k$ and therefore $k=2\ell$. More generally, if there are $1\leq i\leq \ell-2$ elements; $l_1, \ldots, l_i$ of length greater than $4$, then $\sum_{t=1}^il_t=2(\ell+2(i-1))$. In fact, $E=2v=4(\varphi-i)+\sum_{t=1}^il_t=4(1+\frac{v}{2}-\frac{\ell}{2}-i)+\sum_{t=1}^il_t=4+2v-2\ell-4i+\sum_{t=1}^il_t$ implies that $\sum_{t=1}^il_t=2(\ell+2(i-1))$.
\end{enumerate}
\end{proof}

From now on, we skip the orientation in the enumeration of the schemes for simplicity.
One can recover the information encoded in the orientation by assigning orientations of edges in a graph and making sure that no graphs are isomorphic by following Definition \ref{def:isomorphism}.

In the following, we show the existence of these graphs that we claimed in the Theorem \ref{thm:remi}.

Let us now discuss the particular cases of $\ell=4,5,6$ in which we only address the $\oD$-loops of length greater than two.

{\bf Case $\ell=4$:}
In case $\ell>\varphi+2$, then $\varphi=1$ and the only $\oD$-loop configuration is $(8)=p(2\times 4,1)$. Otherwise the only $\oD$-loop configurations are $(4, 4, \ldots, 4, p(2(4+2(i-1)),i)^{>2})$; $i=1,2$ namely $(4, 4, \ldots, 4, 8)$
and $(4, 4, \ldots,4, 6, 6)$. A graphical representations of these is given in Figures \ref{l4phi1} and \ref{l4phi2}.

\begin{figure}[H]
\centering
\includegraphics[width=0.6\textwidth]{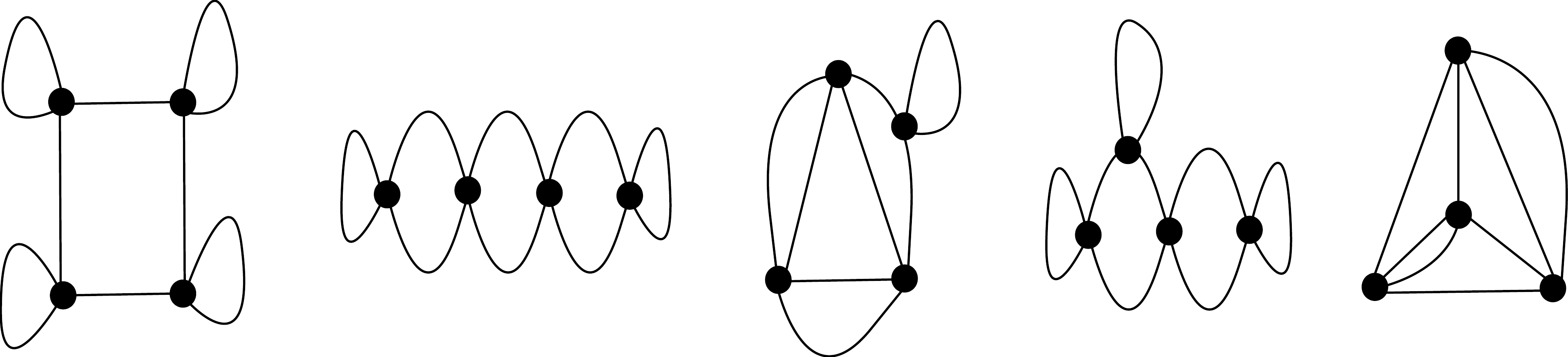}
\caption{\small{Some schematic representation of $\ell=4$ and $\varphi=1$ graphs.}}
\label{l4phi1}
\end{figure}

\begin{figure}[H]
\centering
\includegraphics[width=0.6\textwidth]{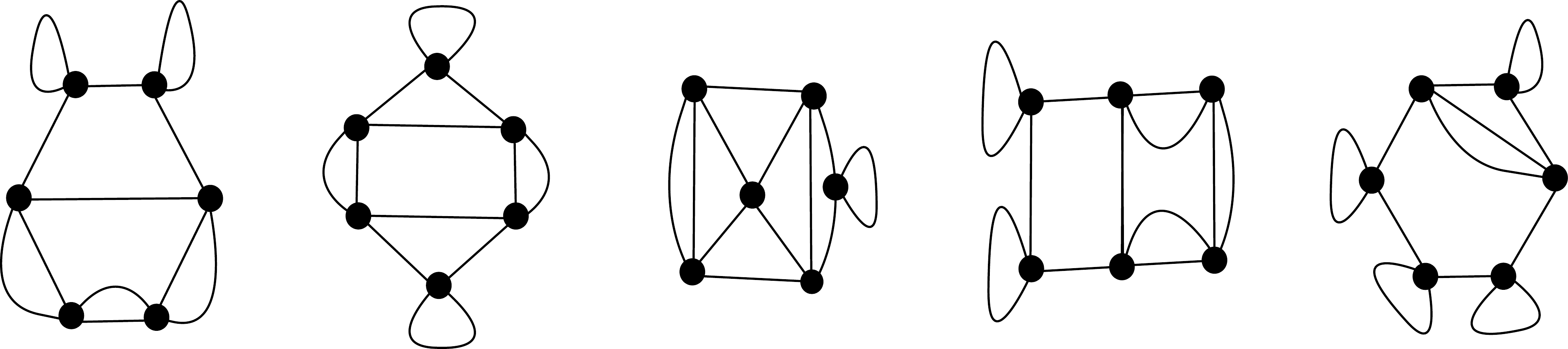}
\caption{\small{Some schematic representations of $\ell=4$ and $\varphi=2$ graphs.}}
\label{l4phi2}
\end{figure}

{\bf Case $\ell=5$:}
For $\ell>\varphi+2$, then $\varphi=1,2$ and the case $\varphi=1$ gives the only $\oD$-loop configuration is $(10)$ while $\varphi=2$ gives two $\oD$-loop configurations of $(6,8)$ and $(4,10)$.
Assume that $\ell\leq\varphi+2$. The only $\oD$-loop configurations are $(4, 4, \ldots, 4, 10)$, $(4, 4, \ldots, 4, 6, 8)$
and $(4, 4, \ldots,4, 6, 6, 6)$. An illustration for the existence of $\oD$-loop configurations of the form $(10)$ are given in Figure \ref{l5phi1} and those of $(6,8)$ and $(4,6,8)$ are respectively the first two figures in \ref{l56conf}.

\begin{figure}[H]
\centering
\includegraphics[width=0.6\textwidth]{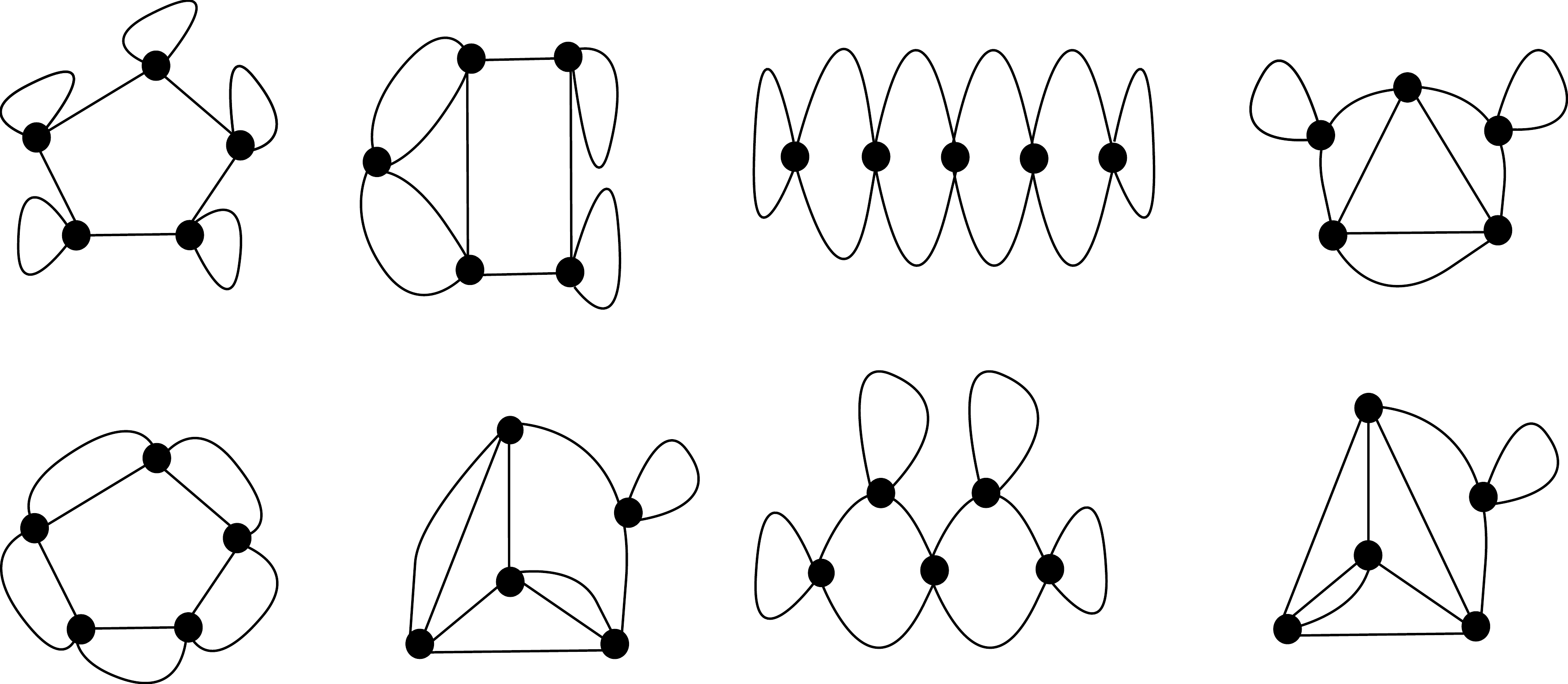}
\caption{\small{Some schematic representations of $\ell=5$ and $\varphi=1$ graphs.}}
\label{l5phi1}
\end{figure}

{\bf Case $\ell=6$:}
Assuming that $\ell>\varphi+2$, $\varphi=1,2,3$ and the case $\varphi=1$ gives the only $\oD$-loop configuration of $(12)$. The case $\varphi=2$ gives three $\oD$-loop configurations of $(6,10)$, $(8,8)$ and $(4,12)$ while $\varphi=3$ gives $(4,4,12)$, $(4,6,10)$, $(4,8,8)$ and $(6,6,8)$.
Suppose that $\ell\leq\varphi+2$. The only $\oD$-loop configurations are $(4, 4, \ldots, 4, 12)$, $(4, 4, \ldots, 4, 6, 10)$, $(4, 4, \ldots, 4, 8, 8)$ and $(4, 4, \ldots, 4, 6, 6, 8)$. An illustration for the existence of $(6,10)$, $(6,6,8)$ are the third  and forth graphs in Figure \ref{l56conf}.

\begin{figure}[H]
\centering
\includegraphics[width=0.6\textwidth]{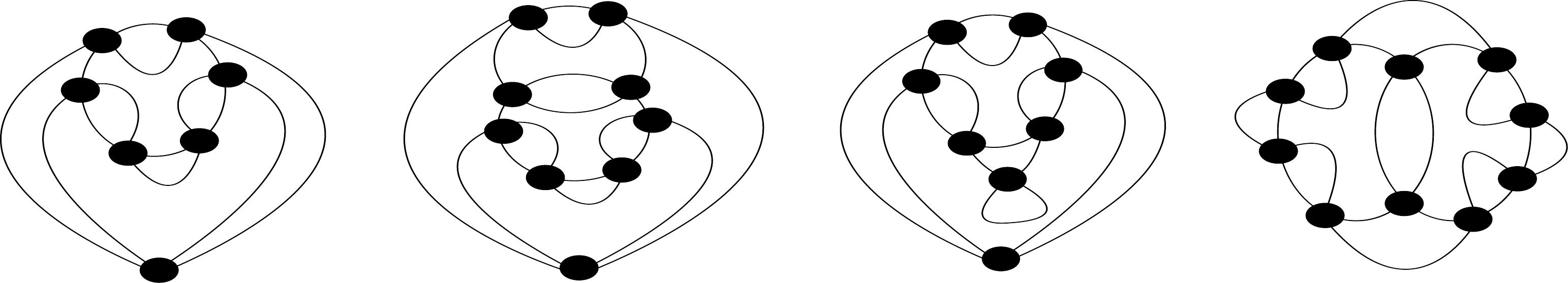}
\caption{\small{Schematic representations of the configurations $(6,8)$ and $(4,6,8)$ for $\ell=5$ followed by $(6,10)$ and $(6,6,8)$ for $\ell=6$ (reading from left to right).}}
\label{l56conf}
\end{figure}

\begin{remark}
The proof of constructing graphs with  $\oD$-loop configurations $(6)$ and $(4,4,\dots,6)$  for $l=3$ is performed differently as compared to the proof in Theorem \ref{thm:ell3g02PI} for the 2PI graphs in the following Theorem \ref{prop:enumerate}. 
This proof in Theorem\ref{prop:enumerate} can also be extended to generate 2PR graphs and is easily generalizable to higher $\ell >3$.
\end{remark}

\begin{definition}[\emph{Necklace} graph]
A \emph{necklace} graph with $n$ vertices is the cyclic graph with $n$ vertices in which every edge is doubled.
\end{definition}
An illustration of a necklace graph with 5 vertices is given by the leftmost graph in the bottom row in Figure \ref{l5phi1}.

\begin{thm}
\label{prop:enumerate}
    Let $G$ be a planar connected 2PI melon-free Feynman graph with grade $\ell$. 
    \begin{enumerate}
    \item If $\ell=3$, $G$ has $\oD$-loop configuration $(6)$ whose only representation is the necklace graph or the configuration $(4,4,\dots, 4, 6)$ which form is illustrated on the left hand side of Figure \ref{fig:49}. 
    
    \item If $\ell= 4$, the graph $G$ has $\oD$-loop configuration $(8)$ (given by the graph 2 in Figure \ref{l45phi1}) or $(4,4,\dots, 4, 8)$  with the form illustrated on the right hand side of Figure \ref{fig:49}. 
    
    \item If $\ell= 5$ and the graph $G$ has $\oD$-loop configuration $(10)$, then it is given by the necklace graph or the graph 4 in Figure \ref{l45phi1}.

    \item 
    If $\ell= 6$ and the graph $G$ has $\oD$-loop configuration $(12)$, then it is given by the second graph in Figure \ref{facesix}, 
    the three graphs in Figure \ref{facefive},
    the third graph in Figure \ref{facefour1}, 
or
the second graph in Figures \ref{facefour2} and \ref{facefour3}.
    \end{enumerate}
    \begin{figure}[H]
\centering
\includegraphics[width=60mm]{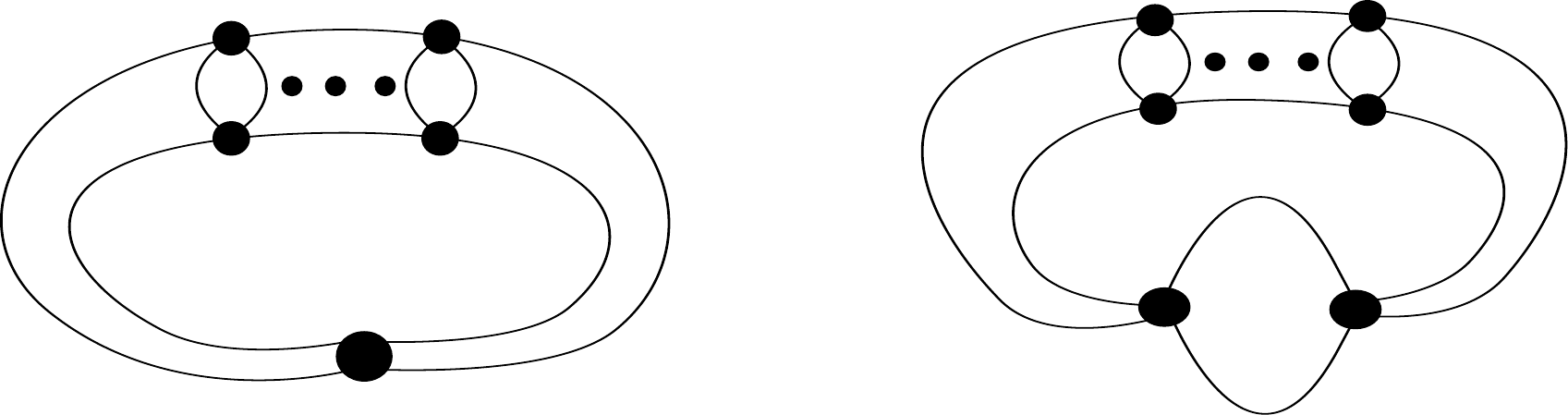}
\caption{\small{Graphs with loop configurations $(4, 4, \ldots, 4, 6)$ and $(4, 4, \ldots, 4, 8).$
}}
\label{fig:49}
\end{figure}
\end{thm}

\begin{proof}
\begin{enumerate}
\item Assume that $\ell= 3$ and $g = 0$. The $\oD$-loop configurations in this case are $(6)$, or $(4, 4, \dots, 4, 6)$. 

We can demonstrate that the only graph with an $\oD$-loop configuration of $(6)$ is the necklace graph. Consider the graph $G$ with $V$, $E$, and $F$ representing the number of vertices, edges, and faces (left or right faces in this context), respectively. Let $v$ and $f$ denote a vertex and a face of $G$, respectively. Clearly, $\sum_v \text{deg}(v) = 2E$ implies $E = 2V$ since the graph is 4-regular. The graph $G$ must contain a triangular face; otherwise, all faces would be $2$-gons, as $G$ is a tadpole-free graph. Assuming all faces are $2$-gons leads to $\sum_f \text{deg}(f) = 2F = 2E$. Using Euler's formula, we get $V - 2V + 2V = 2$, which simplifies to $V = 2$. However, this is incorrect because $G$ has three vertices (substituting $\ell = 3$ in \eqref{eq:seven}). There is only one way to connect the two half-edges attached to each vertex of this triangle, resulting in the necklace graph.

Now, let's examine the $\oD$-loop configuration of the form $(4, 6)$. In this case, the graph has 5 vertices. The $4$-edge $\oD$-loop cannot connect only two vertices, as this would result in a self-loop and a disconnected graph. Similarly, the $4$-edge $\oD$-loop cannot involve only three vertices without forming a tadpole or a self-loop. Therefore, the only feasible configuration is that the $4$-edge $\oD$-loop connects 4 vertices, as illustrated in Figure \ref{l3phi4}.

\begin{figure}[H]
\centering
\includegraphics[width=30mm]{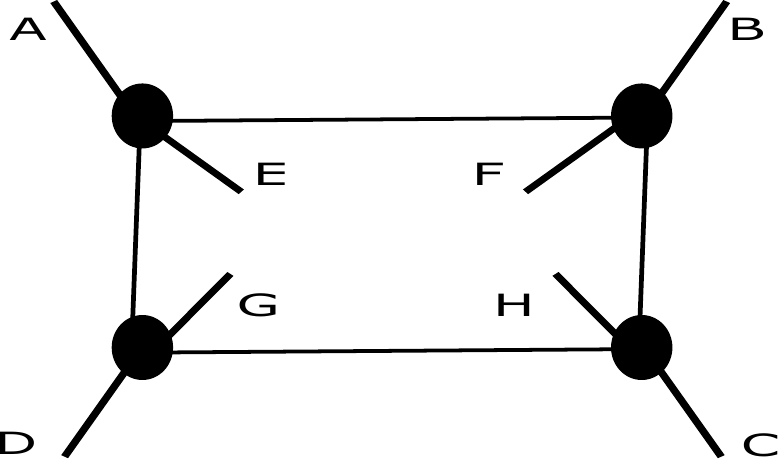}
\caption{\small{Representation of a $4$-edge $\oD$-loop.}}
\label{l3phi4}
\end{figure}

The fifth vertex can only be connected to the half-edges $A$, $B$, $C$, and $D$ or to the half-edges $E$, $F$, $G$, and $H$, with both cases resulting in isomorphic graphs. Assume this vertex is connected to the half-edges $A$, $B$, $C$, and $D$. The half-edge $E$ can only be connected to $F$ or $G$; otherwise, the resulting graph would not be planar. The same constraint applies to the other half-edges $F$, $G$, and $H$. The final configuration is depicted in Figure \ref{l3sixfour}, thus concluding the proof for the $(4, 6)$ configuration.

\begin{figure}[H]
\centering
\includegraphics[width=30mm]{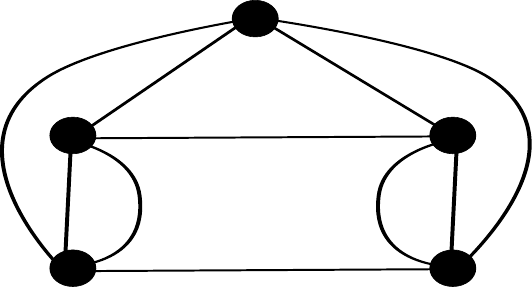}
\caption{Representation of the configuration $(4, 6)$.}\label{l3sixfour}
\end{figure}

Now, let's consider the configuration $(4, 4, 6)$. In this case, the graph has 7 vertices, and a $4$-edge loop is formed as shown in Figure \ref{l3phi4}. The second $4$-edge loop in this configuration cannot be formed using only the four vertices of this graph, as this would result in disconnected graphs. Additionally, any $\oD$-loop must involve an even number of these vertices. Therefore, the second $4$-edge $\oD$-loop can only be created by using two of the vertices from the previous $4$-edge $\oD$-loop, resulting in the two representations shown in Figure \ref{l3sixfourn}.
\begin{figure}[H]
\centering
\includegraphics[width=80mm]{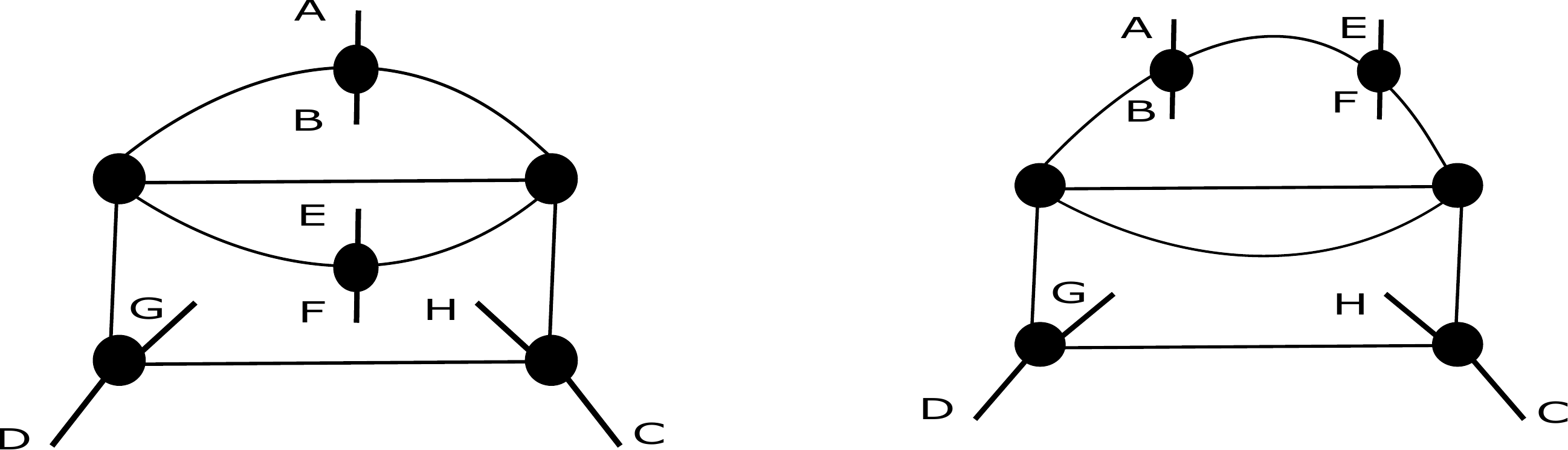}
\caption{\small{Schematic representations toward the configuration $(4, 4, 6)$.
}}
\label{l3sixfourn}
\end{figure}
The representation on the left-hand side is impossible because the half-edges $B$ and $E$ cannot be connected to any other half-edges without creating a non-planar graph. Therefore, we consider the representation on the right-hand side, which implies that $B$ should be joined to $F$ and $G$ to $H$. The remaining vertex will then connect to the four remaining half-edges, resulting in the expected graph.

Next, we turn our attention to the configuration $(4, 4, 4, 6)$. Following the previous proof, there are two $4$-edge $\oD$-loops that are connected, as shown in the graph on the right-hand side of Figure \ref{l3sixfourn}. Otherwise, the three $4$-edge $\oD$-loops would be separated, requiring more than 9 vertices. The third $4$-edge loop can only be constructed by using two new vertices and the half-edges $A$, $B$, $E$, and $F$ or $D$, $C$, $H$, and $G$. Finally, the remaining four half-edges in the resulting graph will be connected to the last vertex.

This recursive procedure can be extended for a general configuration of the form $(4, 4, \dots, 4, 6)$.

\item Assume that $\ell = 4$. If there is a face of length 4, it is evident that the graph either contains a tadpole or more than one $\oD$-loop. Proceeding similarly to the $\ell = 3$ case, we find that the graph has a triangular face. Since it is a connected graph, it takes the form of Graph 1 in Figure \ref{l45phi1}. The half-edge $C$ in this figure can only be connected to $B$ or $H$; otherwise, we would create a tadpole or place an odd number of half-edges in a face, violating the graph's planarity. If $C$ is connected to $B$, then the connections are as follows: $A$ to $D$, $E$ to $F$, and $H$ to $G$, resulting in Graph 2 in Figure \ref{l45phi1}. If $C$ is connected to $H$, then $B$ can only be connected to $A$, and the connections $E$ to $F$ and $D$ to $G$ follow to avoid creating a tadpole or having an odd number of half-edges in a face. It is clear that both configurations result in isomorphic graphs. A similar recursive procedure to the $\ell = 3$ case yields graphs with configurations $(4, \ldots, 4, 8)$.

\begin{figure}[H]
\centering
\includegraphics[width=120mm]{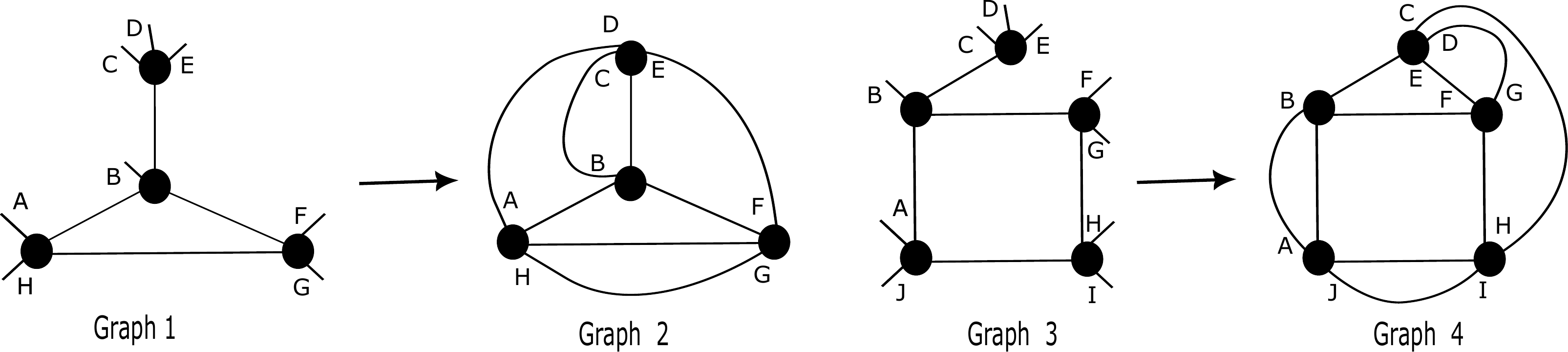}
\caption{\small{Schematic representations of $\ell=4, 5$ and $\varphi=1$ graph. 
}}
\label{l45phi1}
\end{figure}

\begin{figure}[H]
\centering
\includegraphics[width=90mm]{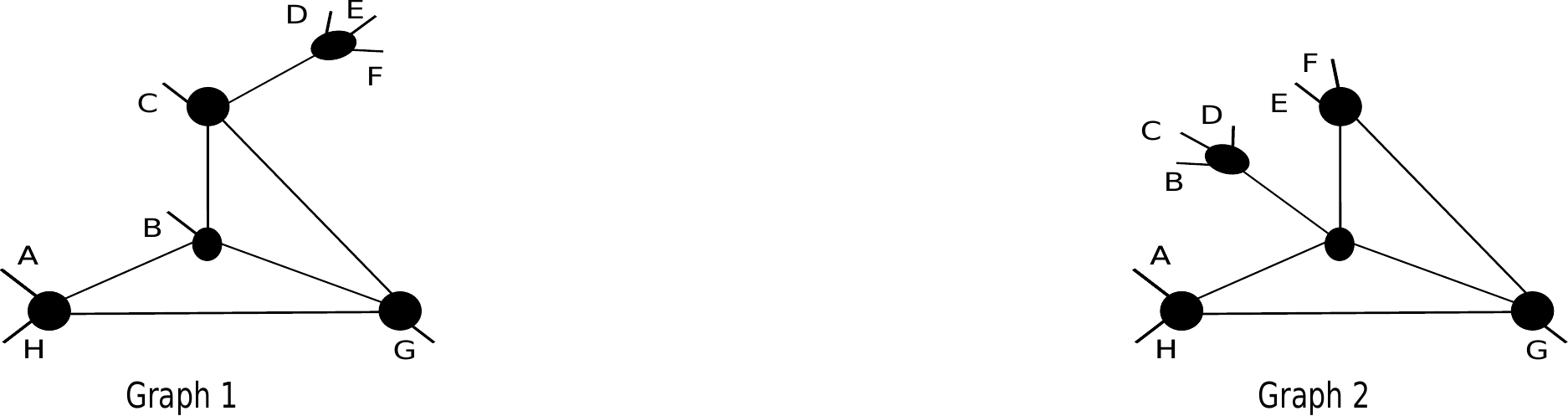}
\caption{\small{Representations with face of length at least 3; $\ell=4$ and $\varphi=1$.
}}
\label{noloopfive}
\end{figure}

\item Now, let's consider the case $\ell = 5$. Assume that the graph $G$ has only faces of length 2 or 3, with the number of such faces denoted by $F_i$ for $i = 2, 3$. Using Euler's formula, we have $V - 2V + F_2 + \frac{2}{3}(E - F_2) = 2$, which implies that $F_2 = 1$ since $V = 5$. Therefore, the graph $G$ should contain only one face of length 2, with all remaining faces being of length 3. This configuration corresponds to Graph 1 in Figure \ref{l45phi1}, to which we need to add one more vertex and connect the remaining half-edges.

We should connect $E$ to $F$ to avoid creating a face longer than 3, and there are exactly two non-isomorphic ways to add the remaining vertex, resulting in the two graphs shown in Figure \ref{noloopfive}. Consider Graph 1 in this figure. The half-edges $F$ and $G$ should be linked, thus connecting $E$ and $H$. If we link $D$ to $C$, then $A$ should connect to $B$, which creates a contradiction because it results in two faces of length 2. Therefore, we should link $D$ to $A$ and $B$ to $C$, which also creates a contradiction by forming a face of length 4 ($ABCD$).

Now, consider Graph 2 in the same figure. We should link $A$ to $B$ and $E$ to $D$. This implies that $F$ can only connect to $G$ or $C$ (implying that $H$ connects to $C$ or $G$, respectively), resulting in two faces of length 2 in both cases. This is the final contradiction. Thus, it is impossible to construct a graph where all faces are of length 2 or 3. Therefore, the graph $G$ must contain a face of length 4 or 5.

Assume there is a face of length 5. It is straightforward to prove that this configuration necessarily results in the necklace graph. Now, let's assume there is a quadrilateral face in $G$. Since $G$ is a connected graph, it takes the form of Graph 3 in Figure \ref{l45phi1}. The half-edge $C$ in this figure can only be connected to $B$, $H$, or $J$; otherwise, we would create a tadpole or place an odd number of half-edges in a face, violating the planarity of $G$.

To avoid having more than one $\oD$-loop, a tadpole, or an even number of half-edges in a face, the half-edge $C$ cannot connect to $J$ without causing a contradiction. If $C$ is connected to $H$, then the half-edge $D$ can only be connected to $G$, and the connections $E$ to $F$, $A$ to $B$, and $I$ to $F$ follow, resulting in Graph 4 in Figure \ref{l45phi1}. If $C$ is connected to $B$, then $D$ should be connected to $A$, $E$ to $F$, $G$ to $H$, and $I$ to $J$, resulting in a graph isomorphic to Graph 4 in Figure \ref{l45phi1}.

Ultimately, we obtain a unique representation given by Graph 4 in Figure \ref{l45phi1}.

\begin{figure}[H]
\centering
\includegraphics[width=110mm]{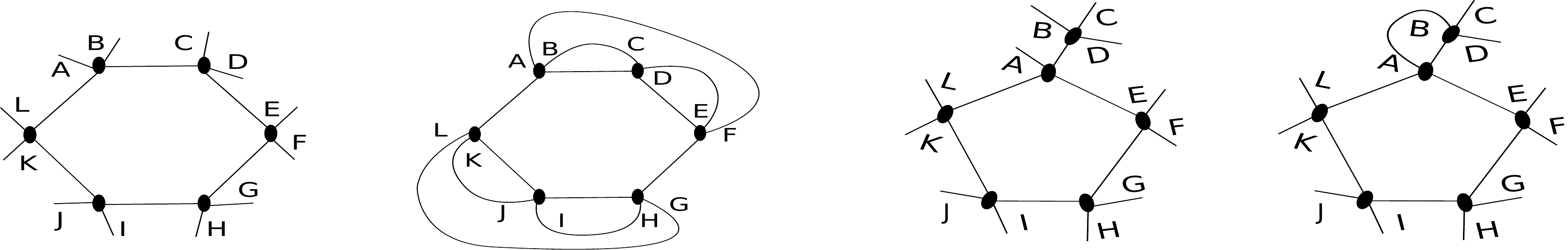}
\caption{\small{Schematic representations with a face of length 6; $\ell=6$ and $\varphi=1$ graph.}}
\label{facesix}
\end{figure}

\item Consider the case where $\ell = 6$. Assume that the graph $G$ has a face of length 6, resulting in the first configuration in Figure \ref{facesix} (reading from left to right). The half-edge $A$ cannot be connected to $B$, $C$, $D$, $E$, $G$, $H$, $I$, $J$, or $K$. Connecting $A$ to $B$ would generate a tadpole, and similarly, connecting $A$ to $J$ would join $L$ to $K$, also creating a tadpole. Connecting $A$ to $C$, $E$, $G$, or $I$ would result in an odd number of half-edges in a face. Connecting $A$ to $D$ would create a $\oD$-loop of length 2 (between $A$ and $D$), resulting in more than one $\oD$-loop in the graph. Linking $A$ to $H$ would join $L$ to $K$, $J$, or $I$, as the graph is planar. Connecting the half-edge $L$ to $K$ would generate a tadpole. Connecting $L$ to $J$ is impossible because it would leave a single half-edge $K$ alone in a face, and connecting $L$ to $I$ would create another $\oD$-loop of length 2 (between $L$ and $I$). Thus, $A$ can only be connected to $F$ or $L$.

To avoid having more than one $\oD$-loop, a tadpole, or an odd number of half-edges in a face, if $A$ connects to $F$, then $B$ should be joined to $C$ and $D$ to $E$. Similarly, $L$ should be connected to $G$, $K$ to $J$, and $I$ to $H$, resulting in the second graph in Figure \ref{facesix} which we will refer to as $G_1$.

If $A$ is connected to $L$, then for the same reasons listed earlier, the half-edge $B$ cannot connect to $D$, $E$, $F$, $G$, $H$, $J$, or $K$. Therefore, $B$ can only be joined to $C$ or $I$. If $B$ connects to $C$, then $D$ should connect to $K$. Due to the same reasons listed earlier, $D$ cannot connect to $F$, $G$, $H$, $I$, or $J$. Furthermore, connecting $D$ to $E$ would result in $F$ connecting to $G$, $H$ to $I$, and $J$ to $K$, creating a graph with two $\oD$-loops. Similarly, $E$ should connect to $J$, $F$ to $G$, and $H$ to $I$, resulting in a graph isomorphic to $G_1$.

\begin{figure}[H]
\centering
\includegraphics[width=100mm]{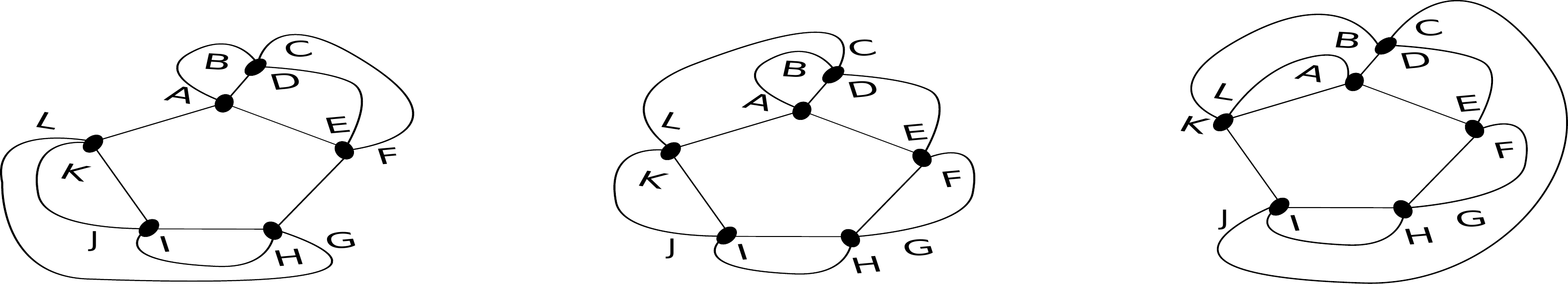}
\caption{\small{Configurations with a face of length 5; $\ell=6$ and $\varphi=1$ graph.}}
\label{facefive}
\end{figure}

Assume now that there is a face of length 5 and no face of greater length. The remaining vertex, not yet used in the face, is connected to one of the half-edges, resulting in the third configuration shown in Figure \ref{facesix}. To avoid having more than one $\oD$-loop, a tadpole, or an odd number of half-edges in a face, we proceed as before. The half-edge $A$ cannot be connected to $C$, $D$, $E$, $F$, $G$, $H$, $I$, $J$, or $K$. Therefore, $A$ can only be connected to $B$ or $L$.

Assume that $A$ is connected to $B$. The half-edge $C$ can only connect to $F$ or $L$. In the first case, $D$ should connect to $E$, $G$ to $L$, $K$ to $J$, and $I$ to $H$, resulting in the first graph in Figure \ref{facefive}, which we will refer to as $G_2$.

Now, consider the second case where $C$ is connected to $L$. The half-edge $D$ can only connect to $E$ or $K$. If $D$ connects to $E$, then $F$ should connect to $G$, $K$ to $J$, and $I$ to $H$, resulting in the second graph in Figure \ref{facefive}, which we will call $G_3$. Connecting $D$ to $K$ implies that $E$ should connect to $J$, $I$ to $H$, and $F$ to $G$, giving a graph isomorphic to $G_2$.

Now, assume that $A$ is connected to $L$. The half-edge $B$ should then connect to $G$, $I$, or $K$. Connecting $B$ to $G$ will result in $C$ connecting to $F$, $D$ to $E$, $H$ to $I$, and $J$ to $K$, producing a graph isomorphic to $G_3$. Alternatively, if $B$ is connected to $K$, then $C$ should connect to $J$, $D$ to $E$, $F$ to $G$, and $H$ to $I$, resulting in the third graph in Figure \ref{facefive}, which we will refer to as $G_4$. Connecting $B$ to $I$ and proceeding similarly will yield a graph isomorphic to $G_4$.

\begin{figure}[H]
\centering
\includegraphics[width=100mm]{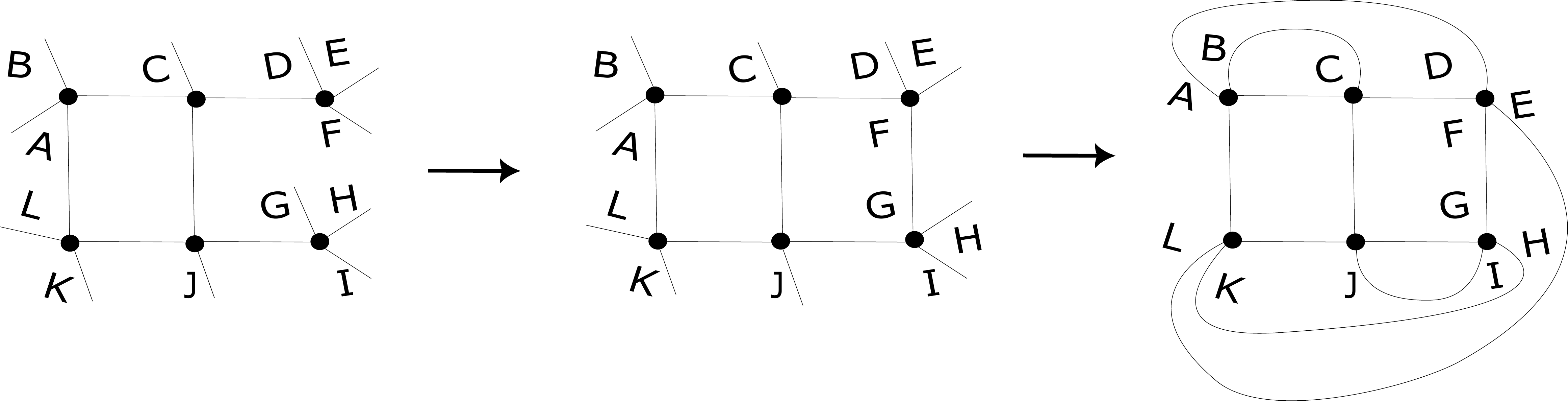}
\caption{\small{Configurations with a face  of length 4; $\ell=6$ and $\varphi=1$ graph.
}}
\label{facefour1}
\end{figure}

\begin{figure}[H]
\centering
\includegraphics[width=100mm]{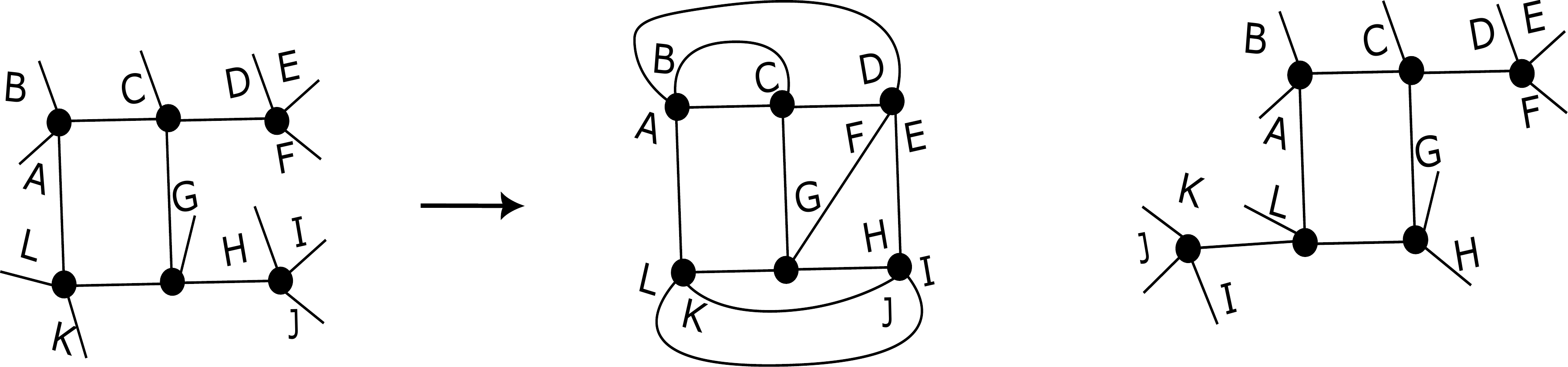}
\caption{\small{Configurations with a face  of length 4; $\ell=6$ and $\varphi=1$ graph.
}}
\label{facefour2}
\end{figure} 

\begin{figure}[H]
\centering
\includegraphics[width=100mm]{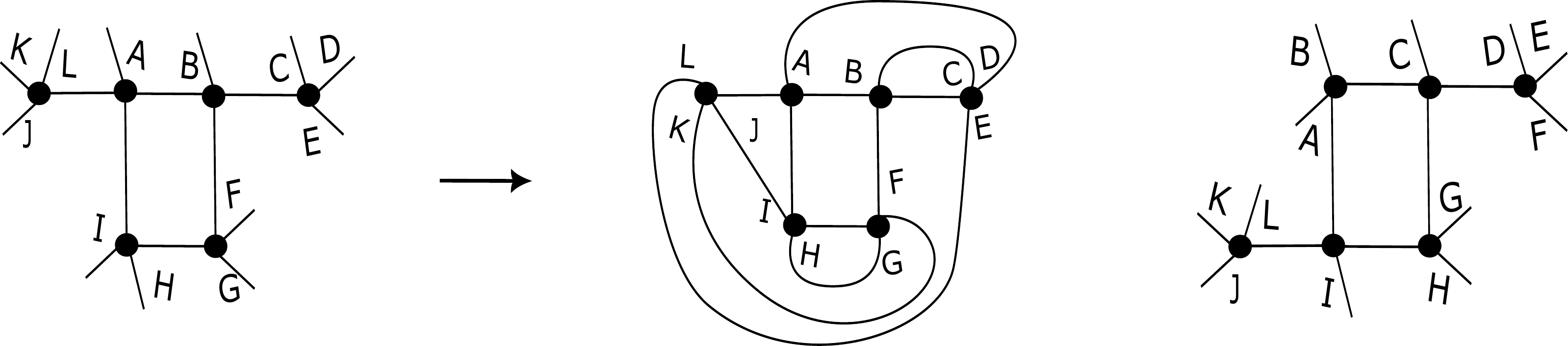}
\caption{\small{Configurations  with a face   of length 4; $\ell=6$ and $\varphi=1$ graph.
}}
\label{facefour3}
\end{figure}

\begin{figure}[H]
\centering
\includegraphics[width=90mm]{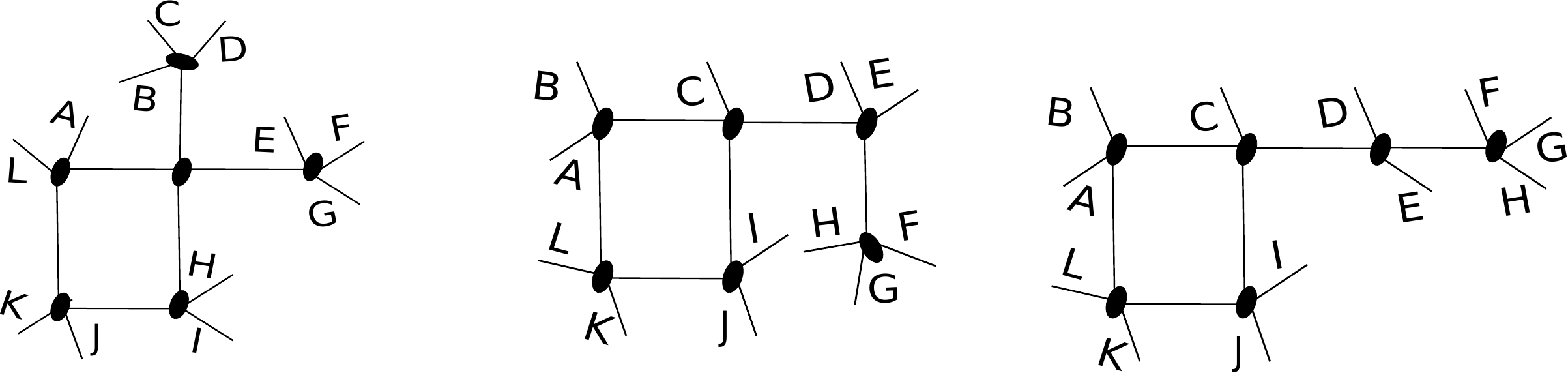}
\caption{\small{Configurations  with a face  of length 4; $\ell=6$ and $\varphi=1$ graph.
}}
\label{facefour4}
\end{figure}

Assume the graph contains a face of length 4, with no other face exceeding this length. The two remaining vertices can connect to the square face in eight distinct configurations: the first configuration shown in Figure \ref{facefour1}, the first and third configurations in Figures \ref{facefour2} and \ref{facefour3}, and the three configurations illustrated in Figure \ref{facefour4}.

Consider the first configuration in Figure \ref{facefour1}. To avoid creating a face longer than 4, the half-edge $F$ should connect to $G$ (as shown in the second graph of the same figure). We maintain the previous conditions, assuming the resulting graph has no more than one $\oD$-loop, no tadpole, or an odd number of half-edges in a face. Consequently, the half-edge $A$ should connect to $D$, $H$, or $L$. If $A$ connects to $D$, then $B$ must connect to $C$, and $E$ must connect to either $H$ or $L$. Connecting $E$ to $H$ is not feasible because $L$ would then connect to $I$ and $K$ to $J$, resulting in a graph with more than one $\oD$-loop. If $E$ connects to $L$, then $H$ connects to $K$ and $I$ to $J$, forming the third graph ($G_5$) in Figure \ref{facefour1}.

Assume $A$ connects to $H$. This implies $E$ should connect to $B$, $C$ to $D$, $L$ to $I$, and $K$ to $J$, resulting in a graph isomorphic to $G_4$. Now, suppose $A$ connects to $L$. The half-edge $B$ should then connect to $C$, $E$, or $I$. Connecting $B$ to $C$ is not possible because it would result in $D$ connecting to $K$, $I$ to $J$, and $H$ to $E$, creating a graph with more than one $\oD$-loop. Connecting $B$ to $E$ is also not feasible because $C$ would connect to $D$, $H$ to $K$, and $I$ to $J$, again resulting in a graph with more than one $\oD$-loop. A similar contradiction arises if $B$ connects to $I$, as $C$ would connect to $D$, $E$ to $H$, and $K$ to $J$.

Now, consider the first configuration in Figure \ref{facefour2}. Similar to the previous configuration, the half-edge $A$ should connect to $D$, $F$, $H$, or $L$. Connecting $A$ to $D$ will join $B$ to $C$, and $E$ to either $H$ or $L$. If $E$ connects to $H$, then $G$ should connect to $F$, $K$ to $J$, and $I$ to $L$, resulting in the second graph ($G_6$) in Figure \ref{facefour2}. If $E$ connects to $L$, then $G$ will connect to $H$, $I$ to $F$, and $J$ to $K$, forming a graph isomorphic to $G_5$.

Assume $A$ connects to $F$. This implies $B$ should connect to $E$, $C$ to $D$, $L$ to $I$, $K$ to $J$, and $G$ to $H$, resulting in a graph isomorphic to $G_2$. Connecting $A$ to $H$ is not feasible because planarity requires $L$ to connect to $I$, $K$ to $J$, $B$ to $E$, $C$ to $D$, and $F$ to $G$, creating a graph with two $\oD$-loops.

Now, connect $A$ to $L$. This implies $B$ should connect to $C$, $E$, or $I$. If $B$ connects to $C$, then $D$ should connect to $I$ or $K$. If $D$ connects to $I$, then $F$ should connect to $G$, $E$ to $H$, and $J$ to $K$, resulting in a graph isomorphic to $G_4$. Connecting $D$ to $K$ is not possible because it would join $E$ to $J$, $F$ to $I$, and $G$ to $H$, creating a graph with two $\oD$-loops.

Connecting $B$ to $E$ is also not feasible because it would result in $C$ connecting to $D$, $F$ to $I$, $G$ to $H$, and $J$ to $K$, forming a graph with two $\oD$-loops. Finally, if $B$ connects to $I$, then $C$ will connect to $D$, $E$ to $H$, $F$ to $G$, and $J$ to $K$, resulting in a graph isomorphic to $G_5$

Now, let's discuss the third configuration involving a face of length 4 in Figure \ref{facefour2}, under the condition that the resulting graph does not have more than one $\oD$-loop, tadpole, or an odd number of half-edges in a face.

The half-edge $A$ can only connect to $D$, $J$, or $L$. Connecting $A$ to $D$ will join $B$ to $C$ and $E$ to either $J$ or $L$. If $E$ connects to $J$, then $F$ will connect to $G$, $H$ to $I$, and $K$ to $L$, resulting in a graph isomorphic to $G_6$. If $E$ connects to $L$, then $K$ will connect to $F$, $J$ to $G$, and $I$ to $H$, forming a graph isomorphic to $G_5$.

Connecting $A$ to $J$ will join $K$ to $L$. The half-edge $B$ can then connect to $C$, $E$, or $I$. Connecting $B$ to $C$ is not feasible as it will join $D$ to $I$, $E$ to $H$, and $F$ to $G$, resulting in a graph with more than one $\oD$-loop. Connecting $B$ to $E$ is also impossible because it will join $C$ to $D$, $F$ to $G$, and $H$ to $I$, again resulting in a graph with more than one $\oD$-loop. If $B$ connects to $I$, then $C$ will join $D$, $E$ to $H$, and $F$ to $G$, resulting in a graph isomorphic to $G_1$.

Now, let's connect $A$ to $L$. This implies that $B$ should connect to $C$, $E$, or $K$. Connecting $B$ to $C$ implies that $D$ should connect to $K$, $E$ to $J$, $F$ to $G$, and $I$ to $H$, resulting in a graph isomorphic to $G_4$. If $B$ connects to $E$, then $C$ should connect to $D$, $F$ to $K$, $G$ to $J$, and $H$ to $I$, resulting in a graph isomorphic to $G_6$.
Consider the first configuration in Figure \ref{facefour3}. We assume the same conditions as before, ensuring that the resulting graph (after connecting some half-edges) does not have more than one $\oD$-loop, tadpole, or an odd number of half-edges in a face. The half-edge $A$ can only connect to $B$, $D$, or $L$. If $A$ connects to $B$, then $C$ should connect to $L$, and $D$ should connect to either $G$ or $K$. Connecting $D$ to $K$ would result in a graph with more than one $\oD$-loop, which is a contradiction. Connecting $D$ to $G$ is also not possible, as it would imply that $H$ connects to $K$, resulting in a graph with more than one $\oD$-loop.

Now, assume $A$ is connected to $D$. In this case, $B$ should connect to $C$, and $E$ to either $F$ or $L$. If $E$ connects to $F$, then $G$ should connect to $L$, $H$ to $K$, and $I$ to $J$, which would again result in a graph with more than one $\oD$-loop. If $E$ connects to $L$, then $F$ should connect to $K$, $G$ to $H$, and $I$ to $J$, resulting in the second graph ($G_7$) in Figure \ref{facefour3}.

Now, assume $A$ connects to $L$. This implies that $B$ should connect to $C$ or $K$. Connecting $B$ to $C$ implies that $D$ will connect to either $G$ or $K$. Connecting $D$ to $G$, with careful analysis, will create a graph isomorphic to $G_1$. If $D$ connects to $K$, a similar analysis will result in a graph with more than one $\oD$-loop, which is a contradiction.

Now, let's consider the third configuration in Figure \ref{facefour3}, under the same assumptions as before. The half-edge $A$ can only connect to $D$, $F$, or $L$. Connecting $A$ to $D$ implies that $B$ should connect to $C$, $E$ should connect to $L$, and $F$ will connect to either $K$ or $G$. If $F$ connects to $K$, then $H$ should connect to $I$ and $G$ to $J$, which would result in a graph with more than one $\oD$-loop, creating a contradiction. Therefore, $F$ can only connect to $G$, implying that $H$ should connect to $K$ and $I$ to $J$, resulting in a graph isomorphic to $G_6$.

Assume $A$ is connected to $F$. A similar procedure results in a graph isomorphic to $G_1$. Connecting $A$ to $L$ will connect $B$ to either $C$ or $E$. Connecting $B$ to $C$ is not feasible because it would, after similar analysis, create a graph with two $\oD$-loops. If $B$ connects to $E$, the same procedure leads to a graph isomorphic to $G_5$.

Under the same hypotheses that the resulting graph (obtained after connecting some half-edges) does not have more than one $\oD$-loop, tadpole, or an odd number of half-edges in a face, we now consider the first configuration in Figure \ref{facefour4}. A thorough analysis shows that $A$ can only connect to $B$ or $H$.

If $A$ connects to $B$, then $C$ should connect to $F$, $J$, or $L$. Connecting $C$ to $F$ will result in a graph isomorphic to $G_4$ after connecting the remaining half-edges. Similarly, connecting $C$ to $J$ or $L$ will result in a graph isomorphic to $G_7$. Specifically, connecting $C$ to $J$ will send $L$ to $K$, $D$ to $E$, $G$ to $H$, and $F$ to $I$. Connecting $C$ to $L$ gives two possibilities: $D$ to $E$ or $K$. If $D$ connects to $E$, the resulting graph is isomorphic to $G_7$, but connecting $D$ to $K$ results in a graph with more than two $\oD$-loops.

Finally, connecting $A$ to $H$ and proceeding similarly results in a graph with more than one $\oD$-loop, which is not possible. Applying the same analysis to the remaining two configurations in Figure \ref{facefour4} results in either graphs with more than one $\oD$-loop or graphs isomorphic to one of the graphs $G_i$, $i=1, \cdots, 7$.

\begin{figure}[H]
\centering
\includegraphics[width=110mm]{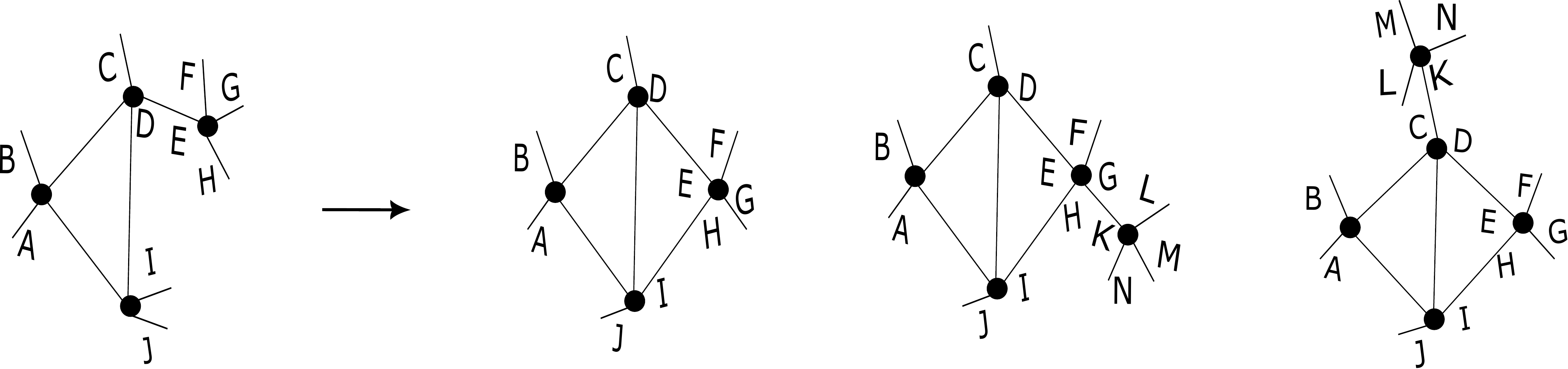}
\caption{\small{Configurations  with a face of length 3; $\ell=6$ and $\varphi=1$ graph.
}}
\label{facethree1}
\end{figure}

\begin{figure}[H]
\centering
\includegraphics[width=90mm]{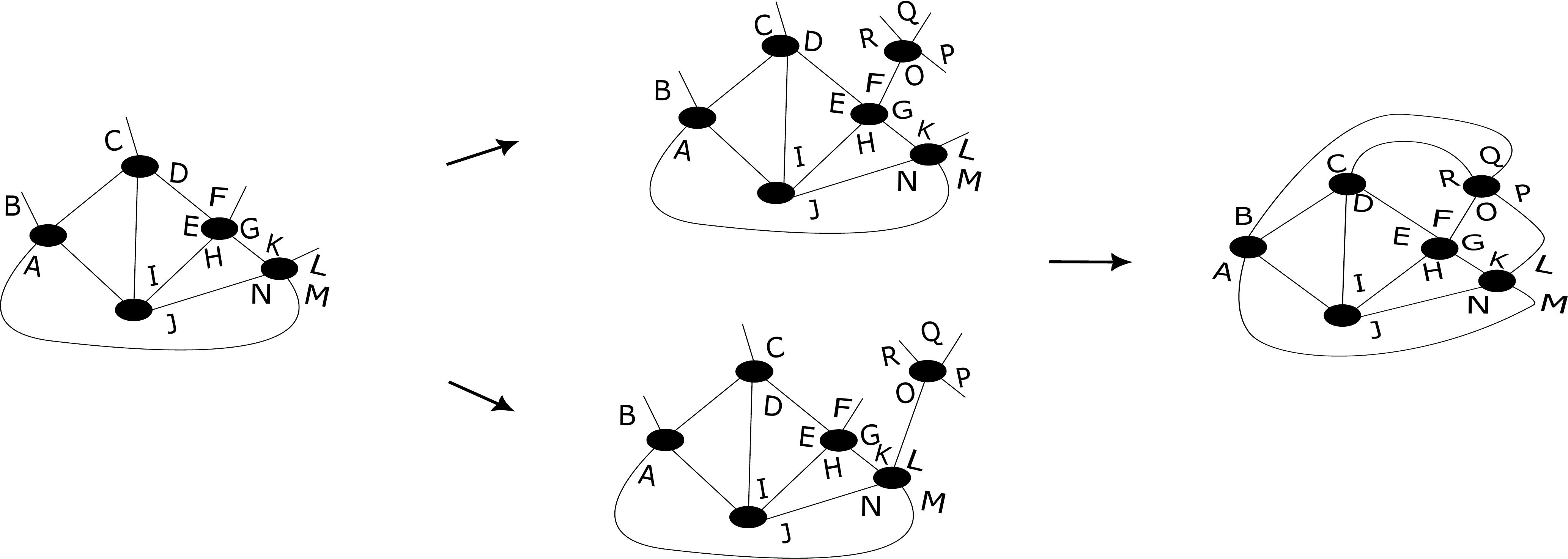}
\caption{\small{Configurations  with a face of length 3; $\ell=6$ and $\varphi=1$ graph.
}
}
\label{facethree2}
\end{figure}

\begin{figure}[H]
\centering
\includegraphics[width=90mm]{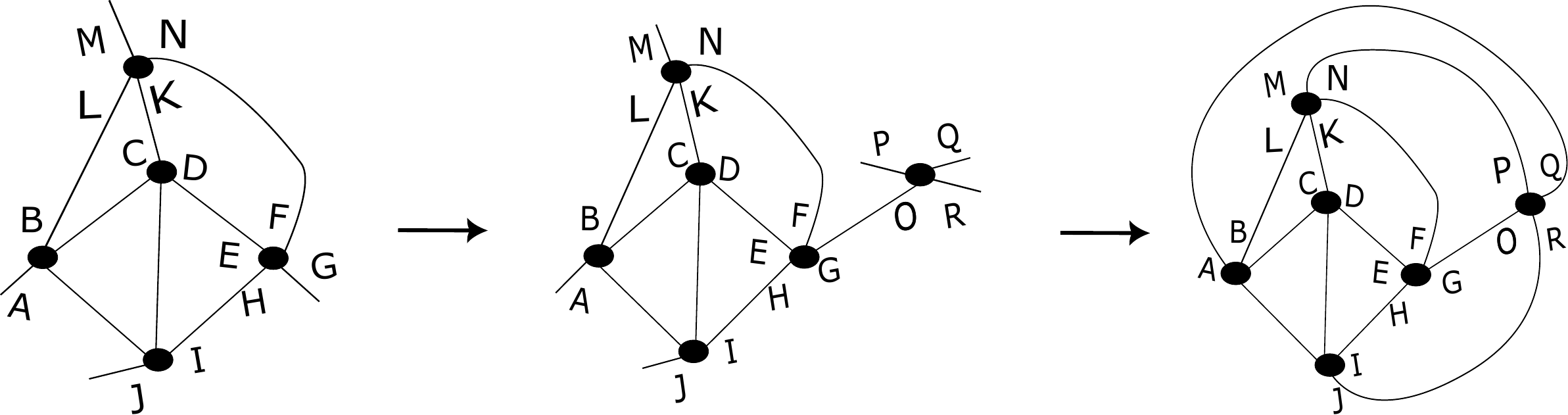}
\caption{\small{Configurations with a face of length 3; $\ell=6$ and $\varphi=1$ graph.
}}
\label{facethree4}
\end{figure}

Assume now that the graph has no face of lenght longer than 3. Therefore, each face is either of length 3 or 2, denoted by $F_2$ and $F_3$ for the number of faces of length 2 and 3, respectively. This leads to the following equations: $F_2 + F_3 = 8$ (using the Euler characteristic) and $2F_2 + 3F_3 = 24$ (since the sum of the lengths of the faces is twice the number of edges). Solving these equations gives $F_2 = 0$ and $F_3 = 8$, indicating that all faces of the graph are triangular.

Assume we have a configuration with a face of length 3, to which we connect an additional vertex as shown in the first configuration in Figure \ref{facethree1}. We need to connect the remaining two vertices, but we must first connect the half-edge $H$ to $I$, ensuring no face is longer than three, as shown in the second configuration in the same figure. Connecting the fifth vertex results in two possible non-isomorphic configurations, depicted in the remaining two graphs in Figure \ref{facethree1}. The third configuration in Figure \ref{facethree1} implies that $N$ should connect to $J$, and $M$ to $A$, resulting in the first configuration in Figure \ref{facethree2}.
This configuration already contains an $O(D)$-loop, and connecting the remaining half-lines will result in in a contradiction; namely $\varphi > 1$. The next two configurations in Figure \ref{facethree2} represent the two possible non-isomorphic configurations that arise from adding the last vertex to the first configuration in the same figure. Since no face is longer than 3, these configurations lead to the same graph: the final graph in Figure \ref{facethree2}, with $\varphi = 3$.

Consider the last configuration in Figure \ref{facethree1}. We should connect $N$ to $F$ and $L$ to $B$. Connecting the last vertex results in a single configuration, as shown in the second configuration of Figure \ref{facethree4}. Consequently, the final graph in the same figure follows. This graph also satisfies $\varphi = 3$.
\end{enumerate}
\end{proof}

\begin{remark}
We could expect to extend some of the results in Theorem \ref{prop:enumerate} in the following way:
    Let $G$ be a connected 2PI melon-free Feynman graph with grade $\ell$. If loop configuration has the form $(4,4,\dots,4, 2\ell)$, its scheme $S_G$ has the form in Figure \ref{fig:cong}.
\begin{figure}[H]
\centering
\includegraphics[width=35mm]
{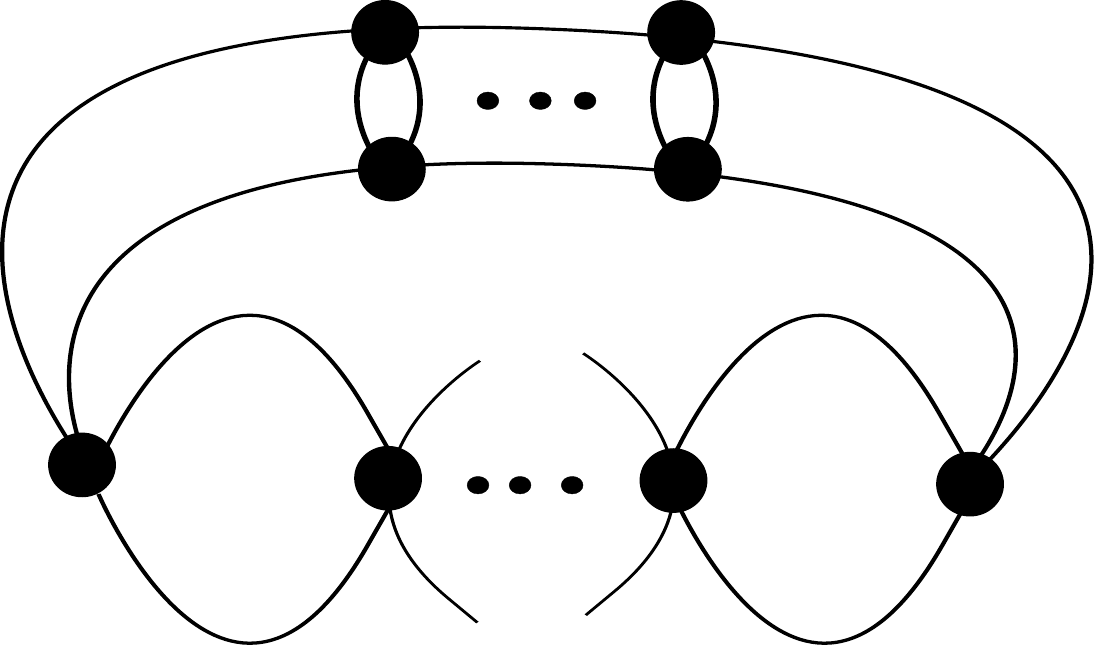}
\caption{\small{Some schematic representations of $(4,4,\dots,4, 2\ell)$ with $\ell-2$ crossings.}}
\label{fig:cong}
\end{figure}

Unfortunately this only works for $\ell=3,4$. Starting from $l=5$, we could find more configurations as illustrated in Figure \ref{fig:conter}.

\begin{figure}[H]
\centering
\includegraphics[width=35mm]
{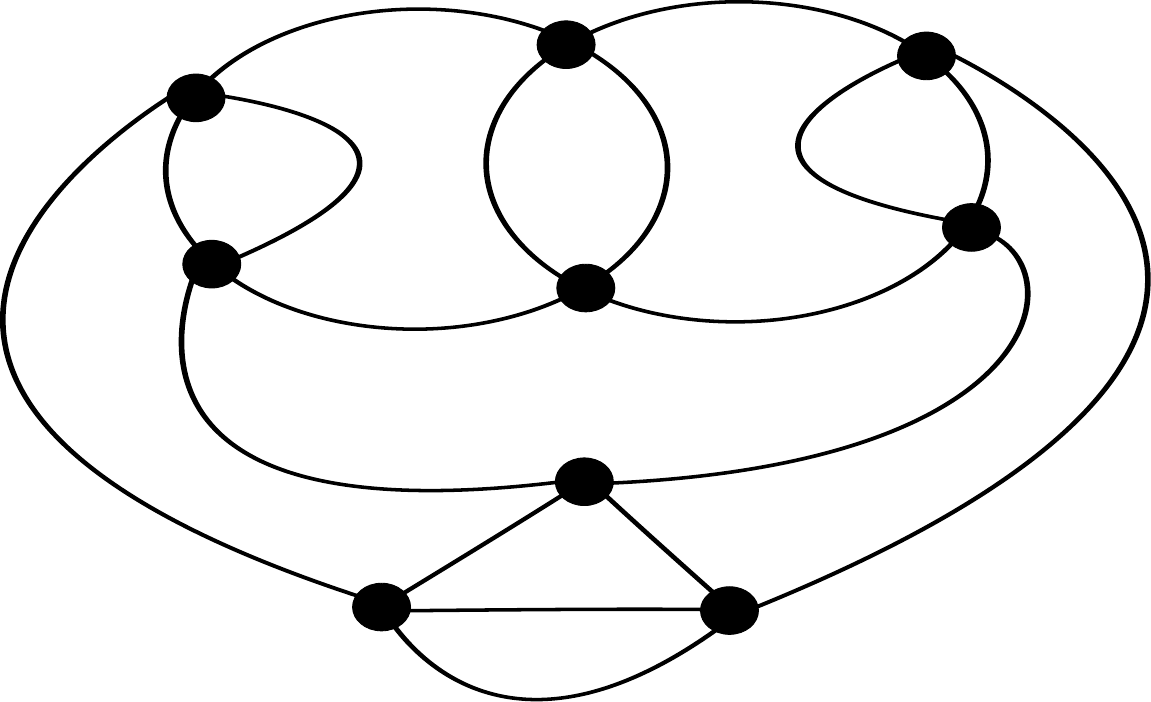}
\caption{\small{Some schematic representations of $(4,4, 10)$.}}
\label{fig:conter}
\end{figure}  
\end{remark}

We can also prove the previous results by relating any of the 4-regular graphs we are studying, to an alternating knot. This is mainly motivated by the following claim: Every four-regular planar graph is the medial graph of some planar graph. For a connected 4-regular planar graph $M$, a planar graph $G$ with $M$ as its medial graph can be constructed as follows. Color the faces of $M$ with only two colors, which is possible because $M$ is Eulerian and thus the dual graph of $M$ is bipartite. The vertices in $G$ correspond to the faces of a single color in $M$. These vertices are linked by an edge for each vertex shared by their corresponding faces in $M$. It should be noted that performing this construction with the faces of the other color as the vertices results in the dual graph of $G$.
Using the checkerboard coloring, the medial graph or 4-regular graph can be associated to an alternating link where each crossing in the link diagram corresponds to a vertex of the graph. This link diagram is an alternating knot for $\varphi=1$.

Using the Rolfsen knot table, there are only one knot with 3 or 4 crossing and two with five crossings. This confirms our results in Proposition \ref{prop:enumerate} for $\ell=3,4,5$ and $\varphi=1$. In the case $\ell=6$ and $\varphi=1$ we could expect to get three graphs because there are three knots with 3 crossing in the Table but there are more than that which actually can be obtain by taking the connected sum (because the table only contains prime knots) of two trefoils or applying the Tait flyping moves (the correspondent to Reidemeister moves on 3-spheres). This leads to the following result whose proof is straightforward.

\begin{thm}
\label{thm:knot}
Each 2PI 4-regular planar graph with $\varphi=1$ and any $\ell$, ignoring the orientation assignment on the edges, is in one-to-one correspondance with reduced alternating knot diagrams with $\ell$ crossings
1) which are projections of the prime knots as listed in the Rolfsen knot table, or 
2) which are obtained after performing the Tait flyping moves.

Furthermore, we can correspond each 2PR 4-regular planar graphs with $\varphi=1$ and any $\ell$ to an alternating knot diagram obtained by performing a connected sum or a Reidemeister move I on the reduced alternating knot diagrams referred above.
\end{thm}

\begin{proof}
Firstly, from Proposition \ref{prop:grade}, if we have $g=0$, $\varphi=1$, immediately $\ell=v$.

Using the one-to-one correspondance between 4-regular planar graphs and alternating knot diagrams, it is sufficient to consider and stay in the space of alternating knot diagrams.

The Tait flyping Conjecture allows us to draw all possible reduced alternating diagrams of a given  alternating knot.
The Tait flyping moves keeps the number of crossing invariant, therefore under the moves, we stay in the same $\ell$.

Furthermore, Reidemeister moves II and III will get us away from the space of alternating knot diagrams, therefore, we will not consider them.
A Reidemeister move I will increase the number of crossings by one, and create a tadpole (self-loop).
Additionally, a connected sum of the alternating knot diagrams can be performed to yield more alternating knot diagrams, however, by definition, it will create a bridge, yielding 2PR graphs.

Finally, in order to obtain the Feynman graphs of the $\uN^2 \times \oD$ matrix model with any $\ell$ and $g=0$ with $\varphi =1$ at the double scaling limit of this current work, one needs to assign orientations to the edges and identify distinct combinatorial maps corresponding to these 4-regular planar graphs obtained by considering alternating knot diagrams obtained above.
\end{proof}

\section{Conclusion}
\label{sec:concl}

In this work, we have classified grade $\ell=1$ and $\ell=2$ with any genus $g$ Feynman graphs generated by a $\uN^2 \times \oD$ multi-matrix model of the type defined by the action \eqref{eq:actionND}. 
The result is given in Theorem \ref{thm:ell1ell2planar} and Theorem \ref{thm:induction} along with the algorithms on how to recursively construct higher genus graphs for a given $\ell$.
Furthermore, we classified grade $\ell=3$ with $g=0$ planar graphs as given in Theorem \ref{thm:ell3g02PI}.
We also classified subclasses of higher grade $\ell>3$ planar ($g=0$) graphs in Theorem \ref{prop:enumerate} and Theorem \ref{thm:knot}.
Theorem \ref{thm:knot} states that a subclass of higher grade $\ell>3$ planar graphs which have one $\oD$-loop ($\varphi = 1$) can be obtained from knot diagrams of Rolfsen table.

One of the main motivations for this work is a possible formulation of topological recursion appearing in this model. In fact, the enumeration of ordinary combinatorial maps was a major precursor of topological recursion. It was the first enumeration to be shown to be governed by topological recursion, and the theory of topological recursion evolved from the abstraction of loop equations from the theory of matrix models in the specific context of the 1-Hermitian matrix model mentioned above \cite{MR2118807}. 

Tensor models are built on tensor integrals, which are a generalization of matrix integrals. As a result, it is not surprising that tensor models generate Schwinger-Dyson equations that generalize those of matrix models.  Despite advances, solving the Schwinger-Dyson equations for tensor models remains extremely difficult. However, some tensor models can be completely solved in the $1/N$ expansion at any order and still satisfy the topological recursion. This is due to the fact that those tensor models are disguised matrix models, and most recent progress in tensor models relied on a correspondence between tensor models and some multi-matrix models with multi-trace interactions. This allows for the indirect study of tensor models via matrix model techniques on the corresponding matrix models. Thus, in \cite{MR3829394}, it was demonstrated that the quartic melonic tensor model, after some transformation to a matrix model, satisfies the blobbed topological recursion.

As explained above, the $\uN^2 \times \oD$ model is affine to matrix models, yet it can be readily interpreted as a tensor model.
What is attractive about this model is that there are two parameters owing to different sizes $N$ and $D$, which then let us incrementally classify the graphs as compared to the usual tensor model where only one parameter $N$ controls the classification of the graphs with one parameter Gurau degree $\omega$.
It means that it is more likely that one can classify higher order graphs, as stated earlier.
If one can classify higher order graphs, one may be able to identify topological recursion at play.
This model's affinity to matrix models, indeed giving topological (genus) expansion, is promising to find such a topological recursion.
This will be left to explore in the future.

Additionally, it is interesting to see the critical/continuum limit(s) of the model for the higher $\ell$, as was carried out in \cite{MR4450018} for $\ell=0$, to see if there is any relevance to quantum gravity (would it give rise to any phases other than branched polymer (tree) or planar?).
In fact, by looking at the higher $\ell =1, \, 2, \, 3$ graphs with any genus that we constructed in this work, one can speculate that the critical behavior at a given $\ell =1, \,2, \, 3 $ in the double scaling limit will likely still be trees as in the case for $\ell=0$.
It appears that the presence of 2PR schemes with B-ladder-vertices will still give rise to the dominant critical behavior similarly to the $\ell=0$ case in \cite{MR4450018}, therefore, we will again obtain trees. 
We will leave the precise analysis in the future work.

Finally, just as in Theorem \ref{thm:knot}, knot theory was useful in classifying subset of graphs of higher $\ell$, can it further help us enumerate more graphs?

\appendix{}
\section{Isomorphism of graphs}
\label{sec:appiso}

In this appendix, we illustrate isomorphisms of graphs.
It serves as a supplement to help elaborate and clarify the analyses of the classification of the graphs.

\begin{figure}[H]
\begin{center}
\begin{minipage}[t]{0.8\textwidth}
\centering
\def\svgwidth{0.8\columnwidth}
\tiny{
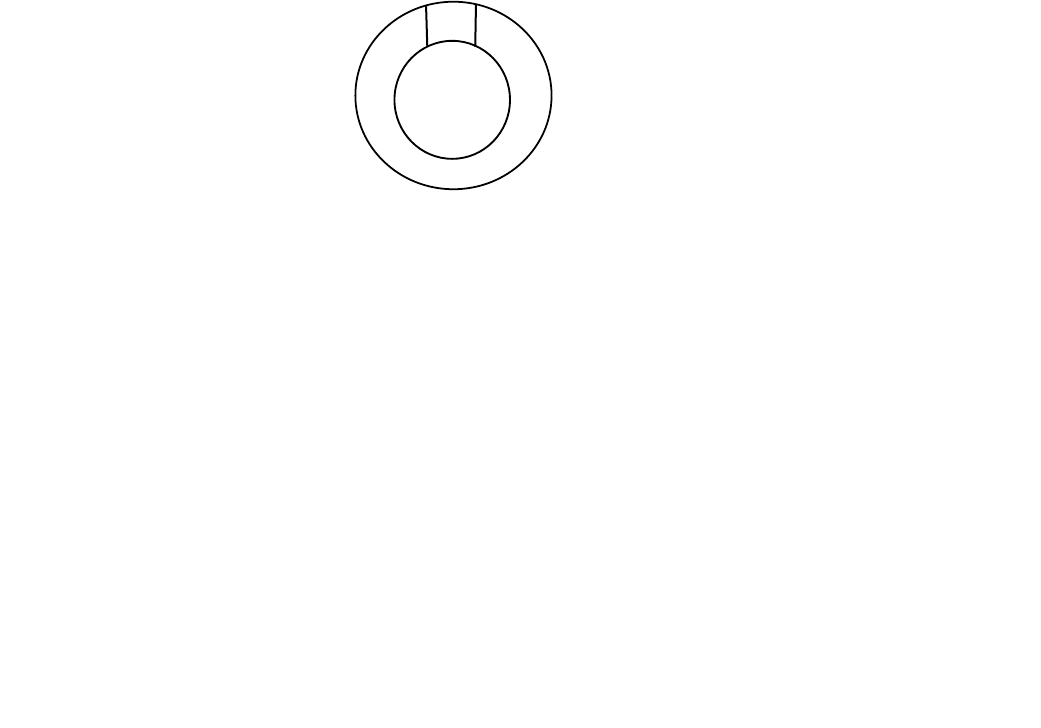
}
\caption{
Isomorphic graphs with cuts (i.e., marked edges) relevant in the operations depicted in Fig. \ref{2prEll1g1}.
}
\label{fig:isomorphismwithcuts}
\end{minipage}
\end{center}
\end{figure}

\begin{figure}[H]
\begin{minipage}[t]{0.95\textwidth}
\centering
\def\svgwidth{1\columnwidth}
\tiny{
\begingroup%
  \makeatletter%
  \providecommand\color[2][]{%
    \errmessage{(Inkscape) Color is used for the text in Inkscape, but the package 'color.sty' is not loaded}%
    \renewcommand\color[2][]{}%
  }%
  \providecommand\transparent[1]{%
    \errmessage{(Inkscape) Transparency is used (non-zero) for the text in Inkscape, but the package 'transparent.sty' is not loaded}%
    \renewcommand\transparent[1]{}%
  }%
  \providecommand\rotatebox[2]{#2}%
  \newcommand*\fsize{\dimexpr\f@size pt\relax}%
  \newcommand*\lineheight[1]{\fontsize{\fsize}{#1\fsize}\selectfont}%
  \ifx\svgwidth\undefined%
    \setlength{\unitlength}{685.59869241bp}%
    \ifx\svgscale\undefined%
      \relax%
    \else%
      \setlength{\unitlength}{\unitlength * \real{\svgscale}}%
    \fi%
  \else%
    \setlength{\unitlength}{\svgwidth}%
  \fi%
  \global\let\svgwidth\undefined%
  \global\let\svgscale\undefined%
  \makeatother%
  \begin{picture}(1,0.35749743)%
    \lineheight{1}%
    \setlength\tabcolsep{0pt}%
    \put(0.17260554,0.05626311){\color[rgb]{0,0,0}\makebox(0,0)[lt]{\lineheight{1.25}\smash{\begin{tabular}[t]{l}\scalebox{2}{$\cong$}\end{tabular}}}}%
    \put(0,0){\includegraphics[width=\unitlength,page=1]{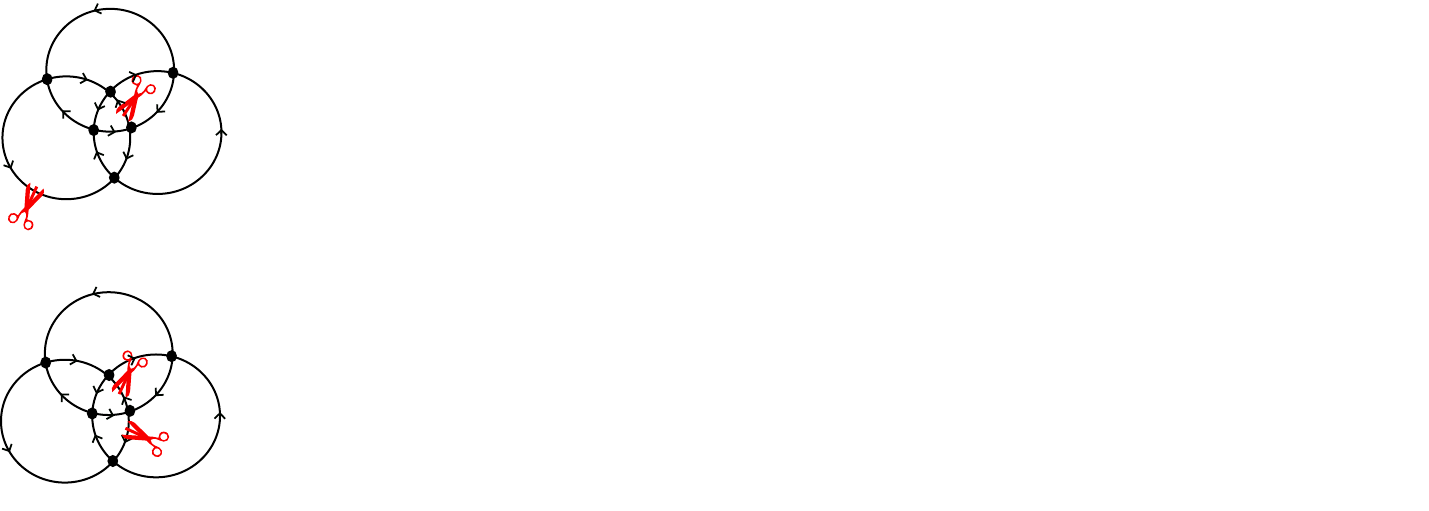}}%
    \put(0.77493589,0.0632119){\color[rgb]{0,0,0}\makebox(0,0)[lt]{\lineheight{1.25}\smash{\begin{tabular}[t]{l}\scalebox{2}{$\cong$}\end{tabular}}}}%
    \put(0,0){\includegraphics[width=\unitlength,page=2]{App1isomophismwithcuts.pdf}}%
    \put(0.46839367,0.05984804){\color[rgb]{0,0,0}\makebox(0,0)[lt]{\lineheight{1.25}\smash{\begin{tabular}[t]{l}\scalebox{2}{$\ncong$}\end{tabular}}}}%
    \put(0.18611791,0.26386224){\color[rgb]{0,0,0}\makebox(0,0)[lt]{\lineheight{1.25}\smash{\begin{tabular}[t]{l}\scalebox{2}{$\cong$}\end{tabular}}}}%
  \end{picture}%
\endgroup%

}
\caption{
Some isomophisms between graphs with cuts (marked edges) which are relevant in the procedure drawn in Fig. \ref{fig:App1}.
}
\label{fig:App1isomophismwithcuts}
\end{minipage}
\end{figure}

\begin{figure}[H]
\begin{minipage}[t]{0.7\textwidth}
\centering
\def\svgwidth{0.65\columnwidth}
\tiny{
\begingroup%
  \makeatletter%
  \providecommand\color[2][]{%
    \errmessage{(Inkscape) Color is used for the text in Inkscape, but the package 'color.sty' is not loaded}%
    \renewcommand\color[2][]{}%
  }%
  \providecommand\transparent[1]{%
    \errmessage{(Inkscape) Transparency is used (non-zero) for the text in Inkscape, but the package 'transparent.sty' is not loaded}%
    \renewcommand\transparent[1]{}%
  }%
  \providecommand\rotatebox[2]{#2}%
  \newcommand*\fsize{\dimexpr\f@size pt\relax}%
  \newcommand*\lineheight[1]{\fontsize{\fsize}{#1\fsize}\selectfont}%
  \ifx\svgwidth\undefined%
    \setlength{\unitlength}{296.62015419bp}%
    \ifx\svgscale\undefined%
      \relax%
    \else%
      \setlength{\unitlength}{\unitlength * \real{\svgscale}}%
    \fi%
  \else%
    \setlength{\unitlength}{\svgwidth}%
  \fi%
  \global\let\svgwidth\undefined%
  \global\let\svgscale\undefined%
  \makeatother%
  \begin{picture}(1,0.35623872)%
    \lineheight{1}%
    \setlength\tabcolsep{0pt}%
    \put(0,0){\includegraphics[width=\unitlength,page=1]{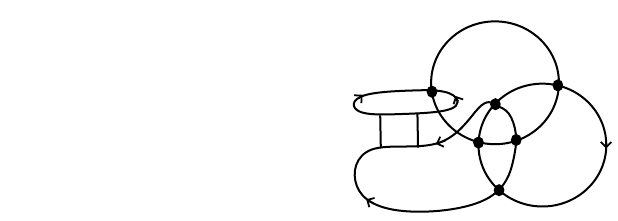}}%
    \put(0.61995709,0.12885146){\color[rgb]{0,0,0}\makebox(0,0)[lt]{\lineheight{1.25}\smash{\begin{tabular}[t]{l}\scalebox{1}{$\widetilde{\rm N}_{\rm o}$}\end{tabular}}}}%
    \put(0,0){\includegraphics[width=\unitlength,page=2]{app1isomorphism.pdf}}%
    \put(0.43179911,0.13178546){\color[rgb]{0,0,0}\makebox(0,0)[lt]{\lineheight{1.25}\smash{\begin{tabular}[t]{l}\scalebox{2}{$\cong$}\end{tabular}}}}%
    \put(0,0){\includegraphics[width=\unitlength,page=3]{app1isomorphism.pdf}}%
    \put(0.07232467,0.09644621){\color[rgb]{0,0,0}\makebox(0,0)[lt]{\lineheight{1.25}\smash{\begin{tabular}[t]{l}\scalebox{1}{$\widetilde{\rm N}_{\rm o}$}\end{tabular}}}}%
    \put(0,0){\includegraphics[width=\unitlength,page=4]{app1isomorphism.pdf}}%
  \end{picture}%
\endgroup%

}
\caption{
Some isomorphic schems relevant in 
Fig. \ref{fig:App1}.
$\widetilde {\rm N}_{\rm o} \in \{ $N-dipole$,\, {\rm N}_{\rm o}\}$.
}
\label{fig:app1isomorphism}
\end{minipage}
\end{figure}

\begin{figure}[H]
\begin{minipage}[t]{0.95\textwidth}
\centering
\def\svgwidth{0.95\columnwidth}
\tiny{
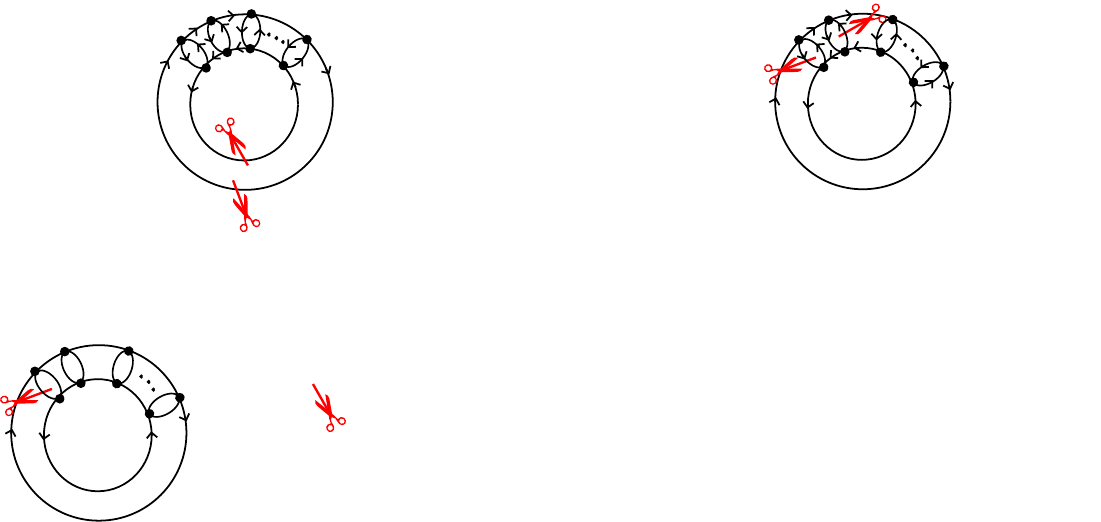
}
\end{minipage}
\vskip 50pt
\begin{minipage}[t]{0.95\textwidth}
\centering
\def\svgwidth{0.95\columnwidth}
\tiny{
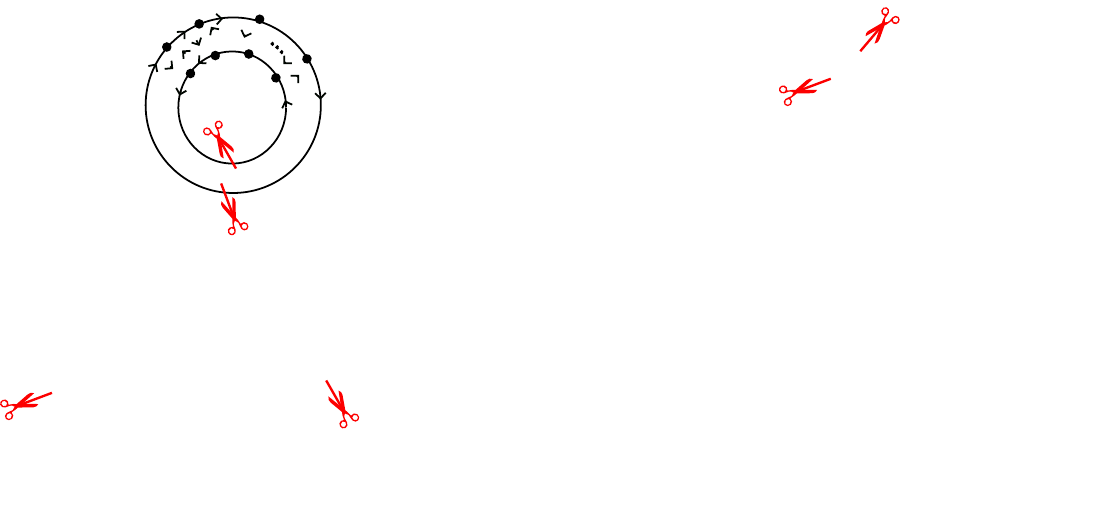
}
\caption{
Some isomorphic graphs with cuts (i.e., marked edges) relevant in Fig. \ref{fig:App2} and Fig. \ref{fig:App21}.
Remark that some are trivial, but we wish to be comprehensive to be clear.
}
\label{fig:App2isomorphismwithcuts}
\end{minipage}
\end{figure}

\begin{figure}[H]
\begin{minipage}[t]{0.95\textwidth}
\centering
\def\svgwidth{0.95\columnwidth}
\tiny{
\begingroup%
  \makeatletter%
  \providecommand\color[2][]{%
    \errmessage{(Inkscape) Color is used for the text in Inkscape, but the package 'color.sty' is not loaded}%
    \renewcommand\color[2][]{}%
  }%
  \providecommand\transparent[1]{%
    \errmessage{(Inkscape) Transparency is used (non-zero) for the text in Inkscape, but the package 'transparent.sty' is not loaded}%
    \renewcommand\transparent[1]{}%
  }%
  \providecommand\rotatebox[2]{#2}%
  \newcommand*\fsize{\dimexpr\f@size pt\relax}%
  \newcommand*\lineheight[1]{\fontsize{\fsize}{#1\fsize}\selectfont}%
  \ifx\svgwidth\undefined%
    \setlength{\unitlength}{459.78592039bp}%
    \ifx\svgscale\undefined%
      \relax%
    \else%
      \setlength{\unitlength}{\unitlength * \real{\svgscale}}%
    \fi%
  \else%
    \setlength{\unitlength}{\svgwidth}%
  \fi%
  \global\let\svgwidth\undefined%
  \global\let\svgscale\undefined%
  \makeatother%
  \begin{picture}(1,0.27624866)%
    \lineheight{1}%
    \setlength\tabcolsep{0pt}%
    \put(0,0){\includegraphics[width=\unitlength,page=1]{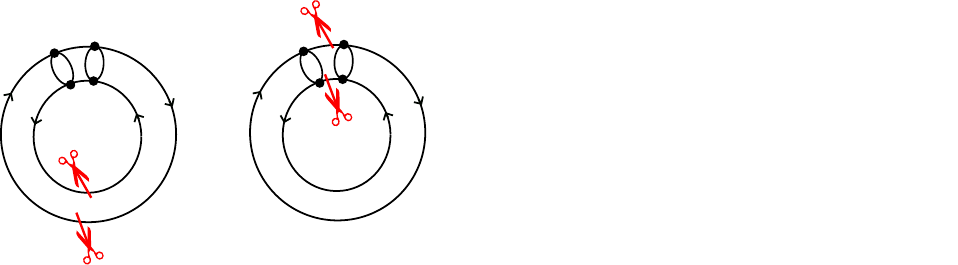}}%
    \put(0.20229097,0.12761577){\color[rgb]{0,0,0}\makebox(0,0)[lt]{\lineheight{1.25}\smash{\begin{tabular}[t]{l}\scalebox{2}{$\cong$}\end{tabular}}}}%
    \put(0,0){\includegraphics[width=\unitlength,page=2]{App2isomorphismwithcutsREDUCED.pdf}}%
    \put(0.73844105,0.12599885){\color[rgb]{0,0,0}\makebox(0,0)[lt]{\lineheight{1.25}\smash{\begin{tabular}[t]{l}\scalebox{2}{$\cong$}\end{tabular}}}}%
    \put(0,0){\includegraphics[width=\unitlength,page=3]{App2isomorphismwithcutsREDUCED.pdf}}%
    \put(0.46716011,0.12826337){\color[rgb]{0,0,0}\makebox(0,0)[lt]{\lineheight{1.25}\smash{\begin{tabular}[t]{l}\scalebox{2}{$\cong$}\end{tabular}}}}%
  \end{picture}%
\endgroup%

}
\end{minipage}
\vskip 30pt
\begin{minipage}[t]{0.95\textwidth}
\centering
\def\svgwidth{0.95\columnwidth}
\tiny{
\begingroup%
  \makeatletter%
  \providecommand\color[2][]{%
    \errmessage{(Inkscape) Color is used for the text in Inkscape, but the package 'color.sty' is not loaded}%
    \renewcommand\color[2][]{}%
  }%
  \providecommand\transparent[1]{%
    \errmessage{(Inkscape) Transparency is used (non-zero) for the text in Inkscape, but the package 'transparent.sty' is not loaded}%
    \renewcommand\transparent[1]{}%
  }%
  \providecommand\rotatebox[2]{#2}%
  \newcommand*\fsize{\dimexpr\f@size pt\relax}%
  \newcommand*\lineheight[1]{\fontsize{\fsize}{#1\fsize}\selectfont}%
  \ifx\svgwidth\undefined%
    \setlength{\unitlength}{461.47028885bp}%
    \ifx\svgscale\undefined%
      \relax%
    \else%
      \setlength{\unitlength}{\unitlength * \real{\svgscale}}%
    \fi%
  \else%
    \setlength{\unitlength}{\svgwidth}%
  \fi%
  \global\let\svgwidth\undefined%
  \global\let\svgscale\undefined%
  \makeatother%
  \begin{picture}(1,0.27881424)%
    \lineheight{1}%
    \setlength\tabcolsep{0pt}%
    \put(0,0){\includegraphics[width=\unitlength,page=1]{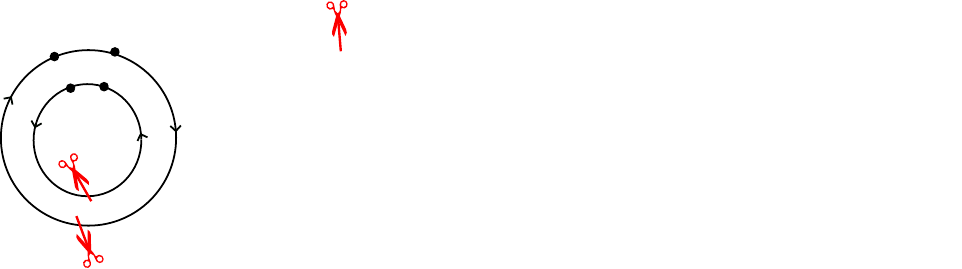}}%
    \put(0.2015526,0.12714996){\color[rgb]{0,0,0}\makebox(0,0)[lt]{\lineheight{1.25}\smash{\begin{tabular}[t]{l}\scalebox{2}{$\cong$}\end{tabular}}}}%
    \put(0,0){\includegraphics[width=\unitlength,page=2]{App2isomorphismwithcutsRREDUCED.pdf}}%
    \put(0.73939574,0.13506333){\color[rgb]{0,0,0}\makebox(0,0)[lt]{\lineheight{1.25}\smash{\begin{tabular}[t]{l}\scalebox{2}{$\cong$}\end{tabular}}}}%
    \put(0,0){\includegraphics[width=\unitlength,page=3]{App2isomorphismwithcutsRREDUCED.pdf}}%
    \put(0.47706924,0.12779519){\color[rgb]{0,0,0}\makebox(0,0)[lt]{\lineheight{1.25}\smash{\begin{tabular}[t]{l}\scalebox{2}{$\cong$}\end{tabular}}}}%
    \put(0,0){\includegraphics[width=\unitlength,page=4]{App2isomorphismwithcutsRREDUCED.pdf}}%
  \end{picture}%
\endgroup%

}
\caption{
Isomorphisms between certain graphs with identified cuts which are relevant for the operations performed in Figures \ref{fig:App2} and \ref{fig:App21}.
Remark that some are trivial, but we wish to be comprehensive to be clear.
}
\label{fig:App2isomorphismwithcutsREDUCED}
\end{minipage}
\end{figure}

\section{Constructing 2PR graphs of $\ell=2$}
\label{sec:app2PR}

In this Appendix, we illustrate some examples to recursively construct 2PR schemes of $\ell=2$ and $g=1$.
In particular, and these examples shown in Figures \ref{fig:2PRex1app} and \ref{fig:2PRex2app} supplement the proof for Theorem \ref{propo:g1ell2}.

\begin{figure}[H]
\begin{minipage}[t]{1\textwidth}
\centering
\def\svgwidth{1\columnwidth}
\tiny{
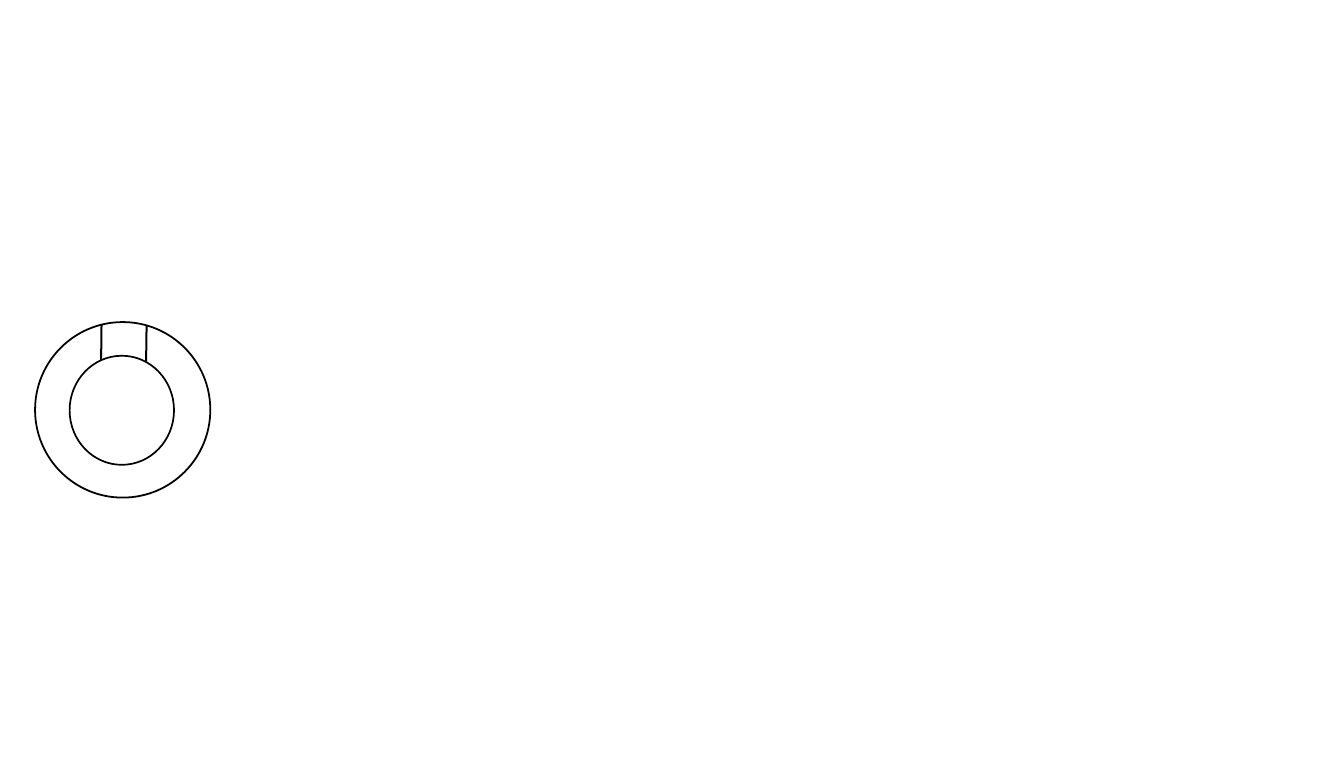
}
\caption{
Some 2PR schemes of $\ell=2$ and $g=1$ obtained from combining a $\ell=2$ and $g = 0$ scheme $S_{G_1}$ with a $\ell=0$ and $g=1$ scheme $S_{G_2}$ via separating dipole or a ladder-vertex
or a two-edge-connection.
These schemes add to Figure \ref{fig:2PRex1}.
$\widetilde {\rm N}_{\rm o} \in \{ $N-dipole$,\, {\rm N}_{\rm o}\}$, 
$Y\in\{ $L-dipole$,\, $R-dipole$,\, {\rm N}_{\rm e}, {\rm L}, {\rm R}, {\rm B}\}$.
}
\label{fig:2PRex1app}
\end{minipage}
\end{figure}

\begin{figure}[H]
\begin{minipage}[t]{1\textwidth}
\centering
\def\svgwidth{1\columnwidth}
\tiny{
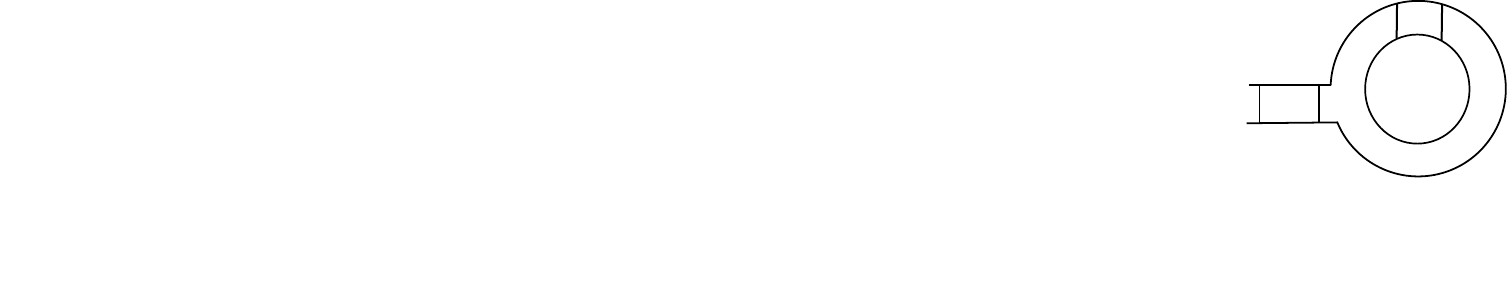
}
\caption{
Some 2PR schemes of $\ell=2$ and $g=1$ obtained from combining a $\ell=1$ and $g = 0$ scheme (``infinity graph") $S_{G_1}$ with a $\ell=1$ and $g=1$ scheme $S_{G_2}$ via separating dipole or a ladder-vertex.  
These schemes add to Figure \ref{fig:2PRex2}.
$ X\in\{$N-dipole$, \,$L-dipole$, \,$R-dipole$,\, {\rm N}_{\rm e}, {\rm N}_{\rm o}, {\rm L}, {\rm R}, {\rm B}\}$.
}
\label{fig:2PRex2app}
\end{minipage}
\end{figure}

\section{Constructing alternating knots corresponding to some $4$-regular planar graphs}
\label{sec:appb}
In this section we give some examples of alternating knots that correspond to 4-regular planar graphs. The knot $3_1$ corresponds to the necklace graphs of 3 vertices. The following graphs $4_1$ corresponds to the second graph in Figure \ref{l45phi1}. The knots $5_1$ and $5_2$ correspond to the necklace graphs of 5 vertices and the forth graph in Figure \ref{l45phi1} respectively. Finally the knots $6_1$, $6_1$ and $6_3$ correspond to the first and third graphs in Figure \ref{facefive} and the third graph of Figure \ref{facefour1} respectively.

We should remark that if we perform a connected sum of two knots of the form $3_1$, then the resulting knot corresponds to second graph in Figure \ref{facesix}.

\begin{figure}[H]
\centering
\includegraphics[width=75mm]
{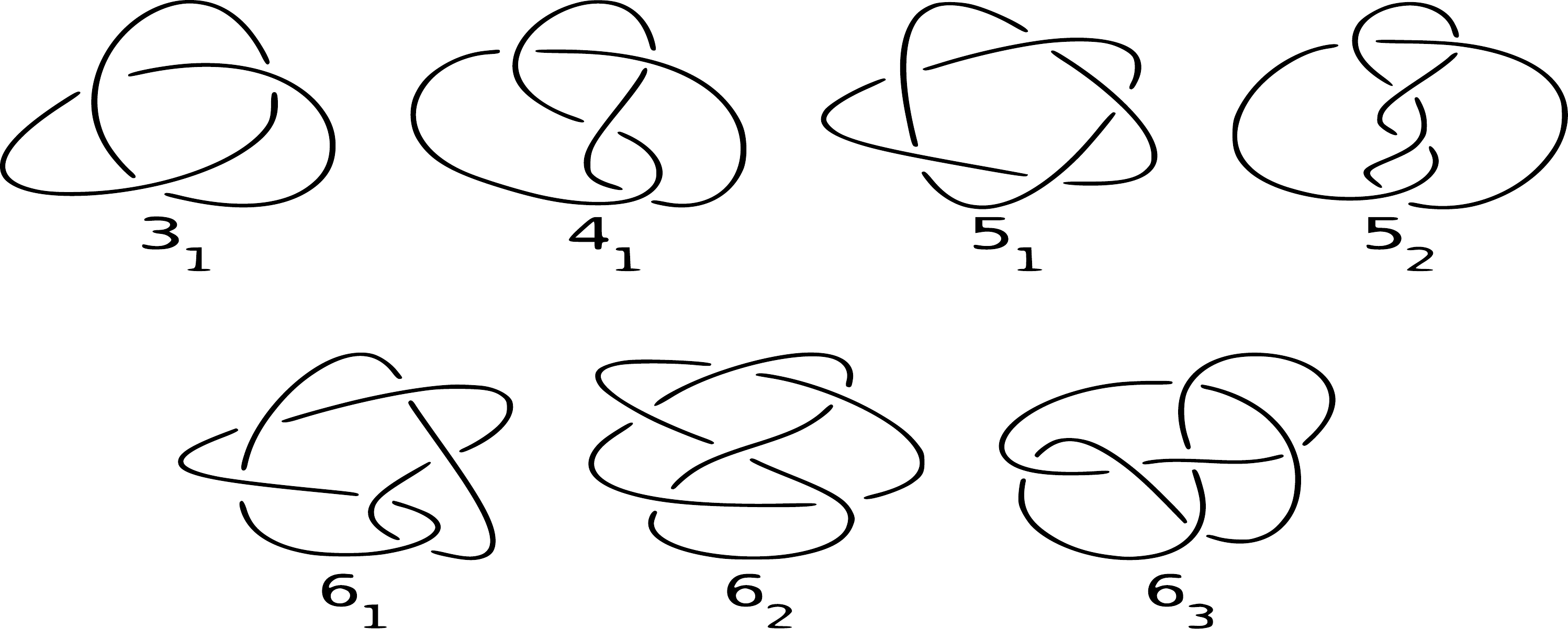}
\caption{Some alternating knots corresponding to $\ell=3,4,5,6$ and $\varphi=1$ graphs.}
\label{fig:conter}
\end{figure}

\section*{Acknowledgments}
The authors would like to thank Andreani Petrou, Andrew Lobb, Dario Benedetti and Sylvain Carrozza for useful discussions.
We would also like to thank the Insitut Henri Poincar\'e thematic program ``Quantum Gravity, Random Geometry and Holography", 2023 where the authors were participating and useful discussions took place for this work.

\bibliographystyle{alpha}
\bibliography{mref.bib}
\end{document}